\documentclass[onecolumn,aps, pra,superscriptaddress,longbibliography,
12pt,tightenlines
]{revtex4-2}

\pdfoutput=1

\usepackage{graphicx}
\usepackage{bm}
\usepackage{amsmath}
\usepackage{environ}
\usepackage{appendix}
\usepackage{physics}

\usepackage[nopar]{lipsum}

\usepackage{amsthm}
\usepackage{amssymb}
\usepackage{amsfonts}

\usepackage{slashed}
\usepackage{bbm}
\usepackage{latexsym,epsfig,bbm}

\usepackage{hyperref}

\hypersetup{colorlinks = true, linkcolor = [rgb]{0.19411,0.51882,0.667058}, urlcolor = [rgb]{0.125490,0.29542,0.1647058}, citecolor = [rgb]{0.75882,0.37411,0.14117}}

\usepackage{color}
\definecolor{blue}{rgb}{0,0.2,1}

\definecolor{red}{rgb}{0.9,0,0}

\newcommand{\vect}[1]{\boldsymbol{#1}}

\newtheorem{theorem}{Theorem}
\newtheorem{lemma}[theorem]{Lemma}
\newtheorem{definition}[theorem]{Definition}
\newtheorem{problem}{Problem}

\newtheorem{remark}{Remark}

\makeatletter
\newcommand{\xboxed}[1]{\mathpalette\xboxed@{#1}}
\newcommand{\xboxed@}[2]{{\fboxsep=1.4pt\boxed{#1#2}}}
\makeatother
\usepackage{mathtools}

\DeclarePairedDelimiter\rbra{\lparen}{\rparen}
\DeclarePairedDelimiter\sbra{\lbrack}{\rbrack}
\DeclarePairedDelimiter\cbra{\{}{\}}
\DeclarePairedDelimiter\Abs{\lVert}{\rVert}
\DeclarePairedDelimiter\ave{\langle}{\rangle}

\newcommand{\poly} {\operatorname{poly}}

\newcommand{\polylog} {\operatorname{polylog}}

\begin{document}

\title{Quantum algorithms for matrix geometric means}

\author{Nana~Liu}
\email{nana.liu@quantumlah.org}
\affiliation{Institute of Natural Sciences, School of Mathematical Sciences, MOE-LSC,
Shanghai Jiao Tong University, Shanghai, 200240, P. R. China}
\affiliation{Shanghai Artificial Intelligence Laboratory, Shanghai, China}
\affiliation{University of Michigan-Shanghai Jiao Tong University Joint Institute, Shanghai 200240, China}

\author{Qisheng Wang}
\email{QishengWang1994@gmail.com}
\affiliation{School of Informatics, University of Edinburgh, EH8 9AB Edinburgh, United Kingdom}
\affiliation{Graduate School of Mathematics, Nagoya University, Nagoya 464-8602, Japan}

\author{Mark~M.~Wilde}
\email{wilde@cornell.edu}
\affiliation{School of Electrical and Computer Engineering, Cornell University, Ithaca, New York 14850, USA}

\author{Zhicheng Zhang}
\email{iszczhang@gmail.com}
\affiliation{Centre for Quantum Software and Information, University of Technology Sydney, Ultimo, NSW 2007, Australia}

\date{\today}

\begin{abstract} 
Matrix geometric means between two positive definite matrices can be defined equivalently from distinct perspectives -- as solutions to certain nonlinear systems of equations, as points along geodesics in Riemannian geometry, and as solutions to certain optimisation problems. 
This diversity already suggests the potential for varied applications, as well as acting as a bridge between different domains. Here we devise new quantum subroutines to efficiently prepare quantum unitary operators that embed the standard matrix geometric mean and its generalisations called the weighted matrix geometric mean. This enables the construction of solutions to the algebraic Riccati equation, which is an important class of nonlinear systems of equations that appears in machine learning, optimal control, estimation, and filtering. 
Using these subroutines,
we present a new class of quantum learning algorithms called \textit{quantum geometric mean metric learning}. This has applications in efficiently finding the best distance measure and solving classification problems in the weakly supervised limit and for anomaly detection, for both classical and quantum problems. We also show how our method can be generalised to a particular $p^{\text{th}}$-order system of nonlinear equations. 
These quantum subroutines for matrix geometric means are also useful in other areas of quantum information. For example, we show how to use them in the estimation of geometric R\'enyi relative entropies and the Uhlmann fidelity by means of the Fuchs--Caves observable.
In particular, our quantum algorithms for estimating the Uhlmann and Matsumoto fidelities have \textit{optimal} dependence on the precision.
Finally, we provide a $\mathsf{BQP}$-complete problem based on matrix geometric means that can be solved by our subroutines, thus characterising their computational capability. 

\end{abstract}

\maketitle 

\newpage

\tableofcontents 

\newpage

\section{Introduction}

Quantum computation is considered a rapidly emerging technology that has important implications for the development of algorithms. Many quantum algorithms that have theoretically demonstrated potential quantum advantage, however, have been chiefly directed towards linear problems -- in part because quantum mechanics is itself linear. These include simulating solutions of linear systems of equations~\cite{harrow2009quantum}, known as quantum linear algebra, and linear ordinary and partial differential equations~\cite{childs2021high, schr2, schr1, an2022theory}.

However, many problems of scientific interest are nonlinear. While most nonlinear systems of equations of interest for applications only appear after discretising nonlinear ordinary and partial differential equations, there is an important class of nonlinear system of equations that is not only relevant to partial differential equations but is also of independent interest. This class consists of the algebraic Riccati equations, which are nonlinear matrix equations with quadratic nonlinearity~\cite{lancaster1995algebraic}. These are also the stationary states of the Riccati matrix differential equations, which are essential for many applications in applied mathematics, science, and engineering problems.  These nonlinear matrix equations are particularly relevant for optimal control, stability theory, filtering (e.g., Kalman filter~\cite{salgado1988connection}), network theory, differential games, and estimation problems~\cite{coppel1974matrix}.

It turns out that solutions to the algebraic Riccati equations are closely connected with the concept of a matrix geometric mean. For example, the unique solution to the simplest algebraic Riccati equation can be precisely expressed as the standard matrix geometric mean, as we will recall later. The matrix geometric means are matrix generalisations of the scalar geometric mean and have a long history in mathematics~\cite{lawson2001,lawson2021expanding};  there are diverse approaches to this same concept. For example, the standard matrix geometric mean can be defined as the output of an optimisation problem. 
The matrix geometric mean between two matrices also has an elegant geometric interpretation as a midpoint along the geodesic joining these two matrices that live in Riemannian space~\cite{lancaster1995algebraic}. The Monge map between two Gaussian distributions, appearing in optimal transport, can also be expressed in terms of the matrix geometric mean~\cite{janati2020entropic}. The standard and weighted matrix geometric means appear in quantum information in the form of quantum entropic~\cite{Matsumoto2018,fang2021GeometricRenyiDivergence} and fidelity~\cite{Fuchs1995,Mat10,cree2020fidelity} measures. 

However, computing the matrix geometric mean involves matrix multiplication and also nonlinear operations like taking inverses and square roots of matrices. Here classical numerical schemes can be inefficient, with costs that are  polynomial in the size of the matrix~\cite{ramesh1989computational}. The processing of several matrix multiplications can, under certain conditions, be more efficient through quantum processing. Our aim here is to construct quantum subroutines that embed the standard and weighted matrix geometric means into unitary operators and to determine the conditions under which these embeddings can indeed be conducted efficiently. There are many such possible unitary operators, and we choose a formalism called block-encoding~\cite{GSLW19, CGJ19, low2019hamiltonian}.

The block-encoding of a non-unitary matrix $Y$ is a unitary matrix $U_Y$ whose upper left-hand corner is proportional to $Y$. The construction of this unitary matrix allows realisation by means of a quantum circuit, which describes unitary evolution. The matrix $Y$ can be subsequently recovered by extracting only the top-left corner through measurement. This provides a convenient building block for constructing sums and (integer and non-integer) powers of matrices $Y$ by concatenating its block-encodings via unitary circuits. This formalism allows us to form the block-encoding of the standard and weighted matrix geometric means,  which are products of matrices and their roots. 
From these block-encodings, one can also recover their expectation values with respect to certain states. These different expectation values are then relevant for various applications, like in machine learning and quantum fidelity estimation. 

Under certain assumptions, we show how these can be efficiently implementable on quantum devices. This efficiency arises from the fact that matrix multiplications can be more efficient with quantum algorithms. 
This observation has an important consequence. It means that a quantum device can efficiently prepare solutions of the (nonlinear) algebraic Riccati equations. The expectation values of these solutions can also be shown to be efficiently recoverable for different applications. Our approach differs from many past works in three key respects: (a) ours is the first quantum subroutine, to the best of our knowledge, to prepare solutions of nonlinear matrix equations without using iterative methods. The solutions themselves are matrices and not vectors, which differs from other quantum algorithms for nonlinear systems of equations, for example~\cite{jin2023time,joseph2020koopman,liu2021efficient}; (b) the solutions are not embedded in a pure quantum state, but rather an observable, thus introducing a novel embedding of the solution. This is important when solutions themselves are in matrix form (for matrix equations), which differs from the quantum embeddings of solutions of discretised nonlinear ordinary and partial differential equations (solutions not in matrix form)~\cite{jin2023time,joseph2020koopman,liu2021efficient}; (c) we show the efficient recovery of outputs for nonlinear systems of equations directly relevant for applications.

One class of applications is in the area of machine learning. Machine learning algorithms often require an assignment of a metric, or distance measure, in order to compute distances between data points. The values of these distances then become central to the outcome, for instance, in making a prediction for classification. This means that the choice of the metric itself is important, but the best metric can depend on the actual data. Learning the metric from given data -- called metric learning -- can also be formulated as a learning problem. While most of these metric learning algorithms require iterative techniques like gradient descent to minimise the proposed loss function, a class of metric learning algorithms called geometric mean metric learning~\cite{zadeh2016geometric} admits closed-form solutions. It has also been shown to attain higher classification accuracy with greater speed than previous methods. Here we devise efficient quantum algorithms, using our quantum subroutine for the matrix geometric mean, for geometric mean metric learning for both classical and quantum data. For quantum data, we propose new algorithms that can be used for the anomaly detection of quantum states, which differs from previous algorithms~\cite{liu2018quantum}. The applicability extends also to asymmetric cases for which there is a higher cost to be paid for false negatives or true positives. This is in fact related to the weighted matrix geometric mean. 

There is also an important connection between the solution of the geometric mean metric learning problem and the Fuchs--Caves observable~\cite{Fuchs1995}, which appears in quantum fidelity estimation. This allows for a re-derivation of quantum fidelity from the point of view of machine learning. We show that our quantum subroutines for the matrix geometric mean can also be used in the efficient estimation of geometric R\'enyi relative entropies and the quantum fidelity by means of the Fuchs--Caves observable. 
This new way of estimating the quantum fidelity has polynomially better performance in precision than previously known fidelity estimation algorithms. It is also shown to be optimal with respect to precision.

We can also extend our method to a more general class of nonlinear systems of equations of $p^{\text{th}}$-degree. These are $p^{\text{th}}$-degree polynomial generalisations of the simplest algebraic Riccati equations. 
We show that the unique solutions of these equations are  weighted matrix geometric means. We similarly devise quantum subroutines to prepare their block-encodings. The weighted matrix geometric mean for two quantum states has an elegant geometric interpretation as the state at $(1/p)^{\text{th}}$ of the length along the geodesic connecting two quantum states in Riemannian space. We also show these are relevant to the weighted version of our new quantum learning algorithm. Furthermore, preparing block-encodings of the weighted matrix geometric means allows us to construct, to the best of our knowledge, the first quantum algorithm for estimating the geometric R\'enyi relative entropies. 

\subsection{Summary of our results}

For convenience, we provide a brief summary of our results here.
Our first contribution consists of basic quantum subroutines in Section~\ref{sec:main} for matrix geometric means (see Definition~\ref{def:mgm}) and their weighted generalisation (see Definition~\ref{def:weightedgeometrimean}).

\textbf{Solving algebraic Riccati equations}.
We then consider the problem of solving the algebraic matrix Riccati equation
\begin{equation} \label{eq:Riccati}
    YAY - B^\dag Y - Y^\dag B - C = 0,
\end{equation}
where $A$, $B$, and $C$ are $d \times d$ complex-valued matrices.
We delineate quantum algorithms with time complexity $O\rbra*{\poly \log d}$ for solving Eq.~\eqref{eq:Riccati} for well-conditioned matrices, in Section~\ref{sec:B=0-Riccati} and Section~\ref{sec:B-neq-0-Riccati}.
Here, we say a matrix $A$ is well-conditioned if $A \geq I/(\poly \log d)$.
The higher-order case $Y\rbra*{AY}^{p-1} = C$ is studied in Section~\ref{sec:high-order-Riccati}. 
In Section~\ref{sec:BQP}, we show that it is $\mathsf{BQP}$-complete to solve the equation $YAY = C$, a special case of Eq.~\eqref{eq:Riccati}, in which case the solution is $Y = A^{-1}\#C$ (see Definition~\ref{def:mgm} for the meaning of this notation).

\textbf{Geometric mean metric learning}. We introduce quantum algorithms for learning the metric in machine learning, by phrasing this as an optimisation problem using a geometric perspective. Unlike other metric learning algorithms, this optimisation problem has a closed-form solution. This follows the geometric mean metric learning method~\cite{zadeh2016geometric}. The solution turns out to be expressible in terms of the matrix geometric mean $Y = A^{-1}\#C$. We design quantum algorithms for the learning task for classical data (Section~\ref{sec:learning-c}) as well as for quantum data (Section~\ref{sec:learning-q}). We present the conditions under which the quantum algorithm is more efficient than the corresponding classical algorithm. For example, the classical learning task with well-conditioned matrices $A$ and $C$ has time complexity $O(\poly (\log d, \log (1/\epsilon)))$. 
We also show that the quantum learning task with well-conditioned quantum states $\rho$ and $\sigma$ has time complexity $O(\poly (\log d, \log (1/\epsilon)))$.
The latter learning task for quantum data is uniquely quantum in nature and has no classical counterpart.

\textbf{(Uhlmann) fidelity estimation}. Based on the Fuchs--Caves observable~\cite{Fuchs1995}, we design a new quantum algorithm for fidelity estimation in Section~\ref{sec:fidelity} via the fidelity formula $F\rbra*{\rho, \sigma} = \Tr\rbra*{\rbra*{\sigma^{-1}\#\rho}\sigma}$, which involves the matrix geometric mean.
We show that our quantum algorithm has 
query complexity $\widetilde{O}\rbra{\kappa^4/\epsilon}$ provided that $\rho, \sigma \geq I/\kappa$ for some known $\kappa > 0$, and that the $\epsilon$-dependence is optimal up to polylogarithmic factors. 

\textbf{Geometric R\'enyi relative entropy}.
In Section~\ref{sec:renyi}, we present the first quantum algorithm for computing the geometric R\'enyi relative entropy, to the best of our knowledge.
In particular, we design a quantum algorithm for computing the geometric fidelity $\widehat{F}_{1/2}\rbra*{\rho, \sigma} \coloneqq  \Tr\rbra*{\rho \# \sigma}$ (also known as the Matsumoto fidelity \cite{Mat10,cree2020fidelity}) with 
query complexity $\widetilde{O}\rbra{\kappa^{3.5}/\epsilon}$ provided that $\rho, \sigma \geq I/\kappa$ for some known $\kappa > 0$,
and we prove that the $\epsilon$-dependence is optimal up to polylogarithmic factors.

\textbf{Organisation of this paper}.
In Section~\ref{sec:background}, we begin with a review of the standard matrix geometric mean, weighted matrix geometric mean, the algebraic Riccati equation, and block-encoding. 
In Section~\ref{sec:main} we compute the costs required to prepare block-encodings of the solutions of algebraic Riccati equations and their $p^{\text{th}}$-order generalisations. Applications are presented Section~\ref{sec:applications}. In Section~\ref{sec:BQP} we show how our new quantum subroutines for the matrix geometric mean can solve a $\mathsf{BQP}$-complete problem. We end in Section~\ref{sec:discussion} with discussions.

\section{Background} \label{sec:background}

In this section we give a brief overview of the standard and weighted matrix geometric means and their role in solving algebraic Riccati equations (see \cite[Chapters~4 \& 6]{bhatia2009positive} and \cite{lawson2001,lawson2021expanding} for more details). We then provide a definition of block-encoding. Throughout the paper, unless otherwise stated, we deal with Hermitian matrices. 

\subsection{Matrix geometric means} 
\label{sec:introgeometricmeans}

\begin{definition}[Matrix geometric mean] \label{def:mgm}
Fix $D \in \mathbb{N}$. Given two $D \times D$ positive definite matrices $A$ and $C$, the matrix geometric mean of $A$ and $C$ is defined as
\begin{align}
    \label{eq:geometricmeans}
    A \# C \coloneqq  A^{1/2}(A^{-1/2}CA^{-1/2})^{1/2}A^{1/2}>0.
\end{align}
Note that the matrix geometric mean between $A^{-1}$ and $C$ is thus defined by
\begin{align} \label{eq:inversegeometricmeans}
A^{-1} \# C =  A^{-1/2}(A^{1/2}CA^{1/2})^{1/2}A^{-1/2}>0.
\end{align}
Alternatively, the matrix geometric mean $A \# C$ can be equivalently be written as 
\begin{align}
    A \# C =\max \Biggl\{ Y \geq 0:  \begin{pmatrix} 
    A & Y \\
    Y & C
    \end{pmatrix} \geq 0 \Biggr\},
\end{align}
where the ordering of Hermitian matrices is given by the L\"owner partial order.
\end{definition}

The matrix geometric mean appears in quantum information, for example, like the Fuchs--Caves observable~\cite{Fuchs1995}, in quantum fidelity and entropy operators like the Tsallis relative operator entropy~\cite{furuichi2004fundamental}, and quantum fidelity measures  between states~\cite{Mat10,Matsumoto2018,cree2020fidelity} and channels~\cite{katariya2021geometric}. This concept can also be generalised to the weighted matrix geometric mean.

\begin{definition}[Weighted matrix geometric mean] \label{def:weightedgeometrimean}
Fix $p>0$. The weighted matrix geometric mean 
with weight~$1/p$ is defined as 
\begin{align}
    A \#_{1/p} \, C \coloneqq  A^{1/2}(A^{-1/2}CA^{-1/2})^{1/p}A^{1/2}.
\end{align}
The weighted matrix geometric mean between $A^{-1}$ and $C$ is then equal to
\begin{align} \label{eq:weightedinversegeometricmeans}
    A^{-1} \#_{1/p} \, C =  A^{-1/2}(A^{1/2}CA^{1/2})^{1/p}A^{-1/2}.
\end{align}
\end{definition}

The canonical matrix geometric mean corresponds to the weighted geometric mean with weight $1/p=1/2$.

We will use the definitions in Eqs.~\eqref{eq:inversegeometricmeans} and~\eqref{eq:weightedinversegeometricmeans} here and throughout because, as we will see later on, they are relevant to solutions of classes of nonlinear matrix equations like the algebraic Riccati equations. \\

For positive definite matrices (which include full-rank density matrices), the standard and weighted matrix geometric means have elegant geometric interpretations. 
It is known that the inner product on the real vector space formed by the set of Hermitian matrices gives rise to a Riemannian metric~\cite[Chapter~6]{bhatia2009positive}. This Riemannian metric is defined on the manifold $M_H$ formed by the set of positive definite matrices. Following~\cite[Eqs.~(6.2) \& (6.4)]{bhatia2009positive}, a trajectory $\gamma: [a, b] \rightarrow M_H$ on this manifold is a piecewise differential path on $M_H$ whose length is defined by $L(\gamma)\coloneqq \int_a^b \left\|\gamma^{-1/2}(t)\gamma'(t)\gamma^{-1/2}(t)\right\|_2\, dt$. Then the distance $\delta(A^{-1}, C)=\inf_{\gamma} L(\gamma)$ between any two positive definite  matrices $A^{-1}$ and $C$ on this manifold is defined to be shortest length joining these two points. Then we have the following result.

\begin{lemma} [{\cite[Theorem 6.1.6]{bhatia2009positive}}]
If $A^{-1}$ and $C$ are  two positive definite matrices, then there exists a unique geodesic joining $A^{-1}$ and $C$. This geodesic has the following parameterisation with $t \in [0,1]$:
\begin{align}
    \gamma_{\operatorname{geod}}(t)=A^{-1/2}(A^{1/2}CA^{1/2})^{1/t}A^{-1/2}, \qquad t \in [0,1].
\end{align}
This geodesic has length given by
\begin{align} \label{eq:geodesicdelta}
    \delta(A^{-1}, C)=L(\gamma_{\operatorname{geod}})=\left\|\log (A^{1/2} C A^{1/2}) \right\|_2.
\end{align}
\end{lemma}
In the above, $\|X \|_2 \coloneqq \sqrt{\operatorname{Tr}[X^\dag X]}$ denotes the Schatten $2$-norm, whereas  $\| \cdot \|$ refers to operator norm throughout our paper. 

From this viewpoint,  the matrix geometric mean $A^{-1} \# C=\gamma_{\operatorname{geod}}(t=1/2)$ can clearly be interpreted as the midpoint along the geodesic joining $A^{-1}$ and $C$. Similarly, the weighted geometric mean with weight $1/p$ can be interpreted as the point along the manifold when~$t=1/p$.

\subsection{Algebraic Riccati equations}

Let us begin with a general form of the algebraic Riccati equation for the unknown $D \times D$ matrix $Y$:
\begin{align} \label{eq:yay1}
    Y^\dag A Y - B^{\dagger} Y-Y^{\dagger}B-C=0,
\end{align}
where $A$, $B$, and  $C$ are $D \times D$ matrices with complex-valued entries. This can be understood as a matrix version of the famous (scalar) quadratic equation $ay^2 - 2by - c = 0$. Solutions of equations like \eqref{eq:yay1} are not always guaranteed to exist, and certain conditions are required to prove the existence of, for instance, Hermitian solutions~\cite{ran1988existence}. See~\cite{shurbet1974quadratic} for conditions on solvability. Even if existence can be shown, the solutions may not be unique or could alternatively be uncountably many~\cite{richardson1986positive, lancaster1980existence, wimmer1984algebraic}. However, there are unique solutions under certain conditions. For instance, if all the matrix entries are real-valued, then for symmetric positive semidefinite $A, C$ and symmetric positive $Y$, there is a unique positive definite solution if and only if an associated matrix $H=\begin{pmatrix} -B & A \\
C & B^T \end{pmatrix}$ has no imaginary eigenvalues~\cite{boyd1991linear}. 

In this paper, we confine our attention to simpler cases, for example in Lemmas~\ref{lem:yay1} and \ref{lem:yay2}, when there are unique solutions.
\begin{lemma}[Solution of simple algebraic Riccati equation]
\label{lem:yay1}
    Consider the following algebraic Riccati equation when $A$ and $C$ are positive definite matrices and $Y$ is Hermitian:
\begin{align} \label{eq:yayc}
    YAY=C.
\end{align}
This equation has a unique positive definite solution  given by the standard matrix geometric mean:
\begin{align} \label{eq:yaysimplesol}
    Y= A^{-1} \# C \coloneqq  A^{-1/2}(A^{1/2}CA^{1/2})^{1/2}A^{-1/2} >0. 
\end{align}
\end{lemma}

\begin{proof}
This lemma is well known from~\cite{kubo1980means, nakamura2007geometric}, but we provide a brief proof for completeness. Starting from the Riccati equation in \eqref{eq:yayc} and by using the fact that $A$ is positive definite with a unique square root, consider that
\begin{align}
    YAY=C  \quad & \Leftrightarrow \quad  YA^{1/2}A^{1/2} Y=C \\
    &  \Leftrightarrow \quad  A^{1/2}YA^{1/2}A^{1/2} YA^{1/2}=A^{1/2}CA^{1/2} \\
    &  \Leftrightarrow \quad  (A^{1/2}YA^{1/2})^2 = A^{1/2}CA^{1/2} .
\end{align}
Since the matrix $A^{1/2}CA^{1/2} $ is positive definite and the equality in the last line above has been shown, both $A^{1/2}CA^{1/2}$ and $(A^{1/2}YA^{1/2})^2$ have a unique positive definite square root, implying that
\begin{align}
A^{1/2}YA^{1/2} = (A^{1/2}CA^{1/2})^{1/2} \quad & \Leftrightarrow \quad Y = A^{-1/2}(A^{1/2}CA^{1/2})^{1/2}A^{-1/2},
\end{align}
thus justifying that $Y= A^{-1} \# C$ is the unique positive definite solution as claimed.
\end{proof}

See \cite{pedersen1972operator} for a discussion of \eqref{eq:yayc} in the infinite-dimensional case.

If $A$ and $C$ are both  positive definite with unit trace $\operatorname{Tr}(A)=1=\operatorname{Tr}(C)$, then $A$ and $C$ can also be interpreted as density matrices. Then the operator $A^{-1} \# C$ is also known as the Fuchs--Caves observable~\cite{wilde2017quantum}, which is of relevance in the study of quantum fidelity. We will return to this point later.  See also \cite[Section~V]{alsing2024geodesics} for an interpretation of \eqref{eq:yayc} when $A$ and $C$ are density matrices.

We can also extend Lemma~\ref{lem:yay1} to the $B \neq 0$ case, and the following holds.

\begin{lemma} \label{lem:yay2} 
If $A$ and $C$ are positive definite, $B$ is an arbitrary matrix, and  $(A^{-1}B)=(A^{-1}B)^{\dagger}$, then a Hermitian solution to Eq.~\eqref{eq:yay1} can be expressed as
\begin{align} \label{eq:yay1sol}
    Y= A^{-1} \# (B^{\dagger} A^{-1} B+C)  +A^{-1} B.
\end{align}
\end{lemma}
\begin{proof}
    See Appendix~\ref{app:proof-yay2}. 
\end{proof}

Classical algorithms for solving algebraic Riccati equations are typically inefficient~\cite{ramesh1989computational} with respect to the size of the problem, i.e., polynomial in $D$. We will be looking at conditions for which a quantum algorithm for solving algebraic Riccati equations can be executed with less complexity.

\subsection{Block-encoding}

Classical information can be embedded in quantum systems in the form of quantum states, either pure or mixed, or in the form of quantum processes. A closed quantum system evolves under a unitary transformation,  represented by a unitary matrix. In this paper, we will be focusing on how a matrix solution to a matrix equation can be embedded in a unitary matrix. Unlike other quantum subroutines that prepare solutions of a linear system of equations embedded in the amplitudes of a pure quantum state, here we first embed the solution $Y$ into a unitary matrix. 

There are different ways of embedding an arbitrary matrix into a unitary matrix. For instance, it is guaranteed by the Sz.~Nagy dilation theorem (see, e.g., \cite[Theorem~1.1]{paulsen2002completely}) that such a unitary matrix should always exist. We choose a flexible dilation known as block-encoding~\cite{GSLW19, CGJ19, low2019hamiltonian}. A unitary matrix $U_Y$ is called a block-encoding of a matrix $Y$ if it satisfies the following definition.

\begin{definition}[Block-encoding]
        \label{def:blk-enc}
	Fix $n,a\in\mathbb{N}$ and $\epsilon, \alpha \geq 0$. Let $Y$ be an $n$-qubit operator.
	An $(n+a)$-qubit unitary $U_Y$ is an $(\alpha,a,\epsilon)$-block-encoding of an operator $Y$ if
	\begin{equation}
		\left\|Y-\alpha \bra{0}_a U_Y\ket{0}_a\right\| \leq \epsilon.
	\end{equation}
\end{definition}

Here $|0\rangle_a$ are the $|0\rangle$ states in the computational basis of the $a$-ancilla qubits.  
The block-encoding formalism allows one to construct, for example, block-encodings of sums of matrices, linear combinations of block-encoded matrices, and polynomial approximations of negative and positive power functions of matrices \cite{GSLW19}. We list several associated lemmas in Appendix~\ref{app:preliminaries} for convenience. 

\section{Quantum subroutines for matrix geometric means, algebraic Riccati equations, and higher-order nonlinear equations}

\label{sec:main}

Let us focus on cases where the solutions to the algebraic Riccati equations can be captured by the matrix geometric mean in Lemmas~\ref{lem:yay1} and~\ref{lem:yay2}. The computation of the matrix geometric mean involves the computation of the square roots of matrices and several matrix multiplications. For $D \times D$ matrices, typically these costs will scale polynomially with~$D$ for a classical algorithm. However, quantum algorithms for matrix multiplications of block-encoded matrices can be performed more efficiently when compared to the number of classical numerical steps. These series of block-encoded matrix multiplications can be achieved in the quantum case via the block-encoding formalism. 

Let us begin with the algebraic Riccati equation in Eq.~\eqref{eq:yay1}:
\begin{align} 
    Y A Y - B^{\dagger} Y-Y^{\dagger}B-C=0. 
\end{align}
It is our goal below  first to construct a block-encoding of the solution $Y$, denoted $U_Y$, under the conditions obeyed in Lemmas~\ref{lem:yay1} and~\ref{lem:yay2}. This we consider as a subroutine that we can then employ in various applications. 

Below we assume that we also have access to the block-encodings of $A$, $B$, $C$ -- denoted $U_A$, $U_B$, $U_C$, respectively -- as well as their inverses $U_A^\dag$, $U_B^\dag$, $U_C^\dag$ and controlled versions. For example, if $A$, $B$, and $ C$ are positive semi-definite matrices with unit trace, these can be considered as density matrices. Then from Lemma~\ref{lmm:purified to block-encoding} we can prepare block-encodings $U_{A}$, $U_B$, $U_{C}$ by access to the unitaries that prepare purifications of $A$, $B$, and $C$, with only a single query to each purification and $O(\log d)$ gates. In more general scenarios, we can leave the preparation of these block-encodings to a later stage, which also depends on the particular application. 
Below, $\kappa_A$ and $\kappa_C$ denote the condition numbers for $A$ and $C$, respectively. 
It is important to clarify that, as assumed in \cite{harrow2009quantum}, all of our quantum algorithms for matrix geometric means assume that
\begin{align}
    I & \geq A \geq I/\kappa_A,\\
    I & \geq C \geq I / \kappa_C, \\
    I & \geq B^\dag B.
\end{align}
This means that $\kappa_A$ and $\kappa_C$ are really equal to the inverses of the minimum eigenvalues of $A$ and $C$, respectively, and $\left \|A\right\|, \left \|B\right\|, \left \|C\right\| \leq 1$. The upper bounds above are automatically satisfied whenever $A$ and $C$ are density matrices.

\subsection{Quantum subroutine for matrix geometric means}

As a warm-up, we present a quantum subroutine for implementing block-encodings of the weighted matrix geometric means.

\begin{lemma}[Block-encoding of weighted matrix geometric mean]
\label{lem:mgm-1}
    Suppose that $U_A,U_C$ are $(1,a,0)$-block-encodings of matrices $A,C$, respectively,
    where $A\geq I/\kappa_A$, $C\geq I/\kappa_C$ and $I$ is the identity matrix.
    For $\epsilon\in (0,1/2)$,
    one can implement a $(2\kappa_A^{1/p}\gamma_p,5a+12,\epsilon)$-block-encoding of $Y$
    for every fixed real $p\neq 0$,
    where 
    \begin{equation}
        \gamma_p=\begin{cases}
            1 & p>0,\\
            \kappa_A^{-1/p}\kappa_C^{-1/p}& p<0, 
        \end{cases}
    \end{equation}
    and
    \begin{equation}
        Y=A\#_{1/p} C=A^{1/2}\rbra*{A^{-1/2}CA^{-1/2}}^{1/p}A^{1/2},
    \end{equation}
    using
    \begin{itemize}
        \item 
        $\widetilde{O}\rbra*{\kappa_A\kappa_C\log^3\rbra*{1/\epsilon}}$ queries to $U_C$,
        $\widetilde{O}\rbra*{\kappa_A^2\kappa_C\log^4\rbra*{1/\epsilon}}$ queries to $U_A$;
        \item
        $\widetilde{O}\rbra*{a\kappa_A^2\kappa_C\log^4\rbra*{1/\epsilon}}$ gates; and
        \item 
        $\poly\rbra*{\kappa_A,\kappa_C,\log\rbra*{1/\epsilon}}$ classical time.
    \end{itemize}
\end{lemma}

\begin{remark}
In the above and in what follows, `queries to $U$' refers to access not only to $U$, but also to its inverse $U^{\dagger}$, controlled-$U$, and controlled-$U^\dag$. Here and in the following, $\widetilde{O}(\cdot)$ suppresses logarithmic factors of functions appearing in $(\cdot)$.
The same convention applies to $\widetilde{\Omega}(\cdot)$ and $\widetilde{\Theta}(\cdot)$.
\end{remark}

\begin{proof}[Proof sketch of Lemma~\ref{lem:mgm-1}]
    See Appendix~\ref{app:warmup-mgm} for a detailed proof. 
    As an illustration for the construction of our quantum subroutines, we outline the basic idea. Other quantum subroutines later presented in this section are obtained using similar ideas. 
    Our approach consists of three main steps:
    \begin{enumerate}
        \item Implement a block-encoding of $A^{-1/2}$, using roughly $\widetilde O\rbra*{\kappa_A}$ queries to a block-encoding of $A$ (for simplicity, we ignore the $\epsilon$-dependence in our brief explanation here). This is done by applying quantum singular value transformation (QSVT) \cite{GSLW19} with polynomial approximations of negative power functions (see Lemma~\ref{lmm:poly negative power}).
        \item Implement a block-encoding of $\rbra*{A^{-1/2} C A^{-1/2}}^{1/p}$, using roughly $\widetilde O\rbra*{\kappa_A \kappa_C}$ queries to a block-encoding of $A^{-1/2} C A^{-1/2}$.
        This is done by applying QSVT with polynomial approximations of positive power functions (see Lemma~\ref{lmm:poly positive power}). 
        Note that a block-encoding of $A^{-1/2} C A^{-1/2}$ can be implemented using $O\rbra*{1}$ queries to block-encodings of $A^{-1/2}$ 
        and $C$ by the method for realising the product of block-encoded matrices (see Lemma~\ref{lmm:product of block-encoding}).
        \item Similar to Step 2, implement a block-encoding of $A^{1/2} \rbra*{A^{-1/2} C A^{-1/2}}^{1/p} A^{1/2}$, using $O\rbra*{1}$ queries to block-encodings of $A^{1/2}$ and $\rbra*{A^{-1/2} C A^{-1/2}}^{1/p}$, where a block-encoding of $A^{1/2}$ can be implemented using $\widetilde O\rbra*{\kappa_A}$ queries to a block-encoding of $A$.
    \end{enumerate}

    To conclude, the overall query complexity is roughly $\widetilde O\rbra*{\kappa_A} \cdot \widetilde O\rbra*{\kappa_A \kappa_C} + \widetilde O\rbra*{\kappa_A} = \widetilde O\rbra*{\kappa_A^2 \kappa_C}$.
    Note that the construction is mainly based on QSVT and thus is also time efficient. 
    So the overall time complexity is equal to the query complexity only up to polylogarithmic factors.
\end{proof}

\subsection{\texorpdfstring{$B=0$}{B = 0} algebraic Riccati equation} \label{sec:B=0-Riccati}

Let us begin with the unique positive definite solution to the algebraic Riccati equation with $B=0$, i.e., Eq.~\eqref{eq:yay1}, which can be expressed as the matrix geometric mean $Y=A^{-1} \# C$, according to Lemma~\ref{lem:yay1}, where $A$ and $C$ are positive definite matrices. Then we have the following lemma, which characterises a block-encoding of the solution in a quantum circuit.

\begin{lemma}
\label{lem:simplericcati}
    Suppose that $U_A,U_C$ are $(1,a,0)$-block-encodings of matrices $A,C$, respectively,
    with $A\geq I/\kappa_A$ and $C\geq I/\kappa_C$.
    For $\epsilon\in (0,1/2)$,
    one can implement a $(2\kappa_A,5a+11,\epsilon)$-block-encoding of $Y$,
    where
    \begin{equation}
        Y=A^{-1}\# C=A^{-1/2}\rbra*{A^{1/2}CA^{1/2}}^{1/2}A^{-1/2},
    \end{equation}
    using
    \begin{itemize}
        \item 
        $\widetilde{O}\rbra*{\kappa_A\kappa_C\log^2\rbra*{1/\epsilon}}$ queries to $U_C$ and
        $\widetilde{O}\rbra*{\kappa_A^2\kappa_C\log^3\rbra*{1/\epsilon}}$ queries to $U_A$;
        \item
        $\widetilde{O}\rbra*{a\kappa_A^2\kappa_C\log^3\rbra*{1/\epsilon}}$ gates; and
        \item 
        $\poly\rbra*{\kappa_A,\kappa_C,\log\rbra*{1/\epsilon}}$ classical time.
    \end{itemize}
\end{lemma}
\begin{proof}
See Appendix~\ref{app:proofb=0}.
\end{proof}

\subsection{\texorpdfstring{$B \neq 0$}{B not equal to 0} algebraic Riccati equation}
\label{sec:B-neq-0-Riccati}

Here we want to construct a block-encoding of a Hermitian solution to the algebraic Riccati equation via the standard matrix geometric mean, according to Lemma~\ref{lem:yay2}. We then have the following lemma.

\begin{lemma} 
    \label{lem:Briccati}
    Suppose that $U_A,U_B,U_C$ are $(1,a,0)$-block-encodings of matrices $A,B,C$, respectively,
    with $A\geq I/\kappa_A$, $C\geq I/\kappa_C$ and
    $A^{-1}B=\rbra*{A^{-1}B}^\dagger$.
    For $\epsilon\in (0,1/2)$,
    one can implement a $(2\kappa_A^{3/2},b,\epsilon)$-block-encoding of $Y$,
    where $b=O\rbra*{a+\log\rbra*{\kappa_A\kappa_C/\epsilon}}$ and
    \begin{align}
        Y&=A^{-1}\# \rbra*{B^\dagger A^{-1}B +C}+A^{-1}B\\
        &=A^{-1/2}\rbra*{A^{1/2}\rbra*{B^\dagger A^{-1} B+C}A^{1/2}}^{1/2}A^{-1/2}+A^{-1}B,
    \end{align}
    using
    \begin{itemize}
        \item 
        $\widetilde O\rbra*{\kappa_A\kappa_C\log^2\rbra*{1/\epsilon}}$ queries to $U_B$ and $U_C$, and
        $\widetilde{O}\rbra*{\kappa_A^2\kappa_C\log^3\rbra*{1/\epsilon}}$ queries to $U_A$;
        \item
        $\widetilde{O}\rbra*{a\kappa_A^2\kappa_C\log^3\rbra*{1/\epsilon}}$ gates; and
        \item 
        $\poly\rbra*{\kappa_A,\kappa_C,\log\rbra*{1/\epsilon}}$ classical time.
    \end{itemize}
\end{lemma}
\begin{proof}
    See Appendix~\ref{app:proofbneq0}.
\end{proof}

\subsection{Higher-order polynomial equations}
\label{sec:high-order-Riccati}

Algebraic Riccati equations are second-order nonlinear equations whose solutions are given by the second-order matrix geometric mean, i.e., $p=2$. We can also generalise our formalism to particular $p^{\text{th}}$-order nonlinear matrix equations, whose solutions involve $p\in \{3,4,\ldots \}$ weighted matrix geometric mean. For example, consider the following $p^{\text{th}}$-order nonlinear matrix equations, which we call the \textit{$p^{\text{th}}$-order YAY algebraic equations}
\begin{align} \label{eq:yayp}
    Y (AY)^{p-1}=C,
\end{align}
where $p$ is the highest order polynomial in $Y$. It is straightforwardly checked that the solutions can be written in terms of the weighted geometric mean from Definition~\ref{def:weightedgeometrimean}:
\begin{align} \label{eq:yaypsol}
    Y=A^{-1/2}(A^{1/2}C A^{1/2})^{1/p}A^{-1/2} = A^{-1} \#_{1/p} C. 
\end{align}
See \cite{FURUTA1988149} for a discussion of this kind of equation in the infinite-dimensional case.

\begin{lemma}[Solution of simple $p^{\text{th}}$-order algebraic $YAY$ equation]
\label{lem:yay-pth-order}
    Fix $p\in \{2,3,4, \ldots \}$. Consider the $p^{\text{th}}$-order algebraic $YAY$ equation when $A$ and $C$ are positive definite matrices:
\begin{align} 
    Y (AY)^{p-1}=C.
    \label{eq:yayp-lemma}
\end{align}
This equation has a unique positive definite solution  given by the following weighted geometric mean:
\begin{align}
    Y=A^{-1/2}(A^{1/2}C A^{1/2})^{1/p}A^{-1/2} = A^{-1} \#_{1/p} C >0. 
\end{align}
\end{lemma}

\begin{proof}
The proof is a generalisation of that for Lemma~\ref{lem:yay1} and we provide it for completeness.
Starting from the equation in \eqref{eq:yayp-lemma} and by using the fact that $A$ is positive definite with a unique square root, consider that
\begin{align}
    Y (AY)^{p-1} = C  \quad & \Leftrightarrow \quad  Y (A^{1/2}A^{1/2}Y)^{p-1} =C \\
    &  \Leftrightarrow \quad  (A^{1/2}YA^{1/2})^p = A^{1/2}CA^{1/2} ,
\end{align}
where the last line is obtained by left and right multiplying the previous line by $A^{1/2}$.
Since the matrix $A^{1/2}CA^{1/2} $ is positive definite and the equality in the last line above has been shown, both $A^{1/2}CA^{1/2}$ and $(A^{1/2}YA^{1/2})^p$ have a unique positive definite $p$th root, implying that
\begin{align}
A^{1/2}YA^{1/2} = (A^{1/2}CA^{1/2})^{1/p} \quad & \Leftrightarrow \quad Y = A^{-1/2}(A^{1/2}CA^{1/2})^{1/p}A^{-1/2},
\end{align}
thus justifying that $Y= A^{-1} \#_{1/p} C$ is the unique positive definite solution as claimed.
\end{proof}

To construct a block-encoding for the weighted geometric mean (and thus the solution of~\eqref{eq:yayp}), we have proven the following lemma, which holds for every non-zero real number~$p$.

\begin{lemma}
    \label{lem:high-order-riccati}
    Suppose that $U_A,U_C$ are $(1,a,0)$-block-encodings of matrices $A,C$, respectively,
    with $A\geq I/\kappa_A$ and $C\geq I/\kappa_C$ and let $p\neq 0$ be any fixed non-zero real number. 
    For $\epsilon\in (0,1/2)$,
    one can implement a $(2\kappa_A\gamma_p,5a+11,\epsilon)$-block-encoding of $Y$,
    where
    \begin{equation}
        Y=A^{-1}\#_{1/p} C=A^{-1/2}\rbra*{A^{1/2}CA^{1/2}}^{1/p}A^{-1/2},
    \end{equation}
    and 
    \begin{equation}
        \gamma_p=\begin{cases}
            1 & p>0,\\
            \kappa_A^{-1/p}\kappa_C^{-1/p}& p<0,
        \end{cases}
    \end{equation}
    using
    \begin{itemize}
        \item 
        $\widetilde{O}\rbra*{\kappa_A\kappa_C\log^2\rbra*{1/\epsilon}}$ queries to $U_C$ and
        $\widetilde{O}\rbra*{\kappa_A^2\kappa_C\log^3\rbra*{1/\epsilon}}$ queries to $U_A$;
        \item
        $\widetilde{O}\rbra*{a\kappa_A^2\kappa_C\log^3\rbra*{1/\epsilon}}$ gates; and
        \item 
        $\poly\rbra*{\kappa_A,\kappa_C,\log\rbra*{1/\epsilon}}$ classical time.
    \end{itemize}
\end{lemma}

\begin{proof}
    See Appendix~\ref{app:high-order-riccati}.
\end{proof}

\section{Applications}

\label{sec:applications}

Here we explore two classes of applications for preparing block-encodings of the matrix geometric mean. The first class of applications is to learning problems, in particular for metric learning from data, both quantum and classical. Next we demonstrate how having access to the matrix geometric mean also allows us to compute some fundamental quantities in quantum information, like the quantum fidelity between two mixed states via the Fuchs--Caves observable, as well as geometric R\'enyi relative entropies.

\subsection{Quantum geometric mean metric learning} \label{sec:learning}

In learning problems, there is typically a loss function $L$ that we want to optimise. Suppose we have $D \times D$  positive definite matrices $Y$, $A$, and $C$. We note that here the uniqueness result in Lemma~\ref{lem:yay1} continues to hold. Consider the following optimisation problem:
\begin{align} \label{eq:metricloss1}
    \min_{Y \geq 0}  L(Y), \qquad L(Y) \coloneqq  \operatorname{Tr}(YA)+\operatorname{Tr}(Y^{-1}C).
\end{align}
It turns out that, for given $A$ and $C$, the unique $Y$ minimising $L(Y)$ is $Y=A^{-1} \# C$. In~\cite{zadeh2016geometric}, this result was proven for real positive definite matrices $A$ and $C$, and here we extend it to positive definite Hermitian matrices. In \cite{Wat13}, the same optimisation problem was considered in the context of quantum fidelity, where it was shown that the optimal value of \eqref{eq:metricloss1} is equal to $\operatorname{Tr}[(A^{1/2} C A^{1/2})^{1/2}]$. 

\begin{lemma} \label{lem:lyoptimise} Fix $A$ and  $C$ to be positive definite matrices. The unique solution to $\min_{Y \geq 0} L(Y)$ where $L(Y) = \operatorname{Tr}(YA)+\operatorname{Tr}(Y^{-1}C)$ is the matrix geometric mean $Y=A^{-1} \#C$. 
\end{lemma}

\begin{proof}
    If $L(Y)$ is both strictly convex and strictly geodesically convex, then the solution to $\nabla L(Y)=0$ will also be a global minimiser. For the proof of strict convexity and strict geodesic convexity, see Appendix~\ref{app:convexity}. Now $\nabla L(Y)=A-Y^{-1}CY^{-1}=0$ implies the algebraic Riccati equation $YAY=C$ or $Y=A^{-1} \# C$, which is the unique solution for positive definite matrices $A$ and $C$. 
\end{proof}

We will use this property and map two learning problems -- one for classical data and another for quantum data -- onto this optimisation problem. Using the block-encoding for the matrix geometric mean in Lemma~\ref{lem:simplericcati}, we then devise quantum algorithms for learning a Euclidean metric from data as well as a $1$-class classification algorithm for quantum states. We also extend to the case of weighted learning, where there are unequal contributions to the loss function in Eq.~\eqref{eq:metricloss1} from $\text{Tr}(YA)$ and $\text{Tr}(Y^{-1}C)$.

\subsubsection{Learning Euclidean metric from data} \label{sec:learning-c}

Machine learning algorithms rely on distance measures to quantify how similar  one set of data is to another. Naturally different distance measures can give rise to different results, and so choosing the right metric is crucial for the success of an algorithm. The distance measure itself can in fact be learned for example in a weakly supervised scenario, and this is called metric learning~\cite{kulis2013metric}. Here we are provided with the following two sets $\mathcal{S}$ (similar) and $\mathcal{D}$ (dissimilar) of pairs (training data) 
 \begin{align}
    \mathcal{S} & \coloneqq \{(\vect{x}, \vect{x}')\,|\, \vect{x}, \vect{x}' \text{ are in the same class}\}, \\
     \mathcal{D} & \coloneqq \{(\vect{x}, \vect{x}')\,|\,  \vect{x}, \vect{x}' \text{ are in different classes}\} ,
 \end{align}
and $\{(\vect{x}^{(k)}, \vect{x}'^{(k)})\}_k$ are the data points, where $k$ labels all the pairs that either belong to $\mathcal{S}$ or $\mathcal{D}$. An important example in metric learning is learning the Euclidean metric from data, which can be reformulated as a simple optimisation problem with a closed-form solution. Learning a Euclidean metric is a common form of metric learning, where we can learn a Mahalanobis distance $d_Y$ 
\begin{align}
    d_Y(\vect{x},\vect{x}')=(\vect{x}-\vect{x}')^T Y (\vect{x}-\vect{x}')=\operatorname{Tr}(Y(\vect{x}-\vect{x}')(\vect{x}-\vect{x}')^T),
\end{align}
with $Y$  a real $ D \times D$ symmetric positive definite matrix. To identify a suitable $Y$, one requires a suitable cost function. 

In \textit{geometric mean metric learning}~\cite{zadeh2016geometric}, we want $d_Y$ to be minimal between data in the same class, i.e., $\mathcal{S}$. At the same time, when the data are in different classes, i.e., $\mathcal{D}$, we want $d_{Y^{-1}}$ to be minimal instead. Thus we want to minimise the sum $\sum_{\mathcal{S}} d_Y+\sum_{\mathcal{D}}d_{Y^{-1}}$. This leads to an optimisation problem of the form in Eq.~\eqref{eq:metricloss1}
 \begin{align} \label{eq:metricloss2}
     &  \min_{Y \geq 0} L(Y), \qquad  L(Y) \coloneqq  \operatorname{Tr}(YA)+\operatorname{Tr}(Y^{-1}C) , \\
     & A  =\sum_{(\vect{x}, \vect{x}') \in \mathcal{S}} (\vect{x}-\vect{x}')(\vect{x}-\vect{x}')^T ,\\
     & C =\sum_{(\vect{x}, \vect{x}') \in \mathcal{D}} (\vect{x}-\vect{x}')(\vect{x}-\vect{x}')^T,
 \end{align}
where we assume $A$ and $C$ are positive definite. 
 From Lemma~\ref{lem:lyoptimise}, we see that the optimal solution to Eq.~\eqref{eq:metricloss2} is the matrix geometric mean $Y=A^{-1} \# C$. We see that this is also in fact the solution of the $B=0$ algebraic Riccati equation $YAY=C$. In Lemma~\ref{lem:simplericcati}, we saw that, given access to the block-encodings of $U_A$, $U_C$, their inverses, and controlled versions, we can construct the block-encoding of $Y$, denoted $U_Y$. Lemma~\ref{lem:simplericcati} also shows that the query and gate complexity costs are efficient in $D$, i.e., $O(\text{poly} \log d)$, when the condition numbers for $A$ and $C$ are also polynomial in $\log d$. 

 While we see that it is possible to efficiently recover $U_Y$, it is not sufficient for a direct application to machine learning. Learning $Y$ is part of the learning stage, but reading off the classical components of $Y$ directly from $U_Y$ is inefficient. However, if we consider the testing stage in machine learning, then we need to compute the actual distance $d_Y$ if we are given a new data pair $(\vect{y}, \vect{y}')$, known as testing data. For this testing data, we do not know its classification into $\mathcal{S}$ or $\mathcal{D}$ a priori. Thus the task is to show that, having access to $U_Y$, it is then sufficient to compute $d_Y$ without needing to read out the elements of $Y$. For example, a large value of $d_Y(\vect{y}, \vect{y}')$ means that we should classify $(\vect{y}, \vect{y}') \in \mathcal{D}$, whereas a small value of $d_Y(\vect{y}, \vect{y}')$ means that we should classify $(\vect{y}, \vect{y}') \in \mathcal{S}$. We discuss later the preparation of the block-encodings of $U_A$, $U_C$. 
 
Before proceeding, we first discuss the preparation of a quantum state that we later require. Given a new data pair (testing data) $(\vect{y}, \vect{y}')$ for which we want to compute $d_Y$, where we use the optimal $Y$, we can define a corresponding pure quantum state with $m=O(\log d)$ qubits $|\psi\rangle_{y,y'}=(1/\mathcal{N}_{\psi})\sum_{i=1}^d (\vect{y}-\vect{y}')_i|i\rangle$, with normalisation constant $\mathcal{N}_{\psi}^2=\sum_{i=1}^d (\vect{y}-\vect{y}')_i^2$. Its amplitudes are proportional to $\vect{y}-\vect{y}'$ for any pair $(\vect{y}, \vect{y}')$.  We say that the state $|\psi\rangle_{y,y'}$ has sparsity~$\sigma$ if $\sigma$ is the number of non-zero entries in the amplitude. We can use optimal state preparation schemes~\cite{gleinig2021efficient,STY+23,zhang2022quantum} to prepare $|\psi\rangle_{y,y'}$.
 
\begin{lemma}[Quantum state preparation, \cite{gleinig2021efficient,STY+23,zhang2022quantum}] \label{lem:initialprep}
  A quantum circuit producing an $m$-qubit state $|z\rangle=\sum_{i=1}^{2^m} z_i |i\rangle$ from $|0\rangle$ with given classical entries $\{z_i\}_{i=1}^{2^m}$ can be implemented by using $O(m\sigma)$ CNOT gates, $O(\sigma(\log \sigma +m))$ one-qubit gates, and $O(1)$ ancilla qubits, where the circuit description can be classically computed with time complexity $O(m\sigma^2\log \sigma)$. 
  Further, the gate depth complexity can be reduced to $\Theta(\log (m \sigma))$, if $O(m \sigma \log \sigma)$ ancilla qubits are used.
\end{lemma}

Since here $m=O(\log d)$, we see that as long as $\sigma$ is small, e.g., $\sigma=O(\poly \log d)$, then the total initial state preparation cost, in either gate complexity and number of ancillas is $O(\poly \log d)$. Next we compute $d_Y$.

\begin{theorem} \label{thm:classicallearning}
    Suppose we are given $U_{A}$, $U_{C}$,  which are $(1, \log d, 0)$-block-encodings $A=\sum_{(\vect{x}, \vect{x}') \in \mathcal{S}} (\vect{x}-\vect{x}')(\vect{x}-\vect{x}')^T$ and $C=\sum_{(\vect{x}, \vect{x}') \in \mathcal{D}} (\vect{x}-\vect{x}')(\vect{x}-\vect{x}')^T$, respectively. We also assume access to their inverses and controlled versions. We assume that the data obeys $\kappa_A, \kappa_C=O(\poly \log d)$. Given a testing data pair $(\vect{y}, \vect{y}')$, we assume that the corresponding state $|\psi\rangle_{y,y'}$ has sparsity $O(\poly \log d)$ and $\mathcal{N}_{\psi}=O(\poly \log d)$. Then computing $d_Y(\vect{y},\vect{y}')$ to precision $\epsilon$ has a query and gate complexity $O(\operatorname{poly}(\log d, 1/\epsilon))$. 
\end{theorem}

\begin{proof}
We first make the observation that, for optimal $Y$, 
 \begin{align} \label{eq:tmp}
     d_Y(\vect{y}, \vect{y'})= \mathcal{N}_{\psi}^2\langle \psi|_{y,y'}Y|\psi\rangle_{y,y'} \approx 2\mathcal{N}_{\psi}^2\kappa_A \operatorname{Tr}(\langle 0 |U_Y| 0\rangle |\psi\rangle_{y,y'} \langle \psi|_{y,y'}),
 \end{align}
where $Y=A^{-1} \# C$ and $U_Y$ is a $(2\kappa_A, 5a+11, \epsilon)$ block-encoding of $Y$. The proportionality constant of $2\kappa_A$ comes from Lemma~\ref{lem:simplericcati}, which shows that $\|Y-2\kappa_A \langle 0|_{5a+11} U_Y |0\rangle_{5a+11}\| \leq \epsilon$, where $a=\log d$. For simplicity, we neglect the subscript on the $|0\rangle$ states. 
 This trace can be interpreted as an expectation value of $|\psi\rangle_{y,y'}$ with $Y$ as the observable, and comes from the definition of $d_Y$ and $|\psi\rangle_{y,y'}$. 

To compute this trace given $U_Y$, we observe $\operatorname{Tr}((S \otimes T)X)=\operatorname{Tr}(X_TS)$, where if $T=\sum_n \lambda_n |u_n\rangle \!\langle v_n|$, then $X_T=\sum_n \lambda_n \langle v_n|X|u_n\rangle$. So we can rewrite
\begin{align}
   d_Y(\vect{y}, \vect{y}') & \approx  2\mathcal{N}^2_{\psi} \kappa_A \operatorname{Tr}(\langle 0 |U_Y| 0\rangle |\psi\rangle_{y,y'} \langle \psi|_{y,y'})\\
   & =2\mathcal{N}^2_{\psi} \kappa_A \operatorname{Tr}((|\psi\rangle_{y,y'} \langle \psi|_{y,y'}\otimes |0\rangle \!\langle 0|)U_Y)\\
   & = 2\mathcal{N}^2_{\psi} \kappa_A  \langle \Psi|_{y,y'} U_Y|\Psi\rangle_{y,y'},
\end{align}
where $|\Psi\rangle_{y,y'}=|\psi\rangle_{y,y'} \otimes |0\rangle$. The last expectation value  can be realised with a conventional swap test~\cite{barenco1997stabilization,buhrman2001quantum} between the states $|\Psi\rangle_{y,y'}$ and $U_Y|\Psi\rangle_{y,y'}$. One can also use the destructive SWAP test (i.e., Bell measurements and classical post-processing)~\cite{garcia2013swap}. Alternatively, $\operatorname{Tr}((|\psi\rangle_{y,y'} \langle \psi|_{y,y'}\otimes |0\rangle \!\langle 0|)U_Y)$ can also be computed through a Hadamard test (Lemma~\ref{lemma:hadamard-test}), where one is given the controlled-$U_Y$ and state $|\psi\rangle_{y,y'} \otimes |0\rangle$. For example, applying the unitary $U_Y$ to $|\psi\rangle_{y,y'} \otimes |0\rangle$ and using the swap test with $|\psi\rangle_{y,y'} \otimes |0\rangle$, we recover $\langle \psi|_{x,x'} Y |\psi\rangle_{x,x'}$ to precision $\epsilon$ with query and gate complexity $O(\operatorname{poly}(\log d, 1/\epsilon))$, when $\kappa_A=O(\poly \log d)$. Now, $d_Y(\vect{y}, \vect{y}')=\mathcal{N}^2_{\psi} \langle \psi|_{y,y'} Y |\psi\rangle_{y,y'}$. Since we only have $\sigma=O(\text{poly}\log d)$ non-zero entries in $\vect{y}-\vect{y}'$, the cost in the classical computation of the normalisation constant is also of order $O(\poly \log d)$. Assuming that $\mathcal{N}^2_{\psi}=O(\poly\log d)$, then we recover $d_Y$ efficiently when given access to $U_Y$ and $|\psi\rangle_{y, y'}$.

We saw that preparing the state $|\psi\rangle_{y,y'}$ according to Lemma~\ref{lem:initialprep} incurs a cost $O(\poly\log d)$.
From Lemma~\ref{lem:simplericcati} we can construct a $(2\kappa_A, O(\log d), \epsilon)$-block-encoding of $Y$ with gate and query complexity $O(\operatorname{poly}(\kappa_A, \kappa_C, \log (1/\epsilon)))$. Since $\kappa_A, \kappa_C=O(\text{poly} \log d)$, then the theorem is proved.  
\end{proof}
Thus, if our assumptions are obeyed, the quantum cost for computation of $d_Y$ can be $O(\text{poly} \log d)$, whereas  classical numerical algorithms for computing the matrix geometric mean alone has cost $O(\operatorname{poly} d)$ for $d \times d$ matrices~\cite{bini2011note, iannazzo2016geometric}. 

In Theorem~\ref{thm:classicallearning}, we also assumed access to $U_A$, $U_C$. We show below the preparation of a block-encoding of density matrices which are proportional to $A$ and $C$ and how this can be used to compute $d_Y$ efficiently. First consider Lemma~\ref{lmm:purified to block-encoding}, which shows how to create a block-encoding of a density matrix. We first observe that we can define density matrices $\rho_A$ and $\rho_C$ where rewrite 
\begin{align}
   &  \rho_A=\frac{A}{\text{Tr}(A)}, \qquad A=\sum_{k \in \mathcal{S}} \mathcal{N}^2_{\psi_k}|\psi_k\rangle\! \langle \psi_k|, \qquad \text{Tr}(A)=\sum_{k \in \mathcal{S}} \mathcal{N}^2_{\psi_k}, \nonumber \\
   & \rho_C=\frac{C}{\text{Tr}(C)}, \qquad C=\sum_{k \in \mathcal{D}} \mathcal{N}^2_{\psi_k}|\psi_k\rangle\! \langle \psi_k|, \qquad \text{Tr}(C)=\sum_{k \in \mathcal{D}} \mathcal{N}^2_{\psi_k},
\end{align}
where $|\psi_k\rangle=(1/\mathcal{N}_{\psi_k})\sum_{i=1}^d (\vect{x}^{(k)}-\vect{x}'^{(k)})_i|i\rangle$ and $\mathcal{N}^2_{\psi_k}=\sum_{i=1}^d (\vect{x}^{(k)}-\vect{x}'^{(k)})_i^2$ is the corresponding normalisation. Then from Lemma~\ref{lmm:purified to block-encoding}, if we are given unitaries  $V_A$ and $V_C$ that prepare  purifications of $\rho_A$ and $\rho_C$, respectively, it is possible to create $U_{\rho_A}$ and $U_{\rho_C}$ using one query to $V_A$ and $V_C$ respectively and $O(\log d)$ gates.  One such class of  states purifying  $\rho_A$ and $\rho_C$ are
\begin{align}
|\Sigma_A\rangle & \coloneqq \sum_{k \in \mathcal{S}} \sqrt{p^{(A)}_k}|k\rangle |\psi_k\rangle,\\
|\Sigma_C\rangle & \coloneqq  \sum_{k \in \mathcal{D}} \sqrt{p^{(C)}_k}|k\rangle |\psi_k\rangle,
\end{align}
where
\begin{align}
    p_k^{(A)} & \coloneqq \mathcal{N}^2_{\psi_k}/ \sum_{l \in \mathcal{S}} \mathcal{N}_{\psi_l}^2 , \\
    p_k^{(C)} & \coloneqq \mathcal{N}^2_{\psi_k}/\sum_{l \in \mathcal{D}} \mathcal{N}_{\psi_l}^2.
\end{align}
 To prepare $|\Sigma_A\rangle$ and $|\Sigma_C\rangle$ we require the controlled unitaries $V_A=\sum_{k \in \mathcal{S}} |k\rangle\! \langle k| \otimes W_k^{(A)}$ and $V_C=\sum_{k \in \mathcal{D}} |k\rangle\! \langle k| \otimes W_k^{(C)}$ acting on states $\sum_{k \in \mathcal{S}} \sqrt{p_k^{(A)}}|k\rangle |0\rangle$ and $\sum_{k \in \mathcal{D}} \sqrt{p_k^{(A)}}|k\rangle |0\rangle$ respectively. Here $W_k^{(A)}$ and $W_k^{(C)}$ are the state preparation circuits from Lemma~\ref{lem:initialprep} that create $|\psi_k\rangle$ where $k \in \mathcal{S}$ and $k\in \mathcal{D}$, respectively. Since $W_k^{(A)}$ and $W_k^{(C)}$ are known circuits and assuming $\sigma=O(\poly\log d)$, it is similarly efficient and also straightforward to realise $V_A$ and $V_C$. Then from Lemma~\ref{lmm:purified to block-encoding} it is possible to create $(1, O(\log d), 0)$-block-encodings of $\rho_A$ and $\rho_C$ with gate and query complexity $O(\poly\log d)$, denoted $U_{\rho_A}$ and $U_{\rho_C}$, respectively. In the case where $\text{Tr}(A)=1=\text{Tr}(C)$, then this automatically gives us the unitaries $U_A$ and $U_C$ required in Theorem~\ref{thm:classicallearning}. 

For general classical data, $\text{Tr}(A)=1=\text{Tr}(C)$ would not hold in general. However, since $A=\text{Tr}(A) \rho_A$ and $C=\text{Tr}(C) \rho_C$, the proof in Theorem~\ref{thm:classicallearning} holds in the same way if we began with $U_{\rho_A}$ and $U_{\rho_C}$, from which we can create $U_{Y'}$ where $Y' \equiv \rho_A^{-1} \# \rho_C=(\text{Tr}(C)/\text{Tr}(A))^{1/2} Y$. This implies 
\begin{equation}
\langle 0|U_Y |0\rangle \approx Y=(\text{Tr}(A)/\text{Tr}(C))^{1/2} Y' \approx (\text{Tr}(A)/\text{Tr}(C))^{1/2} \langle 0| U_{Y'}|0\rangle.    
\end{equation}
Following through the same proof idea as in Theorem~\ref{thm:classicallearning} allows us to extract $d_Y'$. To recover $d_Y$, we just use $d_Y \approx (\text{Tr}(A)/\text{Tr}(C))^{1/2} d_Y'$. These normalisations can be efficiently recovered by assuming the states $|\psi_k\rangle$ have low sparsity $\sigma_k=O(\poly \log d)$ for each $k$. This also means that the normalisations $\text{Tr}(A)$ and $\text{Tr}(C)$ are efficient to compute. So long as $(\text{Tr}(A)/\text{Tr}(C))^{1/2}$ is $O(\poly \log d)$, then $d_Y$ is efficiently estimable.

\subsubsection{1-class quantum learning} \label{sec:learning-q}

Here we propose a new quantum classification problem that is a 1-class problem. This means that given a quantum state, we only want to know whether this state belongs to a class $\mathcal{A}$ or not. This problem occurs in many areas in machine learning, in particular in anomaly detection, where $\mathcal{A}$ is the class of states that are considered anomalous. Here we can be provided with the following training data:
\begin{align} \label{eq:qlearning}
    & \rho=\frac{1}{N}\sum_{i=1}^N \rho_i, \qquad \rho_i \in \mathcal{A},  \\
    & \sigma=\frac{1}{M} \sum_{i=1}^M \sigma_i, \qquad \sigma_i \notin \mathcal{A},
\end{align}
where $\{\rho_i\}_i$ and $\{\sigma_i\}_i$ are sets of $D$-dimensional states. 
In anomaly detection scenarios, there are usually much fewer examples of anomalous states than `normal' states, so that $N \ll M$. However, we will not focus on subtleties associated with imbalanced training data here. 

Suppose that we have an incoming quantum state $\xi$ and we want to flag this as belonging to the class $\mathcal{A}$ or not. Then it is useful to learn an `observable' or a `witness' $Y$ such that its expectation value $\operatorname{Tr}(Y \xi)$ is large when $\xi$ is flagged as anomalous, belonging to $\mathcal{A}$, but this value is small when $\xi$ is `normal'. Thus we can set up an optimisation problem of the form 
\begin{align}
\label{eq:quantum1class}
       \min_{Y \geq 0} L(Y), \qquad L(Y) \coloneqq  \operatorname{Tr}(Y\sigma)+\operatorname{Tr}(Y^{-1}\rho).
\end{align}
It is sensible in the above to minimise $\operatorname{Tr}(Y^{-1}\rho)$ in $L(Y)$ above since it is simple to show that a small value of $\operatorname{Tr}(Y^{-1}\rho)$ implies a large value of $\operatorname{Tr}(Y \rho)$. Since $\rho$ is a density matrix, it can be shown that $\operatorname{Tr}(Y^{-1}\rho) \geq \operatorname{Tr}(Y\rho)^{-1}$, which is a consequence of the operator Jensen inequality (see~\cite[Eqs.~(29)--(35)]{katariya2021rld}). Thus $\operatorname{Tr}(Y^{-1}\rho) \leq \lambda$ implies $\operatorname{Tr}(Y \rho) \geq 1/\lambda$. 

When $\rho$ and $ \sigma$ are (positive definite) density matrices, the unique solution to Eq.~\eqref{eq:quantum1class} is given by the matrix geometric mean $Y=\sigma^{-1} \# \rho$. We can therefore proceed as before to compute $\operatorname{Tr}(Y \xi)$, except now we do not need to be concerned with state preparation of $\rho$ and $\sigma$, and we can assume that we are given copies of $\rho$ and $\sigma$. Thus, given access to $U_Y$, we can estimate the following expectation:
\begin{align}
    \operatorname{Tr}(Y \xi) \approx 2\kappa_{\sigma} \operatorname{Tr}((\xi \otimes |0\rangle \!\langle 0|) U_Y),
\end{align}
where the $\kappa_{\sigma}$ constant follows from Lemma~\ref{lem:simplericcati} and the error in the above estimate is upper-bounded by $\epsilon$. 
We then have the following result.

\begin{theorem} \label{thm:unweightedquantum}
    Suppose that we are given the block-encodings $U_{\rho}$ and $U_{\sigma}$, where $\rho \in \mathcal{A}$, $\sigma \notin \mathcal{A}$ and that we are also given access to multiple copies of $\xi$. Suppose further that $\kappa_{\rho}, \kappa_{\sigma}=O(\poly\log d)$. Then computing $\operatorname{Tr}(Y\xi)$ for the optimal $Y$ in Eq.~\eqref{eq:quantum1class} to precision~$\epsilon>0$ has a query and gate complexity $O(\operatorname{poly}(\log d, 1/\epsilon))$. 
\end{theorem}

\begin{proof}
 From Lemma~\ref{lem:simplericcati} we can construct a $(2\kappa_\sigma, O(\log d), \epsilon)$-block-encoding of $Y$ with gate and query complexity $O(\operatorname{poly}(\kappa_{\rho}, \kappa_{\sigma}, \log (1/\epsilon)))$. Considering that $\kappa_{\rho}, \kappa_{\sigma}=O(\poly\log d)$, applying the unitary $U_Y$ to $\xi \otimes |0\rangle \!\langle 0|$, and using the Hadamard test (Lemma~\ref{lemma:hadamard-test}) with $\xi\otimes |0\rangle \!\langle 0|$, we recover $\operatorname{Tr}(Y \xi)$ to precision $\epsilon$ with query and gate complexity $O(\operatorname{poly}(\log d, 1/\epsilon))$.
\end{proof}

 We emphasise that this problem is entirely quantum in nature as we are given directly only quantum data.

\begin{remark}
    The assumption of $U_\rho$ and $U_\sigma$ as block-encodings of $\rho$ and $\sigma$, respectively, is without loss of generality in practice. 
    There are two quantum input models for quantum states that are commonly employed in quantum algorithms:
    \begin{itemize}
        \item \textbf{Quantum query access model}. In this model, quantum unitary oracles $\mathcal{O}_\rho$ and $\mathcal{O}_\sigma$ are given such that they prepare purifications of $\rho$ and $\sigma$, respectively. 
        By the technique of purified density matrix in \cite{low2019hamiltonian} (see Lemma~\ref{lmm:purified to block-encoding}), we can implement $U_\rho$ and $U_\sigma$ from $\mathcal{O}_\rho$ and $\mathcal{O}_\sigma$ with query and gate complexity $\widetilde O\rbra*{1}$. Therefore, Theorem~\ref{thm:unweightedquantum} can be adapted to the quantum query access model with query and gate complexity $O(\poly(\log d, 1/\epsilon))$.
        \item \textbf{Quantum sample access model}. In this model, independent and identical copies of $\rho$ and $\sigma$ are given. 
        By the technique of density matrix exponentiation \cite{lloyd2014quantum,KLL+17}, we can implement unitary operators that are block-encodings of $\rho$ and $\sigma$ using their copies (which was first noted in \cite{GLM+22} and later investigated in \cite{GP22,WZ23,WZ24}). 
        In this way, Theorem~\ref{thm:unweightedquantum} can be adapted to the quantum sample access model with sample and gate complexity $O(\poly(\log d, 1/\epsilon))$.
    \end{itemize}
\end{remark}

 A very interesting observation to note here is that the matrix geometric mean solution $Y$ to Eq.~\eqref{eq:quantum1class} is precisely the Fuchs--Caves observable~\cite{Fuchs1995}, which is important for distinguishing two states $\rho$ and $\sigma$. From this observation, we can motivate the Fuchs--Caves observable as the observable that gives rise to a kind of `optimal witness' that distinguishes  $\rho$ and $\sigma$ and the value of this `witness' is precisely quantum fidelity, as shown in the next section. This provides an alternative motivation for the form of quantum fidelity between two mixed states from a metric learning viewpoint. In fact, a protocol involving a measurement of the Fuchs--Caves observable also achieves an upper bound on sample complexity for the quantum hypothesis testing problem in distinguishing  $\rho$ and $\sigma$ \cite[Appendix~F]{cheng2024sample}. Thus, up to constant factors, the strategy also minimises the number of copies of each state used for a given tolerated precision in distinguishing the states. We note that the loss function also appears in Eq.~(6) in~\cite{Wat13}, but this is motivated from a different perspective.

\subsubsection{Extension to weighted geometric mean metric learning}

The two terms in the loss function in Eq.~\eqref{eq:qlearning}, involving $\sigma$ and $\rho$ respectively, have equal weights. This means the learning algorithm deems closeness to $\rho$ and farness to $\sigma$ of equal `importance'. However, there are  scenarios, especially in anomaly detection, where asymmetry is preferable. For example, this occurs when there is a higher cost in getting false negatives. 

Modifying Eq.~\eqref{eq:qlearning} by simply multiplying each of the two terms by different constants $\alpha, \beta$ leads to $L(Y)=\alpha \text{Tr}(Y \sigma)+\beta \text{Tr}(Y^{-1} \rho)$. However, this only rescales the optimal solution $Y \rightarrow (\beta/\alpha)^{1/2} Y$ by a constant factor, as observed in \cite{zadeh2016geometric}. A new loss function is therefore necessary for the asymmetric case. 

Following \cite{zadeh2016geometric}, one can first observe that the solution $Y=\sigma^{-1} \# \rho$ is in fact also a solution to the following optimisation problem when $t=1/2$:
\begin{align} \label{eq:weightedL}
    \min_{\tilde{Y} \geq 0} L_t(\tilde{Y}), \qquad L_t(\tilde{Y})\coloneqq(1-t) \delta(\tilde{Y}, \sigma^{-1})+t\delta(\tilde{Y},\rho), \qquad t \in [0,1],
\end{align}
where $\delta$ is the geodesic distance defined in Eq.~\eqref{eq:geodesicdelta}. While the mathematical proof is more involved, this fact can easily be understood from the geometric viewpoint. Here $Y=\sigma^{-1} \# \rho$ can be understood as the midpoint along the unique geodesic in Riemannian space joining $\sigma^{-1}$ and $\rho$. When $t=1/2$, the optimal $\tilde{Y}$ is then the point along this geodesic that simultaneously minimises the distance between $\tilde{Y}$ and $\sigma^{-1}$, as well as $\tilde{Y}$ and $\rho$. This clearly must be the midpoint. Similar geometric reasoning leads one to generalise to $t \neq 1/2$ where the solution to Eq.~\eqref{eq:weightedL} is the weighted matrix geometric mean $\tilde{Y}=\sigma^{-1} \#_t \rho$. That this is the unique solution to Eq.~\eqref{eq:weightedL} is a special case (the $n=2$ case) in~\cite{lawson2011monotonic} and proofs can also be found in~\cite[Chapter~6]{bhatia2009positive}. Also see~\cite{zadeh2016geometric} for a discussion in the context of geometric mean metric learning. 

We can proceed similarly to 1-class quantum learning algorithm with equal weights as described in the previous section. The goal is also to output $\text{Tr}(\tilde{Y} \xi)$ for some input test state $\xi$. Here we require instead the construction of block-encodings of the weighted matrix geometric mean, as given in Lemma~\ref{lem:high-order-riccati}. However, for any $t>0$ ($p>0$ in Lemma~\ref{lem:high-order-riccati}), we see that there is no scaling difference for constructing the block-encoding for the weighted version. Thus the cost, up to constant and logarithmic factors, is identical for the quantum weighted geometric mean metric learning algorithm as for the unweighted version in Theorem~\ref{thm:unweightedquantum}. 

\subsection{Estimation of quantum fidelity and geometric R\'enyi relative entropies}

Here we describe our quantum algorithms for estimating quantum fidelity and geometric R\'enyi relative entropies using our quantum subroutines for preparing block-encodings of the standard and weighted matrix geometric means. 

\subsubsection{Fidelity} \label{sec:fidelity}

The fidelity between two mixed quantum states is defined by~\cite{Uhl76}
\begin{equation}
    F\rbra*{\rho, \sigma} \coloneqq  \Tr \rbra*{\rbra*{\sigma^{1/2} \rho \sigma^{1/2}}^{1/2}},
\end{equation}
which is a commonly considered measure of the closeness of or similarity between two quantum states. 
Estimating the value of fidelity is a fundamental task in quantum information theory.
When given matrix descriptions of the states $\rho$ and $\sigma$, it can be calculated directly using the formula above or as the solution to a semidefinite optimisation problem~\cite{Wat13}. 
Recently, several time-efficient quantum algorithms for fidelity estimation have been developed when one has access to state-preparation circuits of $\rho$ and $\sigma$~\cite{WZC+23,GP22,WGL+22}.

Here, we introduce a new approach for fidelity estimation that is based on the Fuchs--Caves observable \cite{Fuchs1995}.
For two quantum states $\rho$ and $\sigma$, this observable is given by $M = \sigma^{-1} \# \rho$. 
Then, the fidelity between $\rho$ and $\sigma$ can be represented as the expectation of $M$ with respect to $\sigma$ (cf.\ \cite[Eq.~(9.159)]{wilde2017quantum}):
\begin{equation} \label{eq:fidelity}
    F\rbra*{\rho, \sigma} = \Tr\rbra*{M\sigma}.
\end{equation}

\begin{theorem} [Fidelity estimation via Fuchs--Caves observable] \label{thm:fidelity-estimation}
    Suppose that $\mathcal{O}_\rho$ and $\mathcal{O}_\sigma$ prepare purifications of mixed quantum states $\rho$ and $\sigma$, respectively. 
    Then, we can estimate $F\rbra*{\rho, \sigma}$ within additive error $\epsilon$ using 
    $\widetilde O\rbra*{\min\cbra{\kappa_\rho^2, \kappa_\sigma^2} \cdot\kappa_\rho \kappa_\sigma/\epsilon  }$
    queries to $\mathcal{O}_\rho$ and $\mathcal{O}_\sigma$, where $\kappa_\rho,\kappa_\sigma>0$ are such that $\rho \geq I/\kappa_\rho$ and $\sigma \geq I/\kappa_\sigma$.
\end{theorem}

\begin{proof}
    Suppose that $\rho$ and $\sigma$ are $n$-qubit mixed quantum states and $\mathcal{O}_\rho$ and $\mathcal{O}_\sigma$ are $\rbra*{n+a}$-qubit unitary operators. 
    By Lemma~\ref{lmm:purified to block-encoding}, we can implement two unitary operators $U_{\rho}$ and $U_{\sigma}$ that are $\rbra*{1, n+a, 0}$-block-encodings of $\rho$ and $\sigma$ using $O\rbra*{1}$ queries to $\mathcal{O}_\rho$ and $\mathcal{O}_\sigma$, respectively. 
    Then, by applying Lemma~\ref{lem:simplericcati}, we can implement a $\rbra*{2\kappa_{\sigma}, b, \delta}$-block-encoding $U_M$ of $M = \sigma^{-1} \# \rho$, using $\widetilde O\rbra*{\kappa_\sigma \kappa_\rho \log^2\rbra*{1/\delta}}$ queries to $U_{\rho}$ and $\widetilde O\rbra*{\kappa_\sigma^2 \kappa_\rho \log^3\rbra*{1/\delta}}$ queries to $U_{\sigma}$, where $b = O\rbra*{n+a}$, and $\kappa_\rho$ and $\kappa_\sigma$ satisfy $\rho \geq I/\kappa_\rho$ and $\sigma \geq I/\kappa_\sigma$.

    By the Hadamard test (given in Lemma~\ref{lemma:hadamard-test}), there is a quantum circuit $C$ that outputs $0$ with probability $\frac{1}{2}\rbra*{ 1+\Re{\Tr\rbra{\bra{0}_b U_M \ket{0}_b \sigma}} }$, using one query to $U_M$ and one sample of~$\sigma$.
    By noting that
    \begin{equation}
        \abs*{ 2\kappa_\sigma \Re{\Tr\rbra{\bra{0}_b U_M \ket{0}_b \sigma}} - \Tr\rbra{M\sigma} } \leq \Theta\rbra{\delta},
    \end{equation}
    we conclude that an $O\rbra{\epsilon/\kappa_\sigma}$-estimate of $\Re{\Tr\rbra{\bra{0}_b U_M \ket{0}_b \sigma}}$ with $\delta = \Theta\rbra{\epsilon/\kappa_\sigma}$ suffices to obtain an $\epsilon$-estimate of $\Tr\rbra{M\sigma}$ (which is the fidelity according to Eq.~\eqref{eq:fidelity}). 
    By quantum amplitude estimation (given in Lemma~\ref{lemma:amp-estimation}), this can be done using $O\rbra{\kappa_\sigma/\epsilon}$ queries to $C$. 

    In summary, an $\epsilon$-estimate of $F\rbra{\rho, \sigma}$ can be obtained by using $\widetilde O\rbra{\kappa_\sigma^3 \kappa_\rho/\epsilon}$ queries to $\mathcal{O}_\sigma$ and $\widetilde O\rbra{\kappa_\sigma^2 \kappa_\rho/\epsilon}$ queries to $\mathcal{O}_\rho$.
    The proof is completed by taking the minimum over symmetric cases (i.e., simply flipping the role of $\rho$ and $\sigma$ since the fidelity formula is symmetric under this exchange).
\end{proof}

The current best quantum query complexity of fidelity estimation is $\widetilde O\rbra{r^{2.5}/\epsilon^5}$, due to~\cite{GP22}, where $r$ is the lower rank of the two input mixed quantum states.
In comparison, our quantum algorithm for fidelity estimation based on the Fuchs--Caves observable, as given in Theorem~\ref{thm:fidelity-estimation}, has a better dependence on the additive error $\epsilon$, if $\kappa_\rho$ and $\kappa_\sigma$ are known in advance.

Moreover, we note that the $\epsilon$-dependence of the quantum algorithm given in Theorem~\ref{thm:fidelity-estimation} is optimal (up to polylogarithmic factors), as stated in Lemma~\ref{lemma:fidelity-lower-bound} below.

\begin{lemma} [Optimal $\epsilon$-dependence of fidelity estimation] \label{lemma:fidelity-lower-bound} 
    Suppose that $\mathcal{O}_\rho$ and $\mathcal{O}_\sigma$ prepare purifications of mixed quantum states $\rho$ and $\sigma$, respectively, satisfying $\rho \geq I/\kappa_\rho$ and $\sigma \geq I/\kappa_\sigma$ for $\kappa_\rho, \kappa_\sigma >0$. 
    Then, every quantum algorithm that estimates $F\rbra{\rho, \sigma}$ within additive error~$\epsilon$ requires query complexity $\Omega\rbra{1/\epsilon}$ even if $\kappa_\rho = \kappa_\sigma = \Theta\rbra{1}$.
\end{lemma}

\begin{proof}
    See Appendix~\ref{app:fidelity-lower-bound}.
\end{proof}

In addition to the optimal $\epsilon$-dependence of fidelity estimation in Lemma \ref{lemma:fidelity-lower-bound}, a quantum query algorithm for estimating the fidelity $F\rbra{\ket{\psi}, \ket{\varphi}}$ between pure states with query complexity in $\Theta\rbra{1/\epsilon}$ was given in \cite{Wan24}. 

\begin{remark}[Sample complexity for fidelity estimation] \label{remark:sample-fidelity-estimation}
    Using the method in Theorem~\ref{thm:fidelity-estimation}, we can also estimate the fidelity by using only samples of quantum states, which is achieved by density matrix exponentiation \cite{lloyd2014quantum,lloyd2014quantum,KLL+17,GP22,WZ24,GKP+24}. 
    As analysed in Appendix~\ref{app:fidelity-estimation-sample}, the sample complexity for fidelity estimation is shown to be $\widetilde{O}\rbra{\min\cbra{\kappa_\rho^5, \kappa_\sigma^5} \cdot \kappa_\rho^2 \kappa_\sigma^2/\epsilon^3} = \poly\rbra{\kappa_\rho, \kappa_\rho} \cdot \widetilde{O}\rbra{1/\epsilon^3}$.
    The prior known sample complexity for fidelity estimation is $\widetilde{O}\rbra{r^{5.5}/\epsilon^{12}}$ due to \cite{GP22}, where $r$ is the lower rank of the two input mixed quantum states.
    
    We also show a sample lower bound $\Omega\rbra{1/\epsilon^2}$ for fidelity estimation in Appendix~\ref{app:fidelity-estimation-sample} even if $\kappa_\rho = \kappa_\sigma = \Theta\rbra{1}$ using the method in the proof of Lemma~\ref{lemma:fidelity-lower-bound}; this can be seen as an analog of the sample lower bound $\Omega\rbra{1/\epsilon^2}$ for pure-state fidelity estimation in \cite{ALL22}.
    In addition, a sample lower bound $\Omega\rbra{r/\epsilon}$ for (low-rank) fidelity estimation is implied in \cite{OW21,BOW19}.

    Currently, quantum algorithms for fidelity estimation with optimal sample complexity are only known for pure states. 
    For estimating the squared fidelity $F^2\rbra{\ket{\psi}, \ket{\varphi}}$, 
    the sample complexity $\Theta\rbra{1/\epsilon^2}$ can be achieved by the SWAP test \cite{buhrman2001quantum}.
    In \cite{ALL22}, they showed that $\Theta\rbra{\max\cbra{1/\epsilon^2, \sqrt{d}/\epsilon}}$ samples are sufficient and necessary to estimate $F^2\rbra{\ket{\psi}, \ket{\varphi}}$ when only single-copy measurements are allowed. Recently in \cite{WZ24b}, the sample complexity of estimating the fidelity $F\rbra{\ket{\psi}, \ket{\varphi}}$ was shown to be $\Theta\rbra{1/\epsilon^2}$.
\end{remark}

\subsubsection{Geometric fidelity and geometric R\'enyi relative entropy} \label{sec:renyi}

Here we present, to the best of our knowledge, the first quantum algorithm for computing the geometric $\alpha$-R\'enyi relative entropy, as introduced in \cite{Matsumoto2018}. For $\alpha \in (0,1)\cup (1, 2]$, the geometric $\alpha$-R\'enyi relative entropy is defined as (see, e.g., \cite[Eq.~(9)]{fang2021GeometricRenyiDivergence} and \cite[Eq.~(7.6.1)]{khatri2024principles}) 
\begin{equation}
    \widehat{D}_{\alpha}\rbra*{\rho \Vert \sigma} \coloneqq \frac{1}{\alpha-1} \log \widehat{F}_\alpha\rbra*{\rho, \sigma},
\end{equation}
where
\begin{equation} \label{eq:def-hat-F-alpha}
    \widehat{F}_\alpha\rbra*{\rho, \sigma} \coloneqq  \Tr\rbra*{\rho \#_{1-\alpha} \sigma} = \Tr\rbra*{\sigma \#_\alpha \rho}
\end{equation}
is known as the geometric $\alpha$-R\'enyi relative quasi-entropy. When $\alpha \in (0,1)$, we also refer to $\widehat{F}_{\alpha}\rbra*{\rho, \sigma}$ as the geometric $\alpha$-fidelity. 
In particular, for the case of $\alpha = 1/2$, the quantity $\widehat{F}_{1/2}\rbra*{\rho, \sigma}$ is the geometric fidelity (also known as the Matsumoto fidelity)~\cite{Mat10,cree2020fidelity}. 
The $\alpha$-geometric R\'enyi relative entropy has several uses in quantum information theory, especially in analysing protocols involving feedback~\cite{fang2021GeometricRenyiDivergence,katariya2021geometric,DKQSWW23}. 

Here we present quantum algorithms in Theorems~\ref{thm:hat-F-alpha} and~\ref{thm:g-Renyi-relative-entropy} for computing the geometric R\'enyi relative (quasi-)entropy.

\begin{theorem} \label{thm:hat-F-alpha}
    Suppose that $\mathcal{O}_\rho$ and $\mathcal{O}_\sigma$ prepare purifications of mixed quantum states $\rho$ and $\sigma$, respectively. 
    Then, for $\alpha \in \rbra*{0, 1} \cup (1, 2]$, we can estimate $\widehat{F}_\alpha\rbra*{\rho, \sigma}$ to within additive error~$\epsilon$ using 
    \begin{itemize}
        \item 
    $\widetilde O\rbra{\min\cbra{\kappa_\rho, \kappa_\sigma}^{\min\cbra{1+\alpha, 2-\alpha}}\cdot \kappa_\rho\kappa_\sigma/\epsilon }$
    queries for $\alpha \in (0, 1)$, and
    \item 
    $\widetilde O\rbra{  \min\cbra{\kappa_\rho\kappa_\sigma^{\alpha-1}, \kappa_\rho^{\alpha-1}\kappa_\sigma, \kappa_{\rho}^{1+\alpha}, \kappa_\sigma^{1+\alpha}}\cdot \kappa_\rho\kappa_\sigma/\epsilon}$
    queries for $\alpha \in (1, 2]$
    \end{itemize}
    to $\mathcal{O}_\rho$ and $\mathcal{O}_\sigma$, where $\kappa_\rho,\kappa_\sigma > 0 $ satisfy $\rho \geq I/\kappa_\rho$ and $\sigma \geq I/\kappa_\sigma$.

    In particular, when $\alpha = 1/2$, $\widehat{F}_{1/2}\rbra{\rho, \sigma}$ is the geometric fidelity (also known as the Matsumoto fidelity), which can be estimated using $\widetilde O\rbra{ \min\cbra{\kappa_\rho, \kappa_\sigma}^{3/2}\cdot \kappa_\rho \kappa_\sigma/\epsilon}$ queries to $\mathcal{O}_\rho$ and~$\mathcal{O}_\sigma$.
\end{theorem}

\begin{proof} See Appendix~\ref{app:prooftheorem18}. 
\end{proof}

\begin{theorem} \label{thm:g-Renyi-relative-entropy}
    Suppose that $\mathcal{O}_\rho$ and $\mathcal{O}_\sigma$ prepare purifications of mixed quantum states $\rho$ and~$\sigma$, respectively. 
    Then, for $\alpha \in \rbra{0, 1} \cup (1, 2]$, we can estimate $\widehat{D}_{\alpha}\rbra{\rho \Vert \sigma}$ within additive error~$\epsilon$ using 
    \begin{itemize}
        \item $\widetilde O\rbra{  \min\cbra{\kappa_\rho, \kappa_\sigma}^{\min\cbra{1+\alpha, 2-\alpha}}\cdot \kappa_\rho\kappa_\sigma^{2-\alpha}/\epsilon}$ queries for $\alpha \in (0, 1)$, and
        \item $\widetilde O\rbra{  \min\cbra{\kappa_\rho\kappa_\sigma^{\alpha-1}, \kappa_\rho^{\alpha-1}\kappa_\sigma, \kappa_{\rho}^{1+\alpha}, \kappa_\sigma^{1+\alpha}}\cdot \kappa_\rho^\alpha\kappa_\sigma/\epsilon}$ queries for $\alpha \in (1, 2]$
    \end{itemize}
    to $\mathcal{O}_\rho$ and $\mathcal{O}_\sigma$, where $\kappa_\rho,\kappa_\sigma > 0 $ satisfy $\rho \geq I/\kappa_\rho$ and $\sigma \geq I/\kappa_\sigma$.
\end{theorem}

\begin{proof} See Appendix~\ref{app:prooftheorem18}. 
\end{proof}

Notably, we show that our quantum algorithm for estimating the geometric fidelity $\widehat{F}_{1/2}\rbra{\rho, \sigma}$ achieves an optimal $\epsilon$-dependence. 
The optimality also holds for $\widehat{F}_{\alpha}\rbra{\rho, \sigma}$ with $\alpha \in (0, 1)$.  

\begin{lemma} [Optimal $\epsilon$-dependence of geometric $\alpha$-fidelity estimation] \label{lemma:optimal-geo-fidelity}
    Suppose that $\mathcal{O}_\rho$ and $\mathcal{O}_\sigma$ prepare purifications of mixed quantum states $\rho$ and $\sigma$, respectively, with $\rho \geq I/\kappa_\rho$ and $\sigma \geq I/\kappa_\sigma$, where $\kappa_\rho,\kappa_\sigma > 0 $. 
    Then, for any constant $\alpha \in (0, 1)$, any quantum algorithm that estimates $\widehat{F}_{\alpha}\rbra{\rho, \sigma}$ within additive error $\epsilon$ requires query complexity $\Omega\rbra{1/\epsilon}$ even if $\kappa_\rho = \kappa_\sigma = \Theta\rbra{1}$, where $\Omega\rbra{\cdot}$ hides a constant factor that depends only on $\alpha$.
\end{lemma}

\begin{proof}
    See Appendix~\ref{app:fidelity-lower-bound}.
\end{proof}
It remains an open problem for optimality still holds if $\alpha \in (1, 2]$. However, note that when $\alpha \in (1, 2]$, the inequality $\widehat{F}_{\alpha}\rbra{\rho, \sigma} \geq 1$ holds, and so $\widehat{F}_{\alpha}\rbra{\rho, \sigma}$ cannot be interpreted as a fidelity for these values of $\alpha$; thus, different techniques are required in order to establish optimality.  

\begin{remark}
    Similar to Remark~\ref{remark:sample-fidelity-estimation}, for estimating the corresponding quantities, we can extend Theorems~\ref{thm:hat-F-alpha} and \ref{thm:g-Renyi-relative-entropy} to quantum algorithms with sample complexity $\poly\rbra{\kappa_\rho, \kappa_\rho} \cdot \widetilde{O}\rbra{1/\epsilon^3}$, and extend Lemma~\ref{lemma:optimal-geo-fidelity} to a sample lower bound of $\Omega\rbra{1/\epsilon^2}$.
\end{remark}

\section{BQP-hardness} \label{sec:BQP}

In this section, we consider the hardness of computing the matrix geometric mean. 
Precisely, we show that our quantum algorithm for matrix geometric means (given in Lemma~\ref{lem:simplericcati}) can be used to solve a $\mathsf{BQP}$-complete problem (defined in Problem~\ref{prob:mgm}).
Roughly speaking, this problem pertains to testing a certain property of the matrix geometric mean of two well-conditioned sparse matrices. 

\begin{problem} [Matrix geometric mean] \label{prob:mgm}
    For functions $\kappa_A \colon \mathbb{N} \to \mathbb{N}$ and $\kappa_C \colon \mathbb{N} \to \mathbb{N}$, let $\textup{MGM}\rbra{\kappa_A, \kappa_C}$ be a decision problem defined as follows. 
    For a size-$n$ instance of $\textup{MGM}\rbra{\kappa_A, \kappa_C}$, let $N = 2^n$ and let $A, C \in \mathbb{C}^{N \times N}$ be $O\rbra{1}$-sparse positive definite matrices with $I/\kappa_A\rbra{n} \leq A \leq I$ and $I/\kappa_C\rbra{n} \leq C \leq I$, given by a $\poly\rbra{n}$-size uniform classical circuit $\mathcal{C}_n$ such that, for every $1 \leq j \leq N$, the circuit $\mathcal{C}_n\rbra{j}$ computes the positions and values of the non-zero entries in the $j$-th row of $A$ and $C$.
    Let $Y \in \mathbb{C}^{N \times N}$ be the matrix geometric mean of $A$ and $C$ such that $YAY = C$.
    The task is to decide which of the following is the case, promised that one of the two holds:
    \begin{itemize}
        \item \emph{Yes}: $\langle \psi | M |\psi\rangle  \geq 2/3$;
        \item \emph{No}: $\langle \psi | M |\psi\rangle \leq 1/3$,
    \end{itemize}
    where  $\ket{\psi} \coloneqq \frac{Y^{2}\ket{0}}{\Abs{Y^2\ket{0}}}$ and $M = |0\rangle\!\langle 0| \otimes I_{N/2}$ measures the first qubit. 
\end{problem}

\begin{theorem} \label{thm:bqp-completeness}
    $\textup{MGM}\rbra{\poly\rbra{n}, \poly\rbra{n}}$ is $\mathsf{BQP}$-complete.
\end{theorem}

\begin{proof}
    The proof consists of two parts: Lemma~\ref{lemma:BQP-hard} and Lemma~\ref{lemma:BQP-containment}.
    \begin{enumerate}
        \item In Lemma~\ref{lemma:BQP-hard}, we state that $\textup{MGM}\rbra{\poly\rbra{n}, \poly\rbra{n}}$ is $\mathsf{BQP}$-hard; the proof employs a reduction of the Quantum Linear Systems Problem (QLSP). 
        \item In Lemma~\ref{lemma:BQP-containment}, we state that $\textup{MGM}\rbra{\poly\rbra{n}, \poly\rbra{n}}$ is in $\mathsf{BQP}$; the proof employs the quantum algorithm for the matrix geometric mean given in Lemma~\ref{lem:simplericcati}.
    \end{enumerate}
\end{proof}

\begin{lemma}
\label{lemma:BQP-hard}
    $\textup{MGM}\rbra{\poly\rbra{n}, \poly\rbra{n}}$ is $\mathsf{BQP}$-hard.
\end{lemma}

\begin{proof}
    We consider the Quantum Linear Systems Problem (QLSP) defined as follows.
    \begin{problem} [QLSP]
        For functions $\kappa \colon \mathbb{N} \to \mathbb{N}$, let $\textup{QLSP}\rbra{\kappa}$ be a decision problem defined as follows. 
        For a size-$n$ instance of $\textup{QLSP}\rbra{\kappa}$, let $N = 2^n$ and $A \in \mathbb{C}^{N \times N}$ be an $O\rbra{1}$-sparse Hermitian matrix such that $I/\kappa\rbra{n} \leq A \leq I$, given by a $\poly\rbra{n}$-size uniform classical circuit $\mathcal{C}_n$ such that for every $1 \leq j \leq N$, $\mathcal{C}_n\rbra{j}$ computes the positions and values of the non-zero entries in the $j$-th row of $A$. 
    The task is to decide which of the following is the case, promised that one of the two holds:
    \begin{itemize}
        \item \emph{Yes} item: $\langle \psi | M |\psi\rangle \geq 2/3$;
        \item \emph{No} item: $\langle \psi | M |\psi\rangle \leq 1/3$,
    \end{itemize}
    where $\ket{\psi} \coloneqq \frac{A^{-1}\ket{0}}{\Abs{A^{-1}\ket{0}}}$ and $M = |0\rangle\!\langle 0| \otimes I_{N/2}$ measures the first qubit. 
    \end{problem}

    It was shown in~\cite{harrow2009quantum} that $\textup{QLSP}\rbra{\poly\rbra{n}}$ is $\mathsf{BQP}$-complete. 
    Here, we reduce $\textup{QLSP}\rbra{\poly\rbra{n}}$ to $\textup{MGM}\rbra{\poly\rbra{n}, \poly\rbra{n}}$, and therefore show the $\mathsf{BQP}$-hardness of $\textup{MGM}\rbra{\poly\rbra{n}, \poly\rbra{n}}$. 
    
    Consider any instance (matrix) $A \in \mathbb{C}^{N \times N}$ of $\textup{QLSP}\rbra{\kappa}$, where $N = 2^n$ and $\kappa = \poly\rbra{n}$. 
    We choose $C = I \in \mathbb{C}^{N \times N}$ to be the identity matrix, which is a $1$-sparse Hermitian matrix and each of whose rows can be easily computed. 
    Note that the matrix geometric mean $Y$ of $A^{-1}$ and $C$ is $Y = A^{-1} \# C = A^{-1/2}$. 
    Then, it can be seen that $Y^2 = A^{-1}$ and thus $\ket{\psi_Y} = Y^2\ket{0}/\Abs{Y^2\ket{0}} = A^{-1}\ket{0}/\Abs{A^{-1}\ket{0}} = \ket{\psi_A}$. 
    Consequently, any quantum algorithm that determines whether $\bra{\psi_Y} M \ket{\psi_Y} \geq 2/3$ or $\bra{\psi_Y} M \ket{\psi_Y} \leq 1/3$ with success probability at least $2/3$ can be used to determine whether $\bra{\psi_A} M \ket{\psi_A} \geq 2/3$ or $\bra{\psi_A} M \ket{\psi_A} \leq 1/3$.
    In summary, $\textup{QLSP}\rbra{\kappa}$ can be reduced to $\textup{MGM}\rbra{\kappa, 1}$ through the above encoding. 
    Therefore, $\textup{MGM}\rbra{\poly\rbra{n}, \poly\rbra{n}}$ is $\mathsf{BQP}$-hard.
\end{proof}

\begin{lemma}
\label{lemma:BQP-containment}
    $\textup{MGM}\rbra{\poly\rbra{n}, \poly\rbra{n}}$ is in $\mathsf{BQP}$.
\end{lemma}
\begin{proof} See Appendix~\ref{app:lemmaMGM}.
\end{proof}

\section{Discussion}

\label{sec:discussion}

We constructed efficient block-encodings of the matrix geometric mean (and weighted matrix geometric mean). These are unique solutions to the simplest algebraic Riccati equations --  quadratically nonlinear system of matrix equations. Unlike the output of most quantum algorithms for linear systems of equations, these solutions of the nonlinear matrix equations are not embedded in pure quantum states, but rather in terms of observables from which we can extract expectation values. 

This allows us to introduce a new class of algorithms for quantum learning, called quantum geometric mean metric learning. For example, this can be applied in a purely quantum setting for picking out anomalous quantum states. This can also be adapted to the case of flexible weights on the cost of flagging an anomaly. The new quantum subroutines can also be used for the first quantum algorithm, to the best of our knowledge, to compute the geometric R\'enyi
relative entropies and new quantum algorithms to compute quantum fidelity by means of the Fuchs–Caves observable. In the latter case, we demonstrate optimal scaling $\Omega(1/\epsilon)$ in precision. 

While most of the applications introduced above are for quantum problems for which there is no direct classical equivalent (although the quantum learning algorithm can also be applied to learning Euclidean distances for classical data), there are potential benefits that the new quantum subroutine can have over purely classical methods. This could be exploited for future applications. For example, classical numerical algorithms to compute the matrix geometric mean have cost $O(\operatorname{poly} D)$ for $D \times D$ matrices~\cite{bini2011note, iannazzo2016geometric}. The same is also true for solving the differential matrix Riccati equation and algebraic matrix Riccati equation~\cite{ramesh1989computational} through iterative methods and other methods based on finding the eigendecomposition of a larger matrix~\cite{arnold1984generalized}. For quantum processing on the other hand, we showed conditions under which the block-encodings of some of these solutions can be obtained with cost $O(\operatorname{poly} \log d)$. 

For example, there are many classical problems for which it is important to compute the matrix geometric mean between two matrices. They appear in imaging~\cite{estatico2013shift, arsigny2007geometric} and in the analysis of multiport electrical networks~\cite{chansangiam2012operator}. The algebraic Riccati equation of the form in Eq.~\eqref{eq:yay1} also appears in optimal control and Kalman filters. Under the assumptions in Lemma~\ref{lem:yay2}  when uniqueness of its solution is also satisfied, it can be possible to construct its block-encoding in Lemma~\ref{lem:Briccati}. Although these assumptions are not generally satisfied, this still gives an idea of the extent and reach of the matrix geometric mean. Extensions of our algorithms to the matrix geometric mean consisting of more than two matrices can also be explored, which already find applications in areas like elasticity and radars~\cite{moakher2006averaging, barbaresco2009new}. It is also intriguing to consider purely quantum extensions of these problems. The main difficulty associated with constructing block-encodings of multivariate geometric means is that they are not known to have an analytical form as they do in the bivariate case; rather, they are constructed as the solutions of nonlinear equations generalizing the simple algebraic Riccati equation~\cite{lawson2021expanding}.
It is worth mentioning that there are many other quantum algorithms for learning problems with different loss functions and the solutions to these problems do not have general analytical forms. 
Examples include semi-definite programming \cite{BS17,vAGGdW20,BKL+19,vAG19b}, linear programming \cite{vAG19,BGJ+23,GJLW23}, and general matrix games \cite{LWCW21}. 
It is interesting to ask whether the techniques developed in this paper can be used in these problems.

In addition to usefulness in applications, the standard and weighted matrix geometric means  also have an elegant interpretation in terms of geodesics in Riemannian space. Despite the importance and beauty of Riemannian geometry in mathematics and other areas in physics, sensing, and machine learning, it has not appeared too much in quantum computation yet, apart from very notable exceptions like~\cite{nielsen2006quantum}. This geometric perspective is useful in understanding the weighted quantum learning algorithm, and we showed how it provided an alternative motivation for the form of quantum fidelity via the Fuchs--Caves observable. There is more potential here for the matrix geometric mean to bring the ideas of geometry closer to quantum information and computation.

\begin{acknowledgments}
N.L.\ acknowledges funding from the Science and Technology Commission of Shanghai Municipality (STCSM) grant no. 24LZ1401200 (21JC1402900). N.L.\ is also supported by NSFC grants No.\ 12471411 and No.\ 12341104, the Shanghai Jiao Tong University 2030 Initiative, and the Fundamental Research Funds for the Central Universities.
Q.W.\ acknowledges support from the Engineering and Physical Sciences Research Council under Grant No.~\mbox{EP/X026167/1} and the MEXT Quantum Leap Flagship Program (MEXT Q-LEAP) under Grant No.~JPMXS0120319794. 
M.M.W.\ acknowledges support from the NSF under grants 2329662 and 2315398.
Z.Z.\ acknowledges support from the Sydney Quantum Academy, NSW, Australia.
\end{acknowledgments}

\bibliography{main}

\appendix

\section{Proof of Lemma~\ref{lem:yay2}}

\label{app:proof-yay2}

This follows by observing that \eqref{eq:yay1} is a matrix
version of the quadratic equation and by following an argument similar to what
is well known as completing the square. Consider
that
\begin{align}
  \rbra*{  Y-A^{-1}B}^\dag  A\rbra*{  Y-A^{-1}B} 
&  =\rbra*{  Y^\dag -\rbra*{  A^{-1}B}  ^{\dag}}  A\rbra*{  Y-A^{-1}%
B} \\
&  =\rbra*{  Y^\dag -B^{\dag}\rbra*{  A^{-1}}  ^{\dag}}  A\rbra*{
Y-A^{-1}B} \\
&  =\rbra*{  Y^\dag -B^{\dag}A^{-1}}  A\rbra*{  Y-A^{-1}B} \\
&  =Y^\dag AY-Y^\dag AA^{-1}B-B^{\dag}A^{-1}AY+B^{\dag}A^{-1}AA^{-1}B\\
&  =Y^\dag AY-Y^\dag B-B^{\dag}Y+B^{\dag}A^{-1}B.
\end{align}
Then%
\begin{equation}
Y^\dag AY-B^{\dag}Y-Y^{\dag}B-C=\rbra*{  Y-A^{-1}B}^\dag   A\rbra*{  Y-A^{-1}B}
-B^{\dag}A^{-1}B-C,
\end{equation}
and so \eqref{eq:yay1} is equivalent to
\begin{equation}
X^\dag AX=D,
\end{equation}
where%
\begin{align}
X  &  =Y-A^{-1}B,\\
D  &  =B^{\dag}A^{-1}B+C.
\end{align}
Observe that $D$ is positive definite because $B^{\dag}A^{-1}B$ is positive
semi-definite and $C$ is positive definite. So this is a reduction to the
original simplified form of the algebraic Riccati equation in \eqref{eq:yayc}, which we know from Lemma~\ref{lem:yay1} has the following unique positive definite solution:
\begin{align}
X  
&  =A^{-1}\#D\\
&  =A^{-1}\#\rbra*{  B^{\dag}A^{-1}B+C}  .
\end{align}
This implies that%
\begin{align}
Y  &  =X+A^{-1}B\\
&  =A^{-1}\#\rbra*{  B^{\dag}A^{-1}B+C}  +A^{-1}B
\end{align}
is a solution of \eqref{eq:yay1}.

\begin{remark}
Contrary to what is stated in the proof of \cite[Corollary~4]{fujii2009riccati}, the solution of \eqref{eq:yay1}, under the assumptions stated in Lemma~\ref{lem:yay2}, is not unique. Indeed, 
    \begin{equation}
Y=-\rbra*{  A^{-1}\#\rbra*{  B^{\dag}A^{-1}B+C}  }  +A^{-1}B
\end{equation}
is also a legitimate solution. In
fact, the following is a matrix version of the famous quadratic formula:
\begin{equation}
A^{-1}B\pm\rbra*{  A^{-1}\#\rbra*{  B^{\dag}A^{-1}B+C}  }  ,
\end{equation}
for which the scalar version is $\frac{b}{a}\pm\sqrt{\frac{1}{a}\rbra*{
\frac{b^{2}}{a}+c}  }$ corresponding to a solution of $ay^{2}-2by-c=0$ (stated after \eqref{eq:yay1}),
under the assumption that $a,c>0$.
\end{remark}

\section{Preliminary lemmas of the block-encoding formalism and other useful results} \label{app:preliminaries}

Let us introduce several preliminary lemmas of the block-encoding formalism,
which enable us to implement various arithmetic operations on the block-encoded matrices.
The first lemma states that, given  block-encodings of two matrices,
we can obtain a block-encoding of their product.

\begin{lemma}[Product of block-encoded matrices {\cite[Lemma 30]{GSLW19}}]
	\label{lmm:product of block-encoding}
	If $U$ is an $(\alpha,a,\delta)$-block-encoding of $A$ and
	$V$ is a $(\beta,b,\epsilon)$-block-encoding of $B$,
	then there is a unitary $W$ that is an $(\alpha\beta, a+b,\alpha\epsilon+\beta\delta)$-block-encoding of $AB$,
	and can be implemented by one query to $U$ and~$V$.
\end{lemma}

Taking the linear combination of several block-encoded matrices is also useful and is stated in the following lemma.

\begin{lemma}[Linear combination of block-encoded matrices {\cite[Lemma 29]{GSLW19}}]
	\label{lmm:LCU}
	Let $m\in \mathbb{N}$ and $\beta>0$ be constant, and let $\boldsymbol{x}=(x_1,\ldots,x_m)\in \mathbb{R}^m$ be a vector such that $\lVert\boldsymbol{x}\rVert_1\leq \beta$.
	Suppose that each $U_j$ is a $(1, a, \epsilon)$-block-encoding of $A_j$ for $j=1$ to $m$.
	Then there is a unitary $U$ that is a $(1, a+\eta\log(1/\epsilon), 2\beta^{-1}\epsilon)$-block-encoding of $\beta^{-1}\sum_{j=1}^m x_j A_j$,
	where $\eta$ is some constant,
	and $U$ can be implemented by one query to each $U_j$ and $\polylog(1/\epsilon)$ gates.
\end{lemma}

To construct our quantum algorithms for matrix geometric means,
we need to deal with the non-linear terms in the matrix geometric means. 
The tool to be used is quantum singular value transformation~\cite{GSLW19},
which in our case is an (approximate) polynomial transformation of the block-encoded matrix,
as stated in the following lemma.

\begin{lemma}[Polynomial eigenvalue transformation {\cite[Theorem 31]{GSLW19}}]
	\label{lmm:svt}
	Let $U$ be a $(1,a,\epsilon)$-block-encoding of a Hermitian matrix $A$.
	If $\delta\geq 0$ and $q(x)\in \mathbb{R}[x]$ is a polynomial of degree~$d$ such that $\abs*{q(x)}\leq 1$ for $x\in [-1,1]$,
	then there is a unitary $\widetilde{U}$ that is a $(1,a+2,4d\sqrt{\epsilon}+\delta)$-block-encoding of $q(A)/2$,
	and can be implemented by $d$ queries to $U$
	and $O((a+1)d)$ gates.
	A description of such an implementation can be computed classically in time $O\rbra*{\poly\rbra*{d, \log(1/\delta)}}$.
\end{lemma}

We also need two polynomial approximation results for applying Lemma~\ref{lmm:svt}
in our scenario. The following two lemmas show low-degree polynomials for approximating the negative and positive power functions, respectively. 

\begin{lemma}[Polynomial approximations of negative power functions {\cite[Corollary 67 in the full version]{GSLW19}}]
	\label{lmm:poly negative power}
	Let $f(x)=\rbra*{x/\delta}^{-c}/2$.
	For $\delta,\epsilon\in (0,1/2)$ and $c>0$,
	there is a polynomial $q(x)$ of degree $O\rbra*{(c+1)\delta^{-1}\log(1/\epsilon)}$ such that
	\begin{itemize}
		\item 
			$\abs*{q(x)-f(x)}\leq \epsilon$ for $x\in [\delta,1]$;
		\item
			$\abs*{q(x)}\leq 1$ for $x\in [-1,\delta)$.
	\end{itemize}
\end{lemma}

\begin{lemma}[Polynomial approximations of positive power functions {\cite[Lemma 10]{CGJ19}}]
	\label{lmm:poly positive power}
	Let $f(x)=x^c/2$.
	For $\delta,\epsilon\in (0,1/2)$ and $c\in (0,1)$,
	there is a polynomial $q(x)$ of degree $O\rbra*{\delta^{-1}\log(1/\epsilon)}$ such that
	\begin{itemize}
		\item 
			$\abs*{q(x)-f(x)}\leq \epsilon$ for $x\in [\delta,1]$;
		\item
			$\abs*{q(x)}\leq 1$ for $x\in [-1,\delta)$.
	\end{itemize}
\end{lemma}

In practice, how to encode the desired matrices into block-encodings
and how to extract useful (classical) information from the block-encodings are of great concern.
For the encoding,
a typical scenario is that we are given sparse oracle access to a sparse matrix,
and we can construct a block-encoding of the matrix, as stated in the following lemma.

\begin{lemma} [Block-encoding of sparse matrices, {\cite[Lemma 48 in the full version]{GSLW19}}]
    \label{lemma:sparse-to-block}
    Suppose $A \in \mathbb{C}^{N \times N}$ is an $s$-sparse matrix such that every entry $A_{j,k}$ satisfies $\abs{A_{j,k}} \leq 1$.
    Suppose sparse oracles $\mathcal{O}_s$ and $\mathcal{O}_A$ are given such that
    \begin{align}
        \mathcal{O}_s \ket{j} \ket{k} & = \ket{j} \ket{l_{j, k}}, \\
        \mathcal{O}_A \ket{j} \ket{k} \ket{0} & = \ket{j} \ket{k} \ket{A_{j, k}},
    \end{align}
    where $l_{j, k}$ denotes the column index of the $k$-th non-zero entry in the $j$-th row.
    Here, we assume that the exact value of the entry $A_{j,k}$ is given in a binary representation. 
    Then, we can implement a quantum circuit that is an $(s, \log_2 N+3, \epsilon)$-block-encoding of $A$, using two queries to $\mathcal{O}_s$, two queries to $\mathcal{O}_A$, and $O(\log N+\log^{2.5}(s/\epsilon))$ one- and two-qubit quantum gates.
\end{lemma}

Another useful case of encoding commonly considered is that we are given purified access to a density operator,
and we can construct a block-encoding of the density operator.

\begin{lemma}[Block-encoding of density operators {\cite[Lemma 7]{low2019hamiltonian}}, {\cite[Lemma 25]{GSLW19}}]
	\label{lmm:purified to block-encoding}
	Let $\rho$ be an $n$-qubit density operator, and let $V_\rho$ be an $(n+a)$-qubit unitary
	that prepares a purification of $\rho$ such that $\tr_a(V_\rho\ket{0}_{n+a}\!\bra{0}V_\rho^\dagger)=\rho$.
	Then there is a $(2n+a)$-qubit unitary $\widetilde{V}$ that is a $(1,n+a,0)$-block-encoding of $\rho$,
	and it can be implemented by one query to $V$ and $O(n)$ gates.
\end{lemma}

To extract classical information from the block-encodings,
one needs to perform quantum measurements. 
The Hadamard test is a useful and efficient way to estimate the expectation value of a quantum observable on a given quantum state.
The following lemma shows a Hadamard test for block-encodings. 

\begin{lemma}[Hadamard test for block-encodings, {\cite[Lemma 9]{GP22}}] \label{lemma:hadamard-test}
    Suppose that $U$ is a unitary operator that is a $\rbra*{1, a, 0}$-block-encoding of an $n$-qubit operator $A$. 
    Then, there is a quantum circuit that outputs $0$ with probability $\frac{1+\Re{\Tr\rbra*{A\rho}}}{2}$, using one query to $U$ and one sample of the mixed quantum state $\rho$.
\end{lemma}

The success probability of extracting classical information from a block-encoding
often depends on the scaling factor of the block-encoding.
For our purpose, we need the following up-scaling lemma for block-encoded operators adapted from \cite{WZ23}.

\begin{lemma} [Up-scaling of block-encoded operators, adapted from {\cite[Lemma 2.8]{WZ23}}] \label{lemma:upscaling}
    Let unitary operator $U$ be an $\rbra{\alpha, a, \epsilon}$-block-encoding of $A$ with $\alpha=\Omega(1)$, $\epsilon \in \rbra{0, 1}$ and $\Abs{A} \leq 1$. 
    Then, 
    we can implement a quantum circuit $U'$ that is a $\rbra{2, a+1, \widetilde{\Theta}\rbra{\sqrt{\alpha\epsilon}}
    }$-block-encoding of $A$, using $\widetilde{O}\rbra{\alpha \log\rbra{{1}/{\epsilon}}}$ queries to $U$, $\widetilde{O}\rbra{a\cdot \alpha \log\rbra{{1}/{\epsilon}}}$ gates,
    and $\poly \rbra{\alpha,\log(1/\epsilon)}$ classical time.
\end{lemma}

Apart from the block-encoding formalism, 
let us introduce other useful results.
Quantum amplitude estimation allows one to estimate the amplitude of a specific component of a quantum state,
stated as follows.

\begin{lemma} [Quantum amplitude estimation {\cite[Theorem 12]{BHMT02}}] \label{lemma:amp-estimation}
    Suppose that unitary operator $U$ is given by
    \begin{equation}
        U \ket{0}\ket{0} = \sqrt{p} \ket{0} \ket{\phi_0} + \sqrt{1 - p} \ket{1} \ket{\phi_1},
    \end{equation}
    where $\ket{\phi_0}$ and $\ket{\phi_1}$ are normalized pure quantum states, and $p \in \sbra*{0, 1}$. 
    Then, we can obtain an estimate $\widetilde p$ of $p$ such that
    \begin{equation}
        \abs*{\widetilde p - p} \leq \frac{2\pi\sqrt{p(1-p)}}{M} + \frac{\pi^2}{M^2}
    \end{equation}
    with probability $\geq 8/\pi^2$ using $O\rbra*{M}$ queries to $U$. 
    In particular, if we take $M = \Theta\rbra{1/\delta}$, then $\widetilde p$ is a $\delta$-estimate of $p$ with high probability. 
\end{lemma}

Finally, the following two lemmas give relevant bounds on the condition number of matrices,
which will be useful in the complexity analysis of our algorithms.

\begin{lemma}
    \label{lmm:kappa_sum}
    Let $A,B> 0$ be two positive definite matrices such that $\left \| A\right\| = \left \| B\right\| = 1 $.
    Then $\kappa_{A+B}^{-1}\geq \kappa_A^{-1}+\kappa_B^{-1}$.
\end{lemma}

\begin{lemma}
    \label{lmm:kappa_product}
    Let $A>0$ be a positive definite matrix, and let $B$ be a matrix of full rank, such that $\left \| A\right\| = \left \| B\right\| = 1 $.
    Then $\kappa_{B^\dagger A B}\leq \kappa_A\kappa_{B^\dagger B}$.
\end{lemma}

\begin{proof}
    First note that $B^\dagger B>0$ and $B^\dagger C B>0$ for every $C>0$.
    Since $A\geq I/\kappa_A$,
    it follows that
    \begin{equation}
        B^\dagger A B\geq \kappa_A^{-1} B^\dagger I B\geq \kappa_A^{-1}\kappa_{B^\dagger B}^{-1} I.
    \end{equation}
    It immediately follows that $\kappa_{B^\dagger A B}^{-1}\geq \kappa_A^{-1}\kappa_{B^\dagger B}^{-1}$.
\end{proof}

\section{Proof of Lemma~\ref{lem:mgm-1}} \label{app:warmup-mgm}

In this appendix we prove Lemma~\ref{lem:mgm-1}.
Let us first prove the following lemma,
which shows that we can implement a block-encoding of the matrix $\rbra*{A^{-1/2} C A^{-1/2}}^{1/p}$.
\begin{lemma}
    \label{lmm:ACA-p-blk-enc}
    Suppose that $U_A,U_C$ are $(1,a,0)$-block-encodings of matrices $A,C$, respectively,
    such that $A\geq I/\kappa_A$ and $C\geq I/\kappa_C$.
    For $\epsilon\in (0,1/2)$,
    one can implement a $(2,3a+7,\epsilon)$-block-encoding of $\kappa_A^{-1/p}\gamma_p^{-1} \rbra*{A^{-1/2} C A^{-1/2}}^{1/p}$
    for any fixed real $p\neq 0$,
    where 
    \begin{equation}
        \gamma_p=\begin{cases}
            1 & p>0,\\
            \kappa_A^{-1/p}\kappa_C^{-1/p}& p<0, 
        \end{cases}
    \end{equation}
    using
    \begin{itemize}
        \item 
        $\widetilde{O}\rbra*{\kappa_A\kappa_C\log^2\rbra*{1/\epsilon}}$ queries to $U_C$,
        $\widetilde{O}\rbra*{\kappa_A^2\kappa_C\log^3\rbra*{1/\epsilon}}$ queries to $U_A$;
        \item
        $\widetilde{O}\rbra*{a\kappa_A^2\kappa_C\log^3\rbra*{1/\epsilon}}$ gates; and
        \item 
        $\poly\rbra*{\kappa_A,\kappa_C,\log\rbra*{1/\epsilon_2}}$ classical time.
    \end{itemize}
\end{lemma}

\begin{proof}
    We first consider the case $p>0$.
    Let us construct $U_{\kappa_A^{-1/p} \rbra*{A^{-1/2} C A^{-1/2}}^{1/p}}$,
    a block-encoding of $\kappa_A^{-1/p} \rbra*{A^{-1/2} C A^{-1/2}}^{1/p}$, 
    step by step as follows.
    Along the way, we also analyse the resources for each step.
    In the remainder of the paper,
    we use the notation $\widetilde{O}_{a_1,\ldots,a_n}(f)$ to denote $O(f\polylog(b_1,\ldots,b_n))$,
    where $a_i,b_i$ are parameters, $f$ is a function, and $b_i=a_i+a_i^{-1}$.
    Similarly, we use $\widetilde{\Omega}_{a_1,\ldots,a_n}(f)$ to denote $\Omega(f/\polylog(b_1,\ldots,b_n))$.
    In context without ambiguity, we just omit the subscripts as usual.
    \begin{enumerate}
        \item 
            $U_{A}\rightarrow U_{(\kappa_A A)^{-1/2}/4}$:
            \begin{itemize}
                \item
                    Construction:
                    \begin{enumerate}
                        \item 
                            Taking $c=1/2$, $\delta=\kappa_A^{-1}$, and $\epsilon=\epsilon_1$ in Lemma~\ref{lmm:poly negative power},
                            we have a polynomial $q_1(x)$ of degree $d_1=O\rbra*{\kappa_A\log(1/\epsilon_1)}$ that approximates $(\kappa_A x)^{-1/2}/2$.
                        \item
                            Taking $U=U_{A}$, $q=q_1(x)$, $\epsilon=0$, and $\delta=\epsilon_1$ in Lemma~\ref{lmm:svt},
                            we have $U_{\rbra*{\kappa_A A}^{-1/2}/4}$, a $(1,a+2,\epsilon_1)$-block-encoding of $q_1(A)/2$,
                            which is therefore a $(1,a+2,2\epsilon_1)$-block-encoding of $\rbra*{\kappa_A A}^{-1/2}/4$.
                    \end{enumerate}
                \item
                    Resources:
                    $O(d_1)$ queries to $U_{A}$,
                    $O(ad_1)$ quantum gates,
                    and $\poly(d_1, \log(1/\epsilon_1))$ classical time.
            \end{itemize}
        \item 
            $U_C,U_{(\kappa_A A)^{-1/2}/4}\rightarrow U_{2^{-4}\kappa_A^{-1} A^{-1/2}C A^{-1/2}}$:
            \begin{itemize}
                \item 
                    Construction:
                    by Lemma~\ref{lmm:product of block-encoding},
                    given $U_{C}$ and $U_{(\kappa_A A)^{-1/2}/4}$,
                    we have $U_{2^{-4}\kappa_A^{-1} A^{-1/2}C A^{-1/2}}$, a $(1,3a+4,4\epsilon_1)$-block-encoding of
                    $2^{-4}\kappa_A^{-1} A^{-1/2}C A^{-1/2}$.
                \item
                    Resources:
                    $O(1)$ queries to $U_{C}$ and $U_{(\kappa_A A)^{-1/2}/4}$.
            \end{itemize}
        \item
            $U_{2^{-4}\kappa_A^{-1} A^{-1/2}C A^{-1/2}}\rightarrow U_{\kappa_A^{-1/p} \rbra*{A^{-1/2} C A^{-1/2}}^{1/p}}$:
            \begin{itemize}
            \item 
                Construction:
                \begin{enumerate}
                    \item 
                        \label{stp:ACA-blk-enc-3a}
                        Taking $c=1/p$, $\delta=2^{-4}\kappa_A^{-1}\kappa_C^{-1}\leq2^{-4}\kappa_{A}^{-1}\kappa_{A^{-1/2}CA^{-1/2}}^{-1}$ 
                        (by Lemma~\ref{lmm:kappa_product} and noting $\kappa_{A^{-1}}\leq 1$),
                        and $\epsilon=\epsilon_2$ in Lemma~\ref{lmm:poly positive power},
                        we have a polynomial $q_2(x)$ of degree $d_2=O\rbra*{\kappa_A\kappa_C\log(1/\epsilon_2)}$ that approximates $x^{1/p}/2$.
                    \item
                        Taking $U=U_{2^{-4}\kappa_A^{-1} A^{-1/2}C A^{-1/2}}$, $q=q_2(x)$,
                        $\epsilon=\Theta(\epsilon_1)$, and $\delta=\epsilon_2$ in Lemma~\ref{lmm:svt},
                        we have $U_{2^{2-4/p}\kappa_A^{-1/p} \rbra*{A^{-1/2} C A^{-1/2}}^{1/p}}$,
                        a $(1,3a+6,\Theta(d_2\epsilon_1^{1/2})+\epsilon_2)$-block-encoding of $q_2\rbra*{ 2^{-4}\kappa_A^{-1} A^{-1/2}C A^{-1/2}}/2$,
                        which is therefore a $(1,3a+10,\Theta(d_2\epsilon_1^{1/2}+\epsilon_2))$-block-encoding of 
                        $\rbra*{2^{-4}\kappa_A^{-1} A^{-1/2}C A^{-1/2}}^{1/p}/4$.
                    \item
                        Taking $U=U_{2^{2-4/p}\kappa_A^{-1/p} \rbra*{A^{-1/2} C A^{-1/2}}^{1/p}}$,
                        $\alpha=2^{2-4/p}$ in Lemma~\ref{lemma:upscaling}, we  obtain $U_{\kappa_A^{-1/p} \rbra*{A^{-1/2} C A^{-1/2}}^{1/p}}$,
                        a $(2,3a+7, \widetilde{\Theta}(d_2^{1/2}\epsilon_1^{1/4}+\epsilon_2^{1/2}))$-block-encoding of
                        $\kappa_A^{-1/p} \rbra*{A^{-1/2} C A^{-1/2}}^{1/p}$,
                       where we use $\sqrt{x+y}\leq \sqrt{x}+\sqrt{y}$. 
                \end{enumerate}
            \item
                Resources:
                $\widetilde{O}\rbra*{d_2\log (\epsilon_1^{-1}\epsilon_2^{-1})}$ queries to $U_{\kappa_A^{-1} A^{-1/2}C A^{-1/2}}$,
                $\widetilde{O}\rbra*{ad_2\log (\epsilon_1^{-1}\epsilon_2^{-1})}$ gates,
                and $\poly\rbra*{d_2, \log(\epsilon_1^{-1}\epsilon_2^{-1})}$ classical time.
            \end{itemize}
    \end{enumerate}
   To bound the final approximation error $\widetilde{\Theta}(d_2^{1/2}\epsilon_1^{1/4}+\epsilon_2^{1/2})$
    in $U_{\kappa_A^{-1/p} \rbra*{A^{-1/2} C A^{-1/2}}^{1/p}}$ by $\epsilon$,
    it is sufficient to take 
    \begin{itemize}
        \item 
        $\epsilon_2=\widetilde{\Theta}\rbra*{\epsilon^2}$; and
        \item $\epsilon_1=\widetilde{\Theta}\rbra*{\epsilon^4 d_2^{-2}}=\widetilde{\Theta}\rbra*{\kappa_A^{-2}\kappa_C^{-2}\epsilon^4\log^{-2}\rbra*{1/\epsilon_2}}=\widetilde{\Theta}\rbra*{\kappa_A^{-2}\kappa_C^{-2}\epsilon^4}$.
    \end{itemize}
    The number of ancilla qubits for constructing $U_{\kappa_A^{-1/p} \rbra*{A^{-1/2} C A^{-1/2}}^{1/p}}$ is $3a+7$.
    
    Finally, let us calculate the complexities of each step.
    \begin{enumerate}
        \item 
        $U_{A}\rightarrow U_{(\kappa_A A)^{-1/2}/4}$:
        
        $\widetilde{O}_{\kappa_A,\kappa_C,\epsilon}\rbra*{\kappa_A\log\rbra*{1/\epsilon}}$ queries to $U_{A}$,
        $\widetilde{O}_{\kappa_A,\kappa_C,\epsilon}\rbra*{a\kappa_A\log\rbra*{1/\epsilon}}$ gates,
        
        and $\poly_{\kappa_A,\kappa_C,\epsilon}\rbra*{\kappa_A,\log\rbra*{1/\epsilon}}$ classical time.
        
        \item 
        $U_C,U_{(\kappa_A A)^{-1/2}/4}\rightarrow U_{2^{-4}\kappa_A^{-1} A^{-1/2}C A^{-1/2}}$:

        $O(1)$ queries to $U_{C}$,
        $\widetilde{O}_{\kappa_A,\kappa_C,\epsilon}\rbra*{\kappa_A\log\rbra*{1/\epsilon}}$  queries to $U_{A}$,
        $\widetilde{O}_{\kappa_A,\kappa_C,\epsilon}\rbra*{a\kappa_A\log\rbra*{1/\epsilon}}$ gates,
        and $\poly_{\kappa_A,\kappa_C,\epsilon}\rbra*{\kappa_A,\log\rbra*{1/\epsilon}}$ classical time.
        
        \item 
        $U_{2^{-4}\kappa_A^{-1} A^{-1/2}C A^{-1/2}}\rightarrow U_{\kappa_A^{-1/p} \rbra*{A^{-1/2} C A^{-1/2}}^{1/p}}$:

        $\widetilde{O}_{\kappa_A,\kappa_C,\epsilon}\rbra*{\kappa_A\kappa_C\log^2\rbra*{1/\epsilon}}$ queries to $U_C$,
        $\widetilde{O}_{\kappa_A,\kappa_C,\epsilon}\rbra*{\kappa_A^2\kappa_C\log^3\rbra*{1/\epsilon}}$ queries to $U_A$,
        
        $\widetilde{O}_{\kappa_A,\kappa_C,\epsilon}\rbra*{a\kappa_A^2\kappa_C\log^3\rbra*{1/\epsilon}}$ gates,
        and $\poly_{\kappa_A,\kappa_C,\epsilon}\rbra*{\kappa_A,\kappa_C,\log\rbra*{1/\epsilon}}$ classical time.
    \end{enumerate}

    For the case $p<0$, the analysis is the same except that
    in Step~\ref{stp:ACA-blk-enc-3a},
    we can use Lemma~\ref{lmm:poly negative power} instead of Lemma~\ref{lmm:poly positive power}.
    This only incurs an additional scaling factor $\kappa_A^{1/p}\kappa_C^{1/p}$ into the final block-encoded matrix,
    without significantly changing the complexity.
\end{proof}

Now we are ready to prove Lemma~\ref{lem:mgm-1},
which gives an implementation of a block-encoding of the weighted matrix geometric mean in Eq.~\eqref{eq:geometricmeans}.

\begin{proof}[Proof of Lemma~\ref{lem:mgm-1}]
    We first consider the case $p>0$.
    Let us construct $U_{\kappa_A^{-1/p}Y}$,
    a $(2,5a+12,\epsilon)$-block-encoding of $\kappa_A^{-1/p}Y$, step by step as follows,
    where
    \begin{equation}
        Y=A\#_{1/p} C=A^{1/2}\rbra*{A^{-1/2}CA^{-1/2}}^{1/p}A^{1/2}.
    \end{equation}
    Along the way, we also analyse the resources for each step.
    \begin{enumerate}
        \item 
            $U_{A}\rightarrow U_{A^{1/2}/4}$:
            \begin{itemize}
                \item 
                    Construction:
                    \begin{enumerate}
                        \item 
                            Taking $c=1/2$, $\delta=\kappa_A^{-1}$, and $\epsilon=\epsilon_1$ in Lemma~\ref{lmm:poly positive power},
                            we have a polynomial $q_1(x)$ of degree $d_1=O\rbra*{\kappa_A\log(1/\epsilon_1)}$ that approximates $x^{1/2}/2$.
                        \item
                            Taking $U=U_{A}$, $q=q_1(x)$, $\epsilon=0$, and $\delta=\epsilon_1$ in Lemma~\ref{lmm:svt},
                            we have $U_{A^{1/2}/4}$, a $(1,a+2,\epsilon_1)$-block-encoding of $q_1(A)/2$,
                            which is therefore a $(1,a+2,2\epsilon_1)$-block-encoding of $A^{1/2}/4$.
                    \end{enumerate}
                \item
                    Resources:
                    $O(d_1)$ queries to $U_{A}$,
                    $O(ad_1)$ quantum gates,
                    and $\poly(d_1, \log(1/\epsilon_1))$ classical time.
            \end{itemize}
        \item
            $U_A,U_C\rightarrow U_{\kappa_A^{-1/p} \rbra*{A^{-1/2} C A^{-1/2}}^{1/p}}$:
            \begin{itemize}
                \item 
                    Construction:

                    Taking $\epsilon= \epsilon_2$ in Lemma~\ref{lmm:ACA-p-blk-enc},
                    we can construct $U_{\kappa_A^{-1/p} \rbra*{A^{-1/2} C A^{-1/2}}^{1/p}}$,
                    a $(2,3a+7,\epsilon_2)$-block-encoding of $U_{\kappa_A^{-1/p} \rbra*{A^{-1/2} C A^{-1/2}}^{1/p}}$.
                \item
                    Resources:

                    According to Lemma~\ref{lmm:ACA-p-blk-enc}, the resources for the above construction are:
                    \begin{itemize}
                        \item 
                        $\widetilde{O}\rbra*{\kappa_A\kappa_C\log^2\rbra*{1/\epsilon_2}}$ queries to $U_C$,
                        $\widetilde{O}\rbra*{\kappa_A^2\kappa_C\log^3\rbra*{1/\epsilon_2}}$ queries to $U_A$;
                        \item
                        $\widetilde{O}\rbra*{a\kappa_A^2\kappa_C\log^3\rbra*{1/\epsilon_2}}$ gates; and
                        \item 
                        $\poly\rbra*{\kappa_A,\kappa_C,\log\rbra*{1/\epsilon_2}}$ classical time.
                    \end{itemize}
            \end{itemize}
        \item
            $U_{A^{1/2}/4},U_{\kappa_A^{-1/p} \rbra*{A^{-1/2} C A^{-1/2}}^{1/p}}\rightarrow U_{\kappa_A^{-1/p}Y}$:
            \begin{itemize}
                \item 
                    Construction:
                    \begin{enumerate}
                        \item 
                            By Lemma~\ref{lmm:product of block-encoding},
                            given $U_{A^{1/2}/4}$ and $U_{\kappa_A^{-1/p} \rbra*{A^{-1/2} C A^{-1/2}}^{1/p}}$,
                            we have $U_{2^{-5}\kappa_A^{-1/p}Y}$, a $(1,5a+11, \Theta(\epsilon_1+\epsilon_2))$-block-encoding of
                            $2^{-5}\kappa_A^{-1/p}Y$.
                        \item
                           
                            Taking $U_{2^{-5}\kappa_A^{-1/p}Y}$,
                            $\alpha=2^{5}$ in Lemma~\ref{lemma:upscaling}, we  obtain $U_{\kappa_A^{-1/p}Y}$,
                            a $(2,5a+12, \widetilde{\Theta}(\epsilon_1^{1/2}+\epsilon_2^{1/2}))$-block-encoding of
                            $\kappa_A^{-1/p}Y$. 
                    \end{enumerate}
                \item
                    Resources:
                   $\widetilde{O}\rbra*{\log(\epsilon_1^{-1}\epsilon_2^{-1})}$ queries to $U_{A^{1/2}/4}$ and $U_{\kappa_A^{-1/p} \rbra*{A^{-1/2} C A^{-1/2}}^{1/p}}$,
                    
                    $\widetilde{O}\rbra*{a\log(\epsilon_1^{-1}\epsilon_2^{-1})}$ gates,
                    and $\poly(\log(\epsilon_1^{-1}\epsilon_2^{-1}))$ classical time. 
            \end{itemize}
    \end{enumerate}
   To bound the final approximation error $\widetilde{\Theta}(\epsilon_1^{1/2}+\epsilon_2^{1/2})$
    in $U_{\kappa_A^{-1/p}Y}$ by $\epsilon$,
    it is sufficient to take 
    $\epsilon_1=\widetilde{\Theta}\rbra*{\epsilon^2}$ and $\epsilon_2=\widetilde{\Theta}\rbra*{\epsilon^2}$.
    The number of ancilla qubits for constructing $\kappa_A^{-1/p}Y$ is $5a+11$.
    Finally, let us calculate the complexities of each step.
    \begin{enumerate}
        \item 
        $U_{A}\rightarrow U_{A^{1/2}/4}$:
        
        $\widetilde{O}_{\kappa_A,\kappa_C,\epsilon}\rbra*{\kappa_A\log\rbra*{1/\epsilon}}$ queries to $U_{A}$,
        $\widetilde{O}_{\kappa_A,\kappa_C,\epsilon}\rbra*{a\kappa_A\log\rbra*{1/\epsilon}}$ gates,
        
        and $\poly_{\kappa_A,\kappa_C,\epsilon}\rbra*{\kappa_A,\log\rbra*{1/\epsilon}}$ classical time.
        
        \item 
        $U_A,U_C\rightarrow U_{\kappa_A^{-1/p} \rbra*{A^{-1/2} C A^{-1/2}}^{1/p}}$:

        $\widetilde{O}\rbra*{\kappa_A\kappa_C\log^2\rbra*{1/\epsilon}}$ queries to $U_C$,
        $\widetilde{O}\rbra*{\kappa_A^2\kappa_C\log^3\rbra*{1/\epsilon}}$ queries to $U_A$,
        $\widetilde{O}\rbra*{a\kappa_A^2\kappa_C\log^3\rbra*{1/\epsilon}}$ gates, and
        $\poly\rbra*{\kappa_A,\kappa_C,\log\rbra*{1/\epsilon_2}}$ classical time.
        
        \item 
        $U_{A^{1/2}/4},U_{\kappa_A^{-1/p} \rbra*{A^{-1/2} C A^{-1/2}}^{1/p}}\rightarrow U_{\kappa_A^{-1/p}Y}$:

        $\widetilde{O}\rbra*{\kappa_A\kappa_C\log^3\rbra*{1/\epsilon}}$ queries to $U_C$,
        $\widetilde{O}\rbra*{\kappa_A^2\kappa_C\log^4\rbra*{1/\epsilon}}$ queries to $U_A$,
        $\widetilde{O}\rbra*{a\kappa_A^2\kappa_C\log^4\rbra*{1/\epsilon}}$ gates, and
        $\poly\rbra*{\kappa_A,\kappa_C,\log\rbra*{1/\epsilon_2}}$ classical time.
    \end{enumerate}

    By Definition~\ref{def:blk-enc}, $U_{\kappa_A^{-1/p}Y}$ is also
    a $(2\kappa_A^{1/p},5a+12,\kappa_A^{1/p}\epsilon)$-block-encoding of $Y$.
    Replacing the precision parameter immediately yields the results for the case $p>0$ in Lemma~\ref{lem:mgm-1}. 
    For the case $p<0$, the analysis is similar and is omitted.
\end{proof}

\section{Proof of Lemma~\ref{lem:simplericcati}} \label{app:proofb=0}

In this appendix we prove Lemma~\ref{lem:simplericcati}.

\begin{proof} [Proof of Lemma~\ref{lem:simplericcati}]
    Let us construct $U_{\kappa_A^{-1}Y}$,
    a $(2,5a+11,\epsilon)$-block-encoding of $\kappa_A^{-1}Y$, step by step as follows,
    where 
    \begin{equation}
        Y=A^{-1}\# C=A^{-1/2}\rbra*{A^{1/2}CA^{1/2}}^{1/2}A^{-1/2}.
    \end{equation}
    Along the way, we also analyse the resources for each step.
    \begin{enumerate}
        \item 
        $U_{A}\rightarrow U_{A^{1/2}/4}$:
        \begin{itemize}
            \item 
                Construction:
                \begin{enumerate}
                    \item 
                        Taking $c=1/2$, $\delta=\kappa_A^{-1}$, and $\epsilon=\epsilon_1$ in Lemma~\ref{lmm:poly positive power},
                        we have a polynomial $q_1(x)$ of degree $d_1=O\rbra*{\kappa_A\log(1/\epsilon_1)}$ that approximates $x^{1/2}/2$.
                    \item
                        Taking $U=U_{A}$, $q=q_1(x)$, $\epsilon=0$, and $\delta=\epsilon_1$ in Lemma~\ref{lmm:svt},
                        we have $U_{A^{1/2}/4}$, a $(1,a+2,\epsilon_1)$-block-encoding of $q_1(A)/2$,
                        which is therefore a $(1,a+2,2\epsilon_1)$-block-encoding of $A^{1/2}/4$.
                \end{enumerate}
            \item
                Resources:
                $O(d_1)$ queries to $U_{A}$,
                and $O(a d_1)$ gates,
                and $\poly(d_1, \log(1/\epsilon_1))$ classical time.
        \end{itemize}
        \item
		$U_{C},U_{A^{1/2}/4}\rightarrow U_{2^{-4}A^{1/2}C A^{1/2}}$:
		\begin{itemize}
            \item 
                Construction:
                by Lemma~\ref{lmm:product of block-encoding},
                given $U_{C}$ and $U_{A^{1/2}/4}$,
                we have $U_{2^{-4}A^{1/2}CA^{1/2}}$, a $(1,3a+4, 4\epsilon_1)$-block-encoding of
                $2^{-4}A^{1/2}CA^{1/2}$.
            \item
                Resources:
                $O(1)$ queries to $U_{C}$ and $U_{A^{1/2}}$.
        \end{itemize}
        \item
        $U_{2^{-4}A^{1/2}CA^{1/2}}\rightarrow U_{2^{-4}\rbra*{A^{1/2}CA^{1/2}}^{1/2}}$:
        \begin{itemize}
            \item 
                Construction:
                \begin{enumerate}
                    \item 
                        \label{stp:B=0-blk-enc-3a}
                        Taking $c=1/2$, $\delta=2^{-4}\kappa_A^{-1}\kappa_C^{-1}\leq2^{-4}\kappa_{A^{1/2}CA^{1/2}}^{-1}$ (by Lemma~\ref{lmm:kappa_product}),
                        and $\epsilon=\epsilon_2$ in Lemma~\ref{lmm:poly positive power},
                        we have a polynomial $q_2(x)$ of degree $d_2=O\rbra*{\kappa_A\kappa_C\log(1/\epsilon_2)}$ that approximates $x^{1/2}/2$.
                    \item
                        Taking $U=U_{2^{-4}A^{1/2}CA^{1/2}}$, $q=q_2(x)$,
                        $\epsilon=\Theta(\epsilon_1)$, and $\delta=\epsilon_2$ in Lemma~\ref{lmm:svt},
                        we have $U_{2^{-4}\rbra*{A^{1/2}CA^{1/2}}^{1/2}}$,
                        a $(1,3a+6,\Theta(d_2\epsilon_1^{1/2})+\epsilon_2)$-block-encoding of $q_2\rbra*{2^{-4}A^{1/2}CA^{1/2}}/2$,
                        which is therefore a $(1,3a+6,\Theta(d_2\epsilon_1^{1/2}+\epsilon_2))$-block-encoding of 
                        $\rbra*{2^{-4}A^{1/2}CA^{1/2}}^{1/2}/4$.
                \end{enumerate}
            \item
                Resources:
                $O(d_2)$ queries to $U_{2^{-4}A^{1/2}CA^{1/2}}$,
                and $O(ad_2)$ gates,
                and $\poly(d_2, \log(1/\epsilon_2))$ classical time.
        \end{itemize}
        \item
        $U_{A}\rightarrow U_{(\kappa_A A)^{-1/2}/4}$:
        \begin{itemize}
            \item
                Construction:
                \begin{enumerate}
                    \item 
                        Taking $c=1/2$, $\delta=\kappa_A^{-1}$, and $\epsilon=\epsilon_3$ in Lemma~\ref{lmm:poly negative power},
                        we have a polynomial $q_3(x)$ of degree $d_3=O\rbra*{\kappa_A\log(1/\epsilon_3)}$ that approximates $(\kappa_A x)^{-1/2}/2$.
                    \item
                        Taking $U=U_{A}$, $q=q_3(x)$, $\epsilon=0$, and $\delta=\epsilon_3$ in Lemma~\ref{lmm:svt},
                        we have $U_{\rbra*{\kappa_A A}^{-1/2}/4}$, a $(1,a+2,\epsilon_3)$-block-encoding of $q_3(A)/2$,
                        which is therefore a $(1,a+2,2\epsilon_3)$-block-encoding of $\rbra*{\kappa_A A}^{-1/2}/4$.
                \end{enumerate}
            \item
                Resources:
                $O(d_3)$ queries to $U_{A}$,
                and $O(ad_3)$ gates,
                and $\poly(d_3, \log(1/\epsilon_3))$ classical time.
        \end{itemize}
        \item
        $U_{(\kappa_A A)^{-1/2}/4},U_{2^{-4}\rbra*{A^{1/2}CA^{1/2}}^{1/2}}\rightarrow U_{\kappa_A^{-1}Y}$:
        \begin{itemize}
            \item 
                Construction:
                \begin{enumerate}
                    \item 
                    By Lemma~\ref{lmm:product of block-encoding},
                    given $U_{(\kappa_A A)^{-1/2}/4}$ and $U_{2^{-4}\rbra*{A^{1/2}CA^{1/2}}^{1/2}}$,
                    we have $U_{2^{-8}\kappa_A^{-1}Y}$,
                    a $(1, 5a+10,\Theta(d_2\epsilon_1^{1/2}+\epsilon_2+\epsilon_3))$-block-encoding 
                    of $2^{-8}\kappa_A^{-1}Y$.
                \item
                   
                    Taking $U_{2^{-8}\kappa_A^{-1}Y}$,
                    $\alpha=2^{8}$ in Lemma~\ref{lemma:upscaling}, we  obtain $U_{\kappa_A^{-1}Y}$,
                    a $(2,5a+11, \widetilde{\Theta}(d_2^{1/2}\epsilon_1^{1/4}+\epsilon_2^{1/2}+\epsilon_3^{1/2}))$-block-encoding of
                    $\kappa_A^{-1/p}Y$. 
                \end{enumerate}
            \item
                Resources:
               $\widetilde{O}\rbra*{\log(\epsilon_1^{-1}\epsilon_2^{-1}\epsilon_3^{-1})}$ queries to $U_{(\kappa_A A)^{-1/2}/4}$ and $U_{2^{-4}\rbra*{A^{1/2}CA^{1/2}}^{1/2}}$,
                
                $\widetilde{O}\rbra*{a\log(\epsilon_1^{-1}\epsilon_2^{-1}\epsilon_3^{-1})}$ gates,
                and $\poly(\log(\epsilon_1^{-1}\epsilon_2^{-1}\epsilon_3^{-1}))$ classical time. 
        \end{itemize}
    \end{enumerate}
   To bound the final approximation error $\widetilde{\Theta}(d_2^{1/2}\epsilon_1^{1/4}+\epsilon_2^{1/2}+\epsilon_3^{1/2})$
    in $U_{\kappa_A^{-1}Y}$ by $\epsilon$,
    it is sufficient to take 
    \begin{itemize}
        \item 
        $\epsilon_3=\widetilde{\Theta}(\epsilon^2)$.
        \item 
        $\epsilon_2=\widetilde{\Theta}(\epsilon^2)$.
        \item 
        $\epsilon_1=\widetilde{\Theta}(\epsilon^4 d_2^{-2})=\widetilde{\Theta}\rbra*{\kappa_A^{-2}\kappa_C^{-2}\epsilon^4\log^{-2}\rbra*{1/\epsilon_2}}=\widetilde{\Theta}\rbra*{\kappa_A^{-2}\kappa_C^{-2}\epsilon^4}$.
    \end{itemize}
    The number of ancilla qubits for constructing $U_{\kappa_A^{-1}Y}$ is $5a+11$.

    Finally, let us calculate the complexities of each step.
    \begin{enumerate}
        \item 
        $U_{A}\rightarrow U_{A^{1/2}/4}$:
        
        $\widetilde{O}_{\kappa_A,\kappa_C,\epsilon}\rbra*{\kappa_A\log\rbra*{1/\epsilon}}$ queries to $U_{A}$,
        $\widetilde{O}_{\kappa_A,\kappa_C,\epsilon}\rbra*{a\kappa_A\log\rbra*{1/\epsilon}}$ gates,
        
        and $\poly_{\kappa_A,\kappa_C,\epsilon}\rbra*{\kappa_A,\log\rbra*{1/\epsilon}}$ classical time.
        
        \item 
        $U_{C},U_{A^{1/2}/4}\rightarrow U_{2^{-4}A^{1/2}C A^{1/2}}$:

        $O(1)$ queries to $U_{C}$,
        $\widetilde{O}_{\kappa_A,\kappa_C,\epsilon}\rbra*{\kappa_A\log\rbra*{1/\epsilon}}$ queries to $U_{A}$,
        $\widetilde{O}_{\kappa_A,\kappa_C,\epsilon}\rbra*{a\kappa_A\log\rbra*{1/\epsilon}}$ gates,
        and $\poly_{\kappa_A,\kappa_C,\epsilon}\rbra*{\kappa_A,\log\rbra*{1/\epsilon}}$ classical time.
        
        \item 
        $U_{2^{-4}A^{1/2}CA^{1/2}}\rightarrow U_{2^{-4}\rbra*{A^{1/2}CA^{1/2}}^{1/2}}$:

        $\widetilde{O}_{\kappa_A,\kappa_C,\epsilon}\rbra*{\kappa_A\kappa_C\log\rbra*{1/\epsilon}}$ queries to $U_C$,
        $\widetilde{O}_{\kappa_A,\kappa_C,\epsilon}\rbra*{\kappa_A^2\kappa_C\log^2\rbra*{1/\epsilon}}$ queries to $U_A$,
        \newline $\widetilde{O}_{\kappa_A,\kappa_C,\epsilon}\rbra*{a\kappa_A^2\kappa_C\log^2\rbra*{1/\epsilon}}$ gates,
        and $\poly_{\kappa_A,\kappa_C,\epsilon}\rbra*{\kappa_A,\kappa_C,\log\rbra*{1/\epsilon}}$ classical time.
        
        \item 
        $U_{A}\rightarrow U_{(\kappa_A A)^{-1/2}/4}$:

        $\widetilde{O}_{\kappa_A,\kappa_C,\epsilon}\rbra*{\kappa_A\log\rbra*{1/\epsilon}}$ queries to $U_{A}$,
        $\widetilde{O}_{\kappa_A,\kappa_C,\epsilon}\rbra*{a\kappa_A\log\rbra*{1/\epsilon}}$ gates,
        and $\poly\rbra*{\kappa_A, \log\rbra*{1/\epsilon}}$ classical time.

        \item
        $U_{(\kappa_A A)^{-1/2}/4},U_{2^{-4}\rbra*{A^{1/2}CA^{1/2}}^{1/2}}\rightarrow U_{\kappa_A^{-1}Y}$:
        
        $\widetilde{O}_{\kappa_A,\kappa_C,\epsilon}\rbra*{\kappa_A\kappa_C\log^2\rbra*{1/\epsilon}}$ queries to $U_C$,
        $\widetilde{O}_{\kappa_A,\kappa_C,\epsilon}\rbra*{\kappa_A^2\kappa_C\log^3\rbra*{1/\epsilon}}$ queries to $U_A$,
        \newline $\widetilde{O}_{\kappa_A,\kappa_C,\epsilon}\rbra*{a\kappa_A^2\kappa_C\log^3\rbra*{1/\epsilon}}$ gates,
        and $\poly_{\kappa_A,\kappa_C,\epsilon}\rbra*{\kappa_A,\kappa_C,\log\rbra*{1/\epsilon}}$ classical time.
    \end{enumerate}
    
    By Definition~\ref{def:blk-enc}, $U_{\kappa_A^{-1}Y}$ is also
    a $(2\kappa_A,5a+11,\kappa_A\epsilon)$-block-encoding of $Y$.
    Replacing the precision parameter immediately yields the results in Lemma~\ref{lem:simplericcati}.
\end{proof}

\section{Proof of Lemma~\ref{lem:Briccati}} \label{app:proofbneq0}

In this appendix we prove Lemma~\ref{lem:Briccati}.

\begin{proof} [Proof of Lemma~\ref{lem:Briccati}]
    Let $D=B^\dagger A^{-1}B +C$
    and $E=A^{1/2}DA^{1/2}$.
    As $B^\dagger A^{-1}B\geq 0$,
    we have 
    \begin{equation}
        \kappa_{D}^{-1}\geq \kappa_C^{-1}\label{eq:kappa_D}
    \end{equation}
    by Lemma~\ref{lmm:kappa_sum}.
    
    Similar to the proof of Lemma~\ref{lem:simplericcati} in Appendix~\ref{app:proofb=0},
    let us construct $U_{\kappa_A^{-3/2}Y}$,
    a $(2,b,\epsilon)$-block-encoding of $\kappa_A^{-3/2}Y$, step by step as follows,
    where $b=O\rbra*{a+\log\rbra*{\kappa_A\kappa_C/\epsilon}}$ and
    \begin{equation}
        Y=A^{-1}\# D+A^{-1}B
        =A^{-1/2}E^{1/2}A^{-1/2}+A^{-1}B.
    \end{equation}
    Along the way, we also analyse the resources for each step.
    \begin{enumerate}
        \item 
        \label{stp:Bneq0-blk-enc-inverse}
         $U_{A}\rightarrow U_{(\kappa_A A)^{-1}/4}$:
        \begin{itemize}
            \item
                Construction:
                \begin{enumerate}
                    \item 
                        Taking $c=1$, $\delta=\kappa_A^{-1}$, and $\epsilon=\epsilon_1$ in Lemma~\ref{lmm:poly negative power},
                        we have a polynomial $q_1(x)$ of degree $d_1=O\rbra*{\kappa_A\log(1/\epsilon_1)}$ that approximates $(\kappa_A x)^{-1}/2$.
                    \item
                        Taking $U=U_{A}$, $q=q_1(x)$, $\epsilon=0$, and $\delta=\epsilon_1$ in Lemma~\ref{lmm:svt},
                        we have $U_{\rbra*{\kappa_A A}^{-1}/4}$, a $(1,a+2,\epsilon_1)$-block-encoding of $q_1(A)/2$,
                        which is therefore a $(1,a+2,2\epsilon_1)$-block-encoding of $\rbra*{\kappa_A A}^{-1}/4$.
                \end{enumerate}
            \item
                Resources:
                $O(d_1)$ queries to $U_{A}$,
                and $O(ad_1)$ gates,
                and $\poly(d_1, \log(1/\epsilon_1))$ classical time.
        \end{itemize}

        \item 
        $U_B,U_{\rbra*{\kappa_A A}^{-1}/4}\rightarrow U_{2^{-2}\kappa_A^{-1}B^\dagger A^{-1}B}$:

		\begin{itemize}
            \item 
                Construction:
                Note that given $U_B$,
                one can construct $U_{B^\dagger}=U_{B}^\dagger$, a $(1,a,0)$-block-encoding of $B^\dagger$,
                using $1$ query to $U_B$.
                By Lemma~\ref{lmm:product of block-encoding},
                given $U_{B}$ and $U_{\rbra*{\kappa_A A}^{-1}/4}$,
                we have $U_{\kappa_A^{-1}B^\dagger A^{-1}B}$, a $(1,3a+2, 2\epsilon_1)$-block-encoding of
                $2^{-2}\kappa_A^{-1}B^\dagger A^{-1}B$.
            \item
                Resources:
                $O(1)$ queries to $U_{B}$ and $U_{\rbra*{\kappa_A A}^{-1}/4}$.
        \end{itemize}

        \item
        $U_C,U_{2^{-2}\kappa_A^{-1}B^\dagger A^{-1}B}\rightarrow U_{2^{-3}\kappa_A^{-1}D}$:

        \begin{itemize}
            \item 
                Construction:
                Taking $m=2$, $\boldsymbol{x}=(1,2^{-2}\kappa_A^{-1})$, $\beta=2$,
                $U_1=U_{2^{-2}\kappa_A^{-1}B^\dagger A^{-1}B}$, $U_2=U_C$ and $\epsilon=\Theta(\epsilon_1)$
                in Lemma~\ref{lmm:LCU},
                we obtain $U_{2^{-3}\kappa_A^{-1}D}$, a $\rbra*{1,3a+2+\eta_1\log\rbra*{1/\epsilon_1},\Theta(\epsilon_1)}$-block-encoding of $2^{-3}\kappa_A^{-1}D$,
                for some constant $\eta_1$.
            \item
                Resources:
                $O(1)$ queries to $U_C$ and $U_{2^{-2}\kappa_A^{-1}B^\dagger A^{-1}B}$,
                and $\polylog\rbra*{1/\epsilon_1}$ gates.
        \end{itemize}

        \item
        \label{stp:Bneq0-blk-enc-sqrt}
        $U_{A}\rightarrow U_{A^{1/2}/4}$:
        \begin{itemize}
            \item 
                Construction:
                \begin{enumerate}
                    \item 
                        Taking $c=1/2$, $\delta=\kappa_A^{-1}$, and $\epsilon=\epsilon_2$ in Lemma~\ref{lmm:poly positive power},
                        we have a polynomial $q_2(x)$ of degree $d_2=O\rbra*{\kappa_A\log(1/\epsilon_2)}$ that approximates $x^{1/2}/2$.
                    \item
                        Taking $U=U_{A}$, $q=q_2(x)$, $\epsilon=0$, and $\delta=\epsilon_2$ in Lemma~\ref{lmm:svt},
                        we have $U_{A^{1/2}/4}$, a $(1,a+2,\epsilon_2)$-block-encoding of $q_2(A)/2$,
                        which is therefore a $(1,a+2,2\epsilon_2)$-block-encoding of $A^{1/2}/4$.
                \end{enumerate}
            \item
                Resources:
                $O(d_2)$ queries to $U_{A}$,
                and $O(a d_2)$ gates,
                and $\poly(d_2, \log(1/\epsilon_2))$ classical time.
        \end{itemize}
        
        \item 
        $U_{2^{-3}\kappa_A^{-1}D},U_{A^{1/2}/4}\rightarrow U_{2^{-7}\kappa_A^{-1}E}$:
        \begin{itemize}
            \item 
                Construction:
                By Lemma~\ref{lmm:product of block-encoding},
                given $U_{2^{-3}\kappa_A^{-1}D}$ and $U_{A^{1/2}/4}$,
                we have $U_{2^{-7}\kappa_A^{-1}E}$, a $(1,5a+6+\eta_1\log\rbra*{1/\epsilon_1}, \Theta(\epsilon_1+\epsilon_2))$-block-encoding of
                $2^{-7}\kappa_A^{-1}E$.
            \item
                Resources:
                $O(1)$ queries to $U_{2^{-3}\kappa_A^{-1}D}$ and $U_{A^{1/2}/4}$.
        \end{itemize}

        \item
        $U_{2^{-7}\kappa_A^{-1}E}\rightarrow U_{2^{-5.5}\kappa_A^{-1/2}E^{1/2}}$:
        \begin{itemize}
            \item 
                Construction:
                \begin{enumerate}
                    \item 
                        Taking $c=1/2$, $\delta=2^{-7}\kappa_A^{-2}\kappa_C^{-1}\leq 2^{-7}\kappa_A^{-2}\kappa_D^{-1}\leq 2^{-7}\kappa_A^{-1}\kappa_{E}^{-1}$ (by Lemma~\ref{eq:kappa_D} and Lemma~\ref{lmm:kappa_product}), and $\epsilon=\epsilon_3$ in Lemma~\ref{lmm:poly positive power},
                        we have a polynomial $q_3(x)$ of degree $d_3=O\rbra*{\kappa_A\kappa_C\log(1/\epsilon_3)}$ that approximates $x^{1/2}/2$.
                    \item
                        Taking $U=U_{2^{-7}\kappa_A^{-1}E}$, $q=q_3(x)$, $\epsilon=\Theta(\epsilon_1+\epsilon_2)$, and $\delta=\epsilon_3$ in Lemma~\ref{lmm:svt},
                        we have $U_{2^{-2.5}\kappa_A^{-1/2}E^{1/2}}$,
                        a $(1,5a+8+\eta_1\log\rbra*{1/\epsilon_1},\Theta(d_3\epsilon_1^{1/2}+d_3\epsilon_2^{1/2})+\epsilon_3)$-block-encoding of $q_3\rbra*{2^{-7}\kappa_A^{-1}E}/2$,
                        which is therefore
                        
                        a $(1,5a+8+\eta_1\log\rbra*{1/\epsilon_1},\Theta(d_3\epsilon_1^{1/2}+d_3\epsilon_2^{1/2}+\epsilon_3))$-block-encoding of $2^{-5.5}\kappa_A^{-1/2}E^{1/2}$.
                \end{enumerate}
            \item
                Resources:
                $O(d_3)$ queries to $U_{2^{-7}\kappa_A^{-1}E}$,
                and $O\rbra*{\rbra*{a+\log\rbra*{1/\epsilon_1}}d_3}$ gates,
                and $\poly(d_3, \log(1/\epsilon_3))$ classical time.
        \end{itemize}
        
        \item 
         $U_{A}\rightarrow U_{(\kappa_A A)^{-1/2}/4}$:
        \begin{itemize}
            \item
                Construction:
                \begin{enumerate}
                    \item 
                        Taking $c=1/2$, $\delta=\kappa_A^{-1}$, and $\epsilon=\epsilon_4$ in Lemma~\ref{lmm:poly negative power},
                        we have a polynomial $q_4(x)$ of degree $d_4=O\rbra*{\kappa_A\log(1/\epsilon_4)}$ that approximates $(\kappa_A x)^{-1/2}/2$.
                    \item
                        Taking $U=U_{A}$, $q=q_4(x)$, $\epsilon=0$, and $\delta=\epsilon_4$ in Lemma~\ref{lmm:svt},
                        we have $U_{\rbra*{\kappa_A A}^{-1/2}/4}$, a $(1,a+2,\epsilon_4)$-block-encoding of $q_4(A)/2$,
                        which is therefore a $(1,a+2,2\epsilon_4)$-block-encoding of $\rbra*{\kappa_A A}^{-1/2}/4$.
                \end{enumerate}
            \item
                Resources:
                $O(d_4)$ queries to $U_{A}$,
                and $O(ad_4)$ gates,
                and $\poly(d_4, \log(1/\epsilon_4))$ classical time.
        \end{itemize}
        
        \item 
        $U_{(\kappa_A A)^{-1/2}/4},U_{2^{-5.5}\kappa_A^{-1/2}E^{1/2}}\rightarrow U_{2^{-9.5}\kappa_A^{-3/2}A^{-1}\# D}$:

        \begin{itemize}
            \item 
                Construction:
                By Lemma~\ref{lmm:product of block-encoding},
                given $U_{(\kappa_A A)^{-1/2}/4}$ and $U_{2^{-5.5}\kappa_A^{-1/2}E^{1/2}}$,
                we have $U_{\kappa_A^{-3/2}A^{-1}\# D}$,
                a $(1,7a+12+\eta_1\log\rbra*{1/\epsilon_1},\Theta(d_3\epsilon_1^{1/2}+d_3\epsilon_2^{1/2}+\epsilon_3+\epsilon_4))$-block-encoding
                of $2^{-9.5}\kappa_A^{-3/2}A^{-1}\# D$.
            \item
                Resources:
                $O(1)$ queries to $U_{(\kappa_A A)^{-1/2}/4}$ and $U_{2^{-5.5}\kappa_A^{-1/2}E^{1/2}}$.
        \end{itemize}
        
        \item 
        $U_{(\kappa_A A)^{-1}/4},U_{B}\rightarrow U_{2^{-2}\kappa_A^{-1}A^{-1}B}$:
        
        \begin{itemize}
            \item 
                Construction:
                By Lemma~\ref{lmm:product of block-encoding},
                given $U_{(\kappa_A A)^{-1}/4}$ and $U_{B}$,
                we have $U_{2^{-2}\kappa_A^{-1}A^{-1} B}$,
                a $(1,2a+2,2\epsilon_1)$-block-encoding
                of $2^{-2}\kappa_A^{-1}A^{-1}B$.
            \item
                Resources:
                $O(1)$ queries to $U_{(\kappa_A A)^{-1}/4}$ and $U_{B}$.
        \end{itemize}
        
        \item 
        $U_{2^{-9.5}\kappa_A^{-3/2}A^{-1}\# D},U_{2^{-2}\kappa_A^{-1}A^{-1}B}\rightarrow U_{\kappa_A^{-3/2}Y}$:
        
        \begin{itemize}
            \item 
                Construction:
                \begin{enumerate}
                    \item 
                    Taking $m=2$, $\boldsymbol{x}=(1,2^{-7.5}\kappa_A^{-1/2})$, $\beta=2$,
                    $U_1=U_{2^{-9.5}\kappa_A^{-3/2}A^{-1}\# D}$, $U_2=U_{2^{-2}\kappa_A^{-1}A^{-1}B}$ and 
                    $\epsilon=\epsilon_5=\Theta(d_3\epsilon_1^{1/2}+d_3\epsilon_2^{1/2}+\epsilon_3+\epsilon_4)$
                    in Lemma~\ref{lmm:LCU},
                    we obtain $U_{2^{-10.5}\kappa_A^{-3/2}Y}$,
                    a $\rbra*{1,7a+12+\eta_1\log\rbra*{1/\epsilon_1}+\eta_2\log\rbra*{1/\epsilon_5},\epsilon_5}$-block-encoding of $2^{-10.5}\kappa_A^{-3/2}Y$,
                    for some constant $\eta_2$.
                    
                    \item 
                    Taking $U=U_{2^{-10.5}\kappa_A^{-3/2}Y}$ and $\alpha=2^{10.5}$ in Lemma~\ref{lemma:upscaling},
                    we obtain $U_{\kappa_A^{-3/2}Y}$,
                    a $\rbra*{2,7a+13+\eta_1\log\rbra*{1/\epsilon_1}+\eta_2\log\rbra*{1/\epsilon_5},\widetilde{\Theta}(\epsilon_5^{1/2})}$-block-encoding of $\kappa_A^{-3/2}Y$.
                \end{enumerate}
            \item
                Resources:
                $\widetilde{O}\rbra*{\log\rbra*{\epsilon_5^{-1}}}$ queries to $U_{2^{-9.5}\kappa_A^{-3/2}A^{-1}\# D}$ and $U_{2^{-2}\kappa_A^{-1}A^{-1}B}$, 
                
               $\widetilde{O}\rbra*{a\log^2\rbra*{\epsilon_1^{-1}\epsilon_5^{-1}}}$ gates
                and $\poly(\log(\epsilon_5^{-1}))$ classical time. 
        \end{itemize}
    \end{enumerate}
   To bound the final approximation error $\widetilde{\Theta}(\epsilon_5^{1/2})$
    in $U_{\kappa_A^{-3/2}Y}$ by $\epsilon$,
    it is sufficient to take 
    \begin{itemize}
        \item 
        $\epsilon_5=\widetilde{\Theta}(\epsilon^2)$.
        \item 
        $\epsilon_4=\epsilon_3=\widetilde{\Theta}(\epsilon^2)$.
        \item 
        $\epsilon_2=\epsilon_1=\widetilde{\Theta}(d_3^{-2}\epsilon^4)=\widetilde{\Theta}\rbra*{\kappa_A^{-2}\kappa_C^{-2}\epsilon^4\log^{-2}\rbra*{1/\epsilon_3}}=\widetilde{\Theta}\rbra*{\kappa_A^{-2}\kappa_C^{-2}\epsilon^4}$.
    \end{itemize}
    The number of ancilla qubits for constructing $U_{\kappa_A^{-3/2}Y}$ is 
    \begin{equation}
        7a+13+\eta_1\log\rbra*{1/\epsilon_1}+\eta_2\log\rbra*{1/\epsilon_5}=O\rbra*{a+\log\rbra*{\kappa_A\kappa_C/\epsilon}}.
    \end{equation}
    Finally, let us calculate the complexities of each step.
    \begin{enumerate}
        \item 
        $U_{A}\rightarrow U_{(\kappa_A A)^{-1}/4}$:
        
        $\widetilde{O}_{\kappa_A,\kappa_C,\epsilon}\rbra*{\kappa_A\log\rbra*{1/\epsilon}}$ queries to $U_{A}$,
        $\widetilde{O}_{\kappa_A,\kappa_C,\epsilon}\rbra*{a\kappa_A\log\rbra*{1/\epsilon}}$ gates,
        
        and $\poly_{\kappa_A,\kappa_C,\epsilon}\rbra*{\kappa_A,\log\rbra*{1/\epsilon}}$ classical time.
        
        \item
        $U_B,U_{\rbra*{\kappa_A A}^{-1}/4}\rightarrow U_{2^{-2}\kappa_A^{-1}B^\dagger A^{-1}B}$:

        $O(1)$ queries to $U_B$,
        $\widetilde{O}_{\kappa_A,\kappa_C,\epsilon}\rbra*{\kappa_A\log\rbra*{1/\epsilon}}$ queries to $U_{A}$,
        $\widetilde{O}_{\kappa_A,\kappa_C,\epsilon}\rbra*{a\kappa_A\log\rbra*{1/\epsilon}}$ gates,
        and $\poly_{\kappa_A,\kappa_C,\epsilon}\rbra*{\kappa_A,\log\rbra*{1/\epsilon}}$ classical time.

        \item
        $U_C,U_{2^{-2}\kappa_A^{-1}B^\dagger A^{-1}B}\rightarrow U_{2^{-3}\kappa_A^{-1}D}$:
        
        $O(1)$ queries to $U_C$ and $U_B$,
        $\widetilde{O}_{\kappa_A,\kappa_C,\epsilon}\rbra*{\kappa_A\log\rbra*{1/\epsilon}}$ queries to $U_{A}$,
        $\widetilde{O}_{\kappa_A,\kappa_C,\epsilon}\rbra*{a\kappa_A\log\rbra*{1/\epsilon}}$ gates,
        and $\poly_{\kappa_A,\kappa_C,\epsilon}\rbra*{\kappa_A,\log\rbra*{1/\epsilon}}$ classical time.

        \item 
        $U_{A}\rightarrow U_{A^{1/2}/4}$:
        
        $\widetilde{O}_{\kappa_A,\kappa_C,\epsilon}\rbra*{\kappa_A\log\rbra*{1/\epsilon}}$ queries to $U_{A}$,
        $\widetilde{O}_{\kappa_A,\kappa_C,\epsilon}\rbra*{a\kappa_A\log\rbra*{1/\epsilon}}$ gates,
        
        and $\poly_{\kappa_A,\kappa_C,\epsilon}\rbra*{\kappa_A,\log\rbra*{1/\epsilon}}$ classical time.

        \item
        $U_{2^{-3}\kappa_A^{-1}D},U_{A^{1/2}/4}\rightarrow U_{2^{-7}\kappa_A^{-1}E}$:
        
        $O(1)$ queries to $U_C$ and $U_B$,
        $\widetilde{O}_{\kappa_A,\kappa_C,\epsilon}\rbra*{\kappa_A\log\rbra*{1/\epsilon}}$ queries to $U_{A}$,
        $\widetilde{O}_{\kappa_A,\kappa_C,\epsilon}\rbra*{a\kappa_A\log\rbra*{1/\epsilon}}$ gates,
        and $\poly_{\kappa_A,\kappa_C,\epsilon}\rbra*{\kappa_A,\log\rbra*{1/\epsilon}}$ classical time.
        
        \item
        $U_{2^{-7}\kappa_A^{-1}E}\rightarrow U_{2^{-5.5}\kappa_A^{-1/2}E^{1/2}}$:
        
        $\widetilde{O}_{\kappa_A,\kappa_C,\epsilon}\rbra*{\kappa_A\kappa_C\log\rbra*{1/\epsilon}}$ queries to $U_C$ and $U_B$,
        $\widetilde{O}_{\kappa_A,\kappa_C,\epsilon}\rbra*{\kappa_A^2\kappa_C\log^2\rbra*{1/\epsilon}}$ queries to $U_{A}$,
        
        $\widetilde{O}_{\kappa_A,\kappa_C,\epsilon}\rbra*{a\kappa_A^2\kappa_C\log^2\rbra*{1/\epsilon}}$ gates,
        and $\poly_{\kappa_A,\kappa_C,\epsilon}\rbra*{\kappa_A,\kappa_C,\log\rbra*{1/\epsilon}}$ classical time.
        
        \item 
        $U_{A}\rightarrow U_{(\kappa_A A)^{-1/2}/4}$:
        
        $\widetilde{O}_{\kappa_A,\kappa_C,\epsilon}\rbra*{\kappa_A\log\rbra*{1/\epsilon}}$ queries to $U_{A}$,
        $\widetilde{O}_{\kappa_A,\kappa_C,\epsilon}\rbra*{a\kappa_A\log\rbra*{1/\epsilon}}$ gates,
        and $\poly\rbra*{\kappa_A,\log\rbra*{1/\epsilon}}$ classical time.
        
        \item 
        $U_{(\kappa_A A)^{-1/2}/4},U_{2^{-5.5}\kappa_A^{-1/2}E^{1/2}}\rightarrow U_{2^{-9.5}\kappa_A^{-3/2}A^{-1}\# D}$:
        
        $\widetilde{O}_{\kappa_A,\kappa_C,\epsilon}\rbra*{\kappa_A\kappa_C\log\rbra*{1/\epsilon}}$ queries to $U_C$ and $U_B$,
        $\widetilde{O}_{\kappa_A,\kappa_C,\epsilon}\rbra*{\kappa_A^2\kappa_C\log^2\rbra*{1/\epsilon}}$ queries to $U_{A}$,
        
        $\widetilde{O}_{\kappa_A,\kappa_C,\epsilon}\rbra*{a\kappa_A^2\kappa_C\log^2\rbra*{1/\epsilon}}$ gates,
        and $\poly_{\kappa_A,\kappa_C,\epsilon}\rbra*{\kappa_A,\kappa_C,\log\rbra*{1/\epsilon}}$ classical time.
        
        \item 
        $U_{(\kappa_A A)^{-1}/4},U_{B}\rightarrow U_{2^{-2}\kappa_A^{-1}A^{-1}B}$:

        $O(1)$ queries to $U_B$,
        $\widetilde{O}_{\kappa_A,\kappa_C,\epsilon}\rbra*{\kappa_A\log\rbra*{1/\epsilon}}$ queries to $U_{A}$,
        $\widetilde{O}_{\kappa_A,\kappa_C,\epsilon}\rbra*{a\kappa_A\log\rbra*{1/\epsilon}}$ gates,
        and $\poly_{\kappa_A,\kappa_C,\epsilon}\rbra*{\kappa_A,\log\rbra*{1/\epsilon}}$ classical time.
        
        \item 
        $U_{2^{-9.5}\kappa_A^{-3/2}A^{-1}\# D},U_{2^{-2}\kappa_A^{-1}A^{-1}B}\rightarrow U_{\kappa_A^{-3/2}Y}$:
        
        $\widetilde{O}_{\kappa_A,\kappa_C,\epsilon}\rbra*{\kappa_A\kappa_C\log^2\rbra*{1/\epsilon}}$ queries to $U_C$ and $U_B$,
        $\widetilde{O}_{\kappa_A,\kappa_C,\epsilon}\rbra*{\kappa_A^2\kappa_C\log^3\rbra*{1/\epsilon}}$ queries to $U_{A}$,
        
        $\widetilde{O}_{\kappa_A,\kappa_C,\epsilon}\rbra*{a\kappa_A^2\kappa_C\log^3\rbra*{1/\epsilon}}$ gates,
        and $\poly_{\kappa_A,\kappa_C,\epsilon}\rbra*{\kappa_A,\kappa_C,\log\rbra*{1/\epsilon}}$ classical time.
    \end{enumerate}
    By Definition~\ref{def:blk-enc}, $U_{\kappa_A^{-3/2}Y}$ is also
    a $(2\kappa_A^{3/2},b,\kappa_A^{3/2}\epsilon)$-block-encoding of $Y$.
    Replacing the precision parameter immediately yields the results in Lemma~\ref{lem:Briccati}.
\end{proof}

\section{Proof of Lemma~\ref{lem:high-order-riccati}}

\label{app:high-order-riccati}

In this appendix we prove Lemma~\ref{lem:high-order-riccati}.

\begin{proof} [Proof of Lemma~\ref{lem:high-order-riccati}]
    The proof is similar to that of Lemma~\ref{lem:simplericcati}.
    For $p>0$, we can simply take $c=1/p$ instead in Step~\ref{stp:B=0-blk-enc-3a}, without significantly changing the complexity.
    
    For $p<0$, in Step~\ref{stp:B=0-blk-enc-3a},
    we can use Lemma \ref{lmm:poly negative power} instead of Lemma \ref{lmm:poly positive power},
    and take $c=-1/p$. This only incurs an additional scaling factor $\kappa_A^{1/p}\kappa_C^{1/p}$ into the final block-encoded matrix,
    without significantly changing the complexity.
\end{proof}

\section{Proof of Lemma~\ref{lem:lyoptimise}} \label{app:convexity}
   Although in~\cite{zadeh2016geometric} the lemma was stated only for real, symmetric positive definite matrices (SPDs), each step in the proof is also applicable to positive definite Hermitian matrices as we show below.\\

     To find the global minimum of $L(Y)$, it is sufficient to find the solution to $\nabla L(Y)=0$ when $L(Y)$ is strictly convex and is also strictly geodesically convex on the manifold of positive definite Hermitian matrices.  For the definition of distances on this manifold and the geometric interpretation for the matrix geometric mean, see Section~\ref{sec:introgeometricmeans}. \\

     The strict convexity of $Y\mapsto L(Y)$ can be proved for the two terms separately since strict convexity is preserved in a sum. The term $\operatorname{Tr}(YA)$ is clearly strictly convex since it is linear and $A$ is positive definite. For strict convexity of the second term $\operatorname{Tr}(Y^{-1}C)$, it follows directly from the fact that $Y \to Y^{-1}$ is strictly operator convex and $C$ is positive definite. As an alternative proof, we evoke the following relationship. It is known that a twice-differentiable function $L: V \rightarrow \mathbf{R}$ on an open subset $\mathcal{Y}$ of a vector space $\mathcal{Z}$ is convex if and only if for all $Y \in \mathcal{Y}$ and $Z \in \mathcal{Z}$
     \begin{align} \label{eq:twicediff}
         \frac{d^2 L(Y+t Z)}{dt^2} \vert_{t=0}>0.
     \end{align}
    Using the Woodbury matrix identity, we can rewrite 
    \begin{align}
         (Y+tZ)^{-1} & =(Y(I+t Y^{-1}Z))^{-1}=Y^{-1}-tY^{-1}(I+tZY^{-1})^{-1}ZY^{-1} \nonumber \\
        & =Y^{-1}-tY^{-1}ZY^{-1}+t^2 Y^{-1} Z Y^{-1} Z Y^{-1}+O(t^3).
    \end{align}
    Therefore the condition in Eq.~\eqref{eq:twicediff} for $\operatorname{Tr}(Y^{-1}C)$ is equivalent to showing that 
    \begin{align}
        \operatorname{Tr}(DC)>0, \qquad D=Y^{-1} Z Y^{-1} Z Y^{-1}.
    \end{align}
   $Y$ is positive definite Hermitian and let $Z$ be Hermitian so $D=D^{\dagger}$. We note that $DC$ is similar to the matrix $D^{-1/2}(DC)D^{1/2}=D^{1/2}CD^{1/2}$, so they have identical eigenvalues. Then it suffices to show that $D^{1/2}CD^{1/2}$ only has positive eigenvalues. Since $D^{1/2}$ is also Hermitian and $C$ is positive definite Hermitian, then $\operatorname{Tr}(D^{1/2}CD^{1/2})=\operatorname{Tr}(CD)>0$.  
    
    By strictly geodesically convex, it means that for all positive definite Hermitian matrices $Y_1, Y_2$, we have 
    \begin{align*}
        L(Y_1 \#_t Y_2) < t L(Y_1)+(1-t)L(Y_2), \qquad t \in [0,1].
    \end{align*}
    To show geodesic convexity, we also need the following two facts. From~\cite{ando1979concavity}, there is the fundamental operator inequality for positive definite matrices for $t \in [0,1]$
    \begin{align}
        Y_1 \#_t Y_2 \leq (1-t)Y_1+tY_2. 
    \end{align}
    For $Y_1 \neq Y_2$ for $t=1/2$, this is a strict inequality. 
    From the definition, it can also be shown that~\cite{bhatia2009positive}
    \begin{align}
        (Y_1 \#_t Y_2)^{-1}=Y_1^{-1} \#_t Y_2^{-1}.
    \end{align}
    Since midpoint convexity (convexity at $t=1/2$) and continuity imply convexity, we have 
\begin{align}
     L(Y_1 \#_{1/2} Y_2) & =\operatorname{Tr}((Y_1 \#_{1/2} Y_2)A)+\operatorname{Tr}((Y_1 \#_{1/2} Y_2)^{-1}C) \nonumber \\
    &<\frac{1}{2}(\operatorname{Tr}(Y_1 A)+ \operatorname{Tr}(Y_2 A))+\operatorname{Tr}((Y_1 \#_{1/2} Y_2)^{-1}C) \nonumber \\
    &=\frac{1}{2}(\operatorname{Tr}(Y_1 A)+ \operatorname{Tr}(Y_2 A))+\operatorname{Tr}((Y_1^{-1} \#_{1/2} Y_2^{-1})C) \nonumber \\
    & < \frac{1}{2}(\operatorname{Tr}(Y_1 A)+\operatorname{Tr}(Y_2 A)+ \operatorname{Tr}(Y_1^{-1}C)+\operatorname{Tr}(Y_2^{-1}C)) \nonumber \\
    &=\frac{1}{2}(L(Y_2)+L(Y_1)).
\end{align}
    Thus we have strict geodesic convexity.

\section{Proof of Theorems~\ref{thm:hat-F-alpha} and~\ref{thm:g-Renyi-relative-entropy}} \label{app:prooftheorem18}
Here we provide details of the proofs of Theorems~\ref{thm:hat-F-alpha} and~\ref{thm:g-Renyi-relative-entropy}.
\begin{proof} [Proof of Theorem~\ref{thm:hat-F-alpha}]
    By Eq.~\eqref{eq:def-hat-F-alpha}, the definition of $\widehat{F}_\alpha\rbra{\rho, \sigma}$, we have
    \begin{equation} \label{eq:hat-F-alpha-by-tr}
        \widehat{F}_\alpha\rbra{\rho, \sigma} = \Tr\rbra*{\rho \rbra*{\rho^{-1/2} \sigma \rho^{-1/2}}^{1-\alpha}} = \Tr\rbra*{\sigma \rbra*{\sigma^{-1/2} \rho \sigma^{-1/2}}^{\alpha}}.
    \end{equation}
    Suppose that $\rho$ and $\sigma$ are $n$-qubit mixed quantum states and $\mathcal{O}_\rho$ and $\mathcal{O}_\sigma$ are $\rbra{n+a}$-qubit unitary operators. 
    By Lemma~\ref{lmm:purified to block-encoding}, we can implement two unitary operators $U_{\rho}$ and $U_{\sigma}$ that are $\rbra{1, n+a, 0}$-block-encodings of $\rho$ and $\sigma$ using $O\rbra{1}$ queries to $\mathcal{O}_\rho$ and $\mathcal{O}_\sigma$, respectively. 
    We consider two approaches via the first and second formulas in Eq.~\eqref{eq:hat-F-alpha-by-tr} separately. 

    \textbf{Via the first formula}. By Lemma~\ref{lmm:ACA-p-blk-enc}, we can implement a $\rbra{2, b, \delta}$-block-encoding $W$ of $\kappa_\rho^{\alpha-1} \gamma_{\alpha} \rbra{\rho^{-1/2} \sigma \rho^{-1/2}}^{1-\alpha}$ using $\widetilde O\rbra{\kappa_\rho^2 \kappa_\sigma \log^2\rbra{1/\delta}}$ queries to $U_\rho$ and $\widetilde O\rbra{\kappa_\rho \kappa_\sigma \log\rbra{1/\delta}}$ queries to $U_\sigma$, where $b = 3a+7$, and 
    \begin{equation}
        \gamma_\alpha = 
        \begin{cases}
            1, & \alpha \in \rbra{0, 1}, \\
            \kappa_\rho^{1-\alpha} \kappa_\sigma^{1-\alpha}, & \alpha \in (1, 2].
        \end{cases}
    \end{equation}

    By the Hadamard test (given in Lemma~\ref{lemma:hadamard-test}), there is a quantum circuit $C$ that outputs $0$ with probability $\frac{1}{2}\rbra{ 1+\Re{\Tr\rbra{\rho \bra{0}_b W \ket{0}_b}} }$, using one query to $W$ and one sample of~$\rho$.
    By noting that
    \begin{equation}
        \abs*{ 2\kappa_\rho^{1-\alpha} \gamma_\alpha^{-1} \Re{\Tr\rbra{\rho \bra{0}_b W \ket{0}_b}} - \widehat{F}_\alpha\rbra{\rho, \sigma} } \leq \Theta\rbra{\kappa_\rho^{1-\alpha}\gamma_\alpha^{-1}\delta},
    \end{equation}
    we conclude that an $O\rbra{\kappa_\rho^{\alpha-1}\gamma_\alpha\epsilon}$-estimate of $\Re{\Tr\rbra{\rho \bra{0}_b U_M \ket{0}_b}}$ with $\delta = \Theta\rbra{\kappa_\rho^{\alpha-1}\gamma_\alpha\epsilon}$ suffices to obtain an $\epsilon$-estimate of $\widehat{F}_\alpha\rbra{\rho, \sigma}$. 
    By quantum amplitude estimation (given in Lemma~\ref{lemma:amp-estimation}), this can be done using $O\rbra{1/\delta} = O\rbra{\kappa_\rho^{1-\alpha}\gamma_\alpha^{-1}\epsilon^{-1}}$ queries to $C$. 
   
    To conclude, an $\epsilon$-estimate of $\widehat{F}_\alpha\rbra{\rho, \sigma}$ can be obtained by using 
    \begin{equation}
        \widetilde O\rbra{\kappa_\rho^2 \kappa_\sigma \log^2\rbra{1/\delta}} \cdot O\rbra{\kappa_\rho^{1-\alpha}\gamma_\alpha^{-1}\epsilon^{-1}} = \begin{cases}
            \widetilde O\rbra{\kappa_\rho^{3-\alpha} \kappa_\sigma/\epsilon}, & \alpha \in \rbra{0, 1}, \\
            \widetilde O\rbra{\kappa_\rho^2\kappa_\sigma^{\alpha}/\epsilon}, & \alpha \in (1, 2],
        \end{cases}
    \end{equation}
    queries to $\mathcal{O}_\rho$ and 
    \begin{equation}
        \widetilde O\rbra{\kappa_\rho \kappa_\sigma \log\rbra{1/\delta}} \cdot O\rbra{\kappa_\rho^{1-\alpha}\gamma_\alpha^{-1}\epsilon^{-1}} = \begin{cases}
            \widetilde O\rbra{\kappa_\rho^{2-\alpha} \kappa_\sigma/\epsilon}, & \alpha \in \rbra{0, 1}, \\
            \widetilde O\rbra{\kappa_\rho\kappa_\sigma^{\alpha}/\epsilon}, & \alpha \in (1, 2],
        \end{cases}
    \end{equation}
    queries to $\mathcal{O}_\sigma$.

    \textbf{Via the second formula}.
    By Lemma~\ref{lmm:ACA-p-blk-enc}, we can implement a $\rbra{2, b, \delta}$-block-encoding $W$ of $\kappa_\sigma^{-\alpha} \rbra{\sigma^{-1/2} \rho \sigma^{-1/2}}^{\alpha}$ using $\widetilde O\rbra{\kappa_\sigma^2 \kappa_\rho \log^2\rbra{1/\delta}}$ queries to $U_\sigma$ and $\widetilde O\rbra{\kappa_\rho \kappa_\sigma \log\rbra{1/\delta}}$ queries to $U_\rho$, where $b = 3a+7$.
    By the Hadamard test (given in Lemma~\ref{lemma:hadamard-test}), there is a quantum circuit $C$ that outputs $0$ with probability $\frac{1}{2}\rbra{ 1+\Re{\Tr\rbra{\rho \bra{0}_b W \ket{0}_b}} }$, using one query to $W$ and one sample of~$\rho$.
    By noting that
    \begin{equation}
        \abs*{ 2\kappa_\sigma^{\alpha}  \Re{\Tr\rbra{\rho \bra{0}_b W \ket{0}_b}} - \widehat{F}_\alpha\rbra{\rho, \sigma} } \leq \Theta\rbra{\kappa_\sigma^{\alpha}\delta},
    \end{equation}
    we conclude that an $O\rbra{\kappa_\sigma^{-\alpha}\epsilon}$-estimate of $\Re{\Tr\rbra{\rho \bra{0}_b U_M \ket{0}_b}}$ with $\delta = \Theta\rbra{\kappa_\sigma^{-\alpha}\epsilon}$ suffices to obtain an $\epsilon$-estimate of $\widehat{F}_\alpha\rbra{\rho, \sigma}$. 
    By quantum amplitude estimation (given in Lemma~\ref{lemma:amp-estimation}), this can be done using $O\rbra{1/\delta} = O\rbra{\kappa_\sigma^{\alpha}\epsilon^{-1}}$ queries to $C$. 
    
    To conclude, an $\epsilon$-estimate of $\widehat{F}_\alpha\rbra{\rho, \sigma}$ can be obtained by using $\widetilde O\rbra{\kappa_\sigma^2 \kappa_\rho \log^2\rbra{1/\delta}} \cdot O\rbra{1/\delta} = \widetilde O\rbra{\kappa_\sigma^{2+\alpha}\kappa_\rho/\epsilon}$ queries to $\mathcal{O}_\sigma$ and $\widetilde O\rbra{\kappa_\sigma \kappa_\rho \log\rbra{1/\delta}} \cdot O\rbra{1/\delta} = \widetilde O\rbra{\kappa_\sigma^{1+\alpha}\kappa_\rho/\epsilon}$ queries to $\mathcal{O}_\rho$.

    \textbf{Conclusion}. Combining the above cases (and their symmetrical cases), the query complexity is 
    \begin{itemize}
        \item $\widetilde O\rbra{\kappa_\rho\kappa_\sigma/\epsilon \cdot \min\cbra{\kappa_\rho, \kappa_\sigma}^{\min\cbra{1+\alpha, 2-\alpha}}}$ for $\alpha \in (0, 1)$,
        \item $\widetilde O\rbra{\kappa_\rho\kappa_\sigma/\epsilon \cdot \min\cbra{\kappa_\rho\kappa_\sigma^{\alpha-1}, \kappa_\rho^{\alpha-1}\kappa_\sigma, \kappa_{\rho}^{1+\alpha}, \kappa_\sigma^{1+\alpha}}}$ for $\alpha \in (1, 2]$.
    \end{itemize}
\end{proof}

\begin{proof} [Proof of Theorem~\ref{thm:g-Renyi-relative-entropy}]
    Note that $I/\kappa_\sigma \leq \rho^{-1/2} \sigma \rho^{-1/2} \leq \kappa_\rho I$. 
    Thus $\kappa_\sigma^{\alpha-1}I \leq \rbra{\rho^{-1/2} \sigma \rho^{-1/2}}^{1-\alpha} \leq \kappa_\rho^{1-\alpha} I$ for $\alpha \in (0, 1)$ and $\kappa_\rho^{1-\alpha} I \leq \rbra{\rho^{-1/2} \sigma \rho^{-1/2}}^{1-\alpha} \leq \kappa_\sigma^{\alpha-1}I$ for $\alpha \in (1, 2]$.
    By Eq.~\eqref{eq:hat-F-alpha-by-tr}, we have $\kappa_\sigma^{\alpha-1} \leq \widehat{F}_\alpha\rbra{\rho, \sigma} \leq \kappa_\rho^{1-\alpha}$ for $\alpha \in (0, 1)$ and $\kappa_\rho^{1-\alpha} \leq \widehat{F}_\alpha\rbra{\rho, \sigma} \leq \kappa_\sigma^{\alpha-1}$ for $\alpha \in (1, 2]$.

    For $\alpha \in (0, 1)$, to estimate $\widehat{D}_{\alpha}\rbra{\rho \Vert \sigma}$ within additive error $\epsilon$, we can estimate $\widehat{F}_\alpha\rbra{\rho, \sigma}$ to relative error $\epsilon$ (i.e., within additive error $\kappa_\sigma^{\alpha-1}\epsilon$). 
    By Theorem~\ref{thm:hat-F-alpha}, this can be done by using using 
    $\widetilde O\rbra{\kappa_\rho\kappa_\sigma^{2-\alpha}/\epsilon \cdot \min\cbra{\kappa_\rho, \kappa_\sigma}^{\min\cbra{1+\alpha, 2-\alpha}}}$
    queries to $\mathcal{O}_\rho$ and $\mathcal{O}_\sigma$.

    For $\alpha \in (1, 2]$, to estimate $\widehat{D}_{\alpha}\rbra{\rho \Vert \sigma}$ within additive error $\epsilon$, we can estimate $\widehat{F}_\alpha\rbra{\rho, \sigma}$ to relative error $\epsilon$ (i.e., within additive error $\kappa_\rho^{1-\alpha}\epsilon$). 
    By Theorem~\ref{thm:hat-F-alpha}, this can be done by using using 
    $\widetilde O\rbra{\kappa_\rho^\alpha\kappa_\sigma/\epsilon \cdot \min\cbra{\kappa_\rho\kappa_\sigma^{\alpha-1}, \kappa_\rho^{\alpha-1}\kappa_\sigma, \kappa_{\rho}^{1+\alpha}, \kappa_\sigma^{1+\alpha}}}$
    queries to $\mathcal{O}_\rho$ and $\mathcal{O}_\sigma$.
\end{proof}
    
\section{Proof of Lemmas~\ref{lemma:fidelity-lower-bound} and \ref{lemma:optimal-geo-fidelity}} \label{app:fidelity-lower-bound}

To prove the lower bound, we need the quantum query lower bound for distinguishing probability distributions given in~\cite{Bel19}.

\begin{lemma} [{\cite[Theorem 4]{Bel19}}] \label{lemma:q-dis-prob-distri}
    Let $p, q \colon \cbra{1, 2, \dots, n} \to \sbra{0, 1}$ be two probability distributions on a sample space of size $n$. 
    Let 
    \begin{align}
        U_p \ket{0} & = \sum_{j=1}^n \sqrt{p_j} \ket{j} \ket{\varphi_j}, \\
        U_q \ket{0} & = \sum_{j=1}^n \sqrt{q_j} \ket{j} \ket{\psi_j},
    \end{align}
    where $\cbra{\ket{\varphi_j}}_{j=1}^n$ and $\cbra{\ket{\psi_j}}_{j=1}^n$ are orthonormal bases. 
    Then, given an unknown unitary operator~$U$, any quantum query algorithm that determines whether $U = U_p$ or $U = U_q$ with probability at least $2/3$, promised that one or the other holds, has query complexity $\Omega\rbra{1/d_{\textup{H}}\rbra{p, q}}$, where
    \begin{equation}
        d_{\textup{H}}\rbra{p, q} \coloneqq \sqrt{\frac 1 2 \sum_{j=1}^n \rbra*{\sqrt{p_j} - \sqrt{q_j}}^2}
    \end{equation}
    is the Hellinger distance. 
\end{lemma}

Lemma~\ref{lemma:q-dis-prob-distri} was also used to prove quantum query lower bounds in \cite[Section 4.2]{GHS21}, \cite[Theorem 13]{LWL24}, and \cite[Section V]{Wan24}.

\begin{proof} [Proof of Lemma~\ref{lemma:fidelity-lower-bound}]
Let $\epsilon \in (0, 1/4)$.
Consider the discrimination of the two probability distributions $p, q \colon \cbra{0, 1} \to \sbra{0, 1}$ on a sample space of size two such that for each $j \in \cbra{0, 1}$,
\begin{align}
    p_j & = \frac{1 + \rbra{-1}^j \epsilon}{2}, \\
    q_j & = \frac{1 + \rbra{-1}^j 2\epsilon}{2}.
\end{align}
It can be verified that their Hellinger distance is upper bounded by
\begin{equation}
    d_\text{H} \rbra{p, q} = \sqrt{1 - \frac{\sqrt{\rbra{1+\epsilon}\rbra{1+2\epsilon}} + \sqrt{\rbra{1-\epsilon}\rbra{1-2\epsilon}}}{2}} \leq \epsilon.
\end{equation}
Suppose that two unitary operators $U_p$ and $U_q$ are given such that
\begin{align}
    U_p \ket{0} & = \sqrt{p_0} \ket{0} \ket{\varphi_0} + \sqrt{p_1} \ket{1} \ket{\varphi_1}, \\
    U_q \ket{0} & = \sqrt{q_0} \ket{0} \ket{\psi_0} + \sqrt{q_1} \ket{1} \ket{\psi_1},
\end{align}
where $\cbra{\ket{\varphi_0}, \ket{\varphi_1}}$ and $\cbra{\ket{\psi_0}, \ket{\psi_1}}$ are orthonormal bases. 

Let $\mathcal{A}\rbra{\mathcal{O}_\rho, \mathcal{O}_\sigma, \kappa_\rho, \kappa_\sigma, \epsilon}$ be any quantum query algorithm that estimates the fidelity $F\rbra{\rho, \sigma}$ between two mixed quantum states $\rho$ and $\sigma$ within additive error $\epsilon$, where $\mathcal{O}_\rho$ and $\mathcal{O}_\sigma$ prepare purifications of $\rho$ and $\sigma$, respectively, with $\rho \geq I/\kappa_\rho$ and $\sigma \geq I/\kappa_\sigma$. 
In the following, we use $\mathcal{A}\rbra{\mathcal{O}_\rho, \mathcal{O}_\sigma, \kappa_\rho, \kappa_\sigma, \epsilon}$ to distinguish $U_p$ and $U_q$. 
We first note that $U_p$ and $U_q$ can be understood as quantum unitary oracles that prepare purifications of the following two quantum states:
\begin{equation} \label{eq:hard-instance-fidelity-estimation}
     \rho = \frac{1 + \epsilon}{2} |0\rangle\!\langle 0| + \frac{1 - \epsilon}{2} |1\rangle\!\langle 1|, \quad \sigma = \frac{1 + 2\epsilon}{2} |0\rangle\!\langle 0| + \frac{1 - 2\epsilon}{2} |1\rangle\!\langle 1|.
\end{equation}
Then, one can set $\kappa_{\rho} = \kappa_\sigma = 4 = \Theta\rbra{1}$. 
Consider the quantum state
\begin{equation} \label{eq:hard-instance-eta}
    \eta = \frac 1 4 |0\rangle\!\langle 0| + \frac{3}{4} |1\rangle\!\langle 1|,
\end{equation}
and let $\mathcal{O}_\eta$ be a quantum oracle that prepares a purification of $\eta$.
We note that
\begin{align}  
F\rbra{\rho, \eta} & = \frac{\sqrt{1+\epsilon}+\sqrt{3\rbra{1-\epsilon}}}{\sqrt{8}}. \\
F\rbra{\sigma, \eta} & = \frac{\sqrt{1+2\epsilon}+\sqrt{3\rbra{1-2\epsilon}}}{\sqrt{8}}.
\end{align}
By simple calculation, we have 
\begin{equation} \label{eq:diff-fidelity}
    F\rbra{\rho, \eta} - F\rbra{\sigma, \eta} \geq \frac{\epsilon}{16}.
\end{equation}
Let $U$ be the unitary oracle to be tested, promised that either $U = U_p$ or $U = U_q$. 
For convenience, suppose that $U$ prepares a purification of $\varrho$, promised that either $\varrho = \rho$ or $\varrho = \sigma$. 
Our algorithm for determining which is the case is given as follows.
\begin{enumerate}
    \item Apply $\mathcal{A}\rbra{U, \mathcal{O}_\eta, 4, 4, \epsilon/64}$ to obtain an $\epsilon/64$-estimate $\tilde x$ of $F\rbra{\varrho, \eta}$. 
    \item If $\abs{\tilde x - F\rbra{\rho, \eta}} \leq \epsilon/32$, then return that $U = U_p$; otherwise, return that $U = U_q$.
\end{enumerate}
It can be verified that the above algorithm determines whether $U = U_p$ or $U = U_q$ with high probability, where the correctness is mainly based on Eq.~\eqref{eq:diff-fidelity}.

On the other hand, by Lemma~\ref{lemma:q-dis-prob-distri}, any quantum query algorithm that distinguishes $U_p$ and $U_q$ has query complexity $\Omega\rbra{1/d_\text{H} \rbra{p, q}} = \Omega\rbra{1/\epsilon}$.
Therefore, the algorithm $\mathcal{A}\rbra{U, \mathcal{O}_\eta, 4, 4, \epsilon/64}$ should use at least $\Omega\rbra{1/\epsilon}$ queries to $U$, which completes the proof. 
\end{proof}

Using the same hard instance, we can prove Lemma~\ref{lemma:optimal-geo-fidelity}.

\begin{proof} [Proof of Lemma~\ref{lemma:optimal-geo-fidelity}]
Note that under the choice of $\rho, \sigma, \eta$ the same as the proof of Lemma~\ref{lemma:fidelity-lower-bound}, we still have
\begin{equation}
    \widehat{F}_{1/2} \rbra{\rho, \eta} - \widehat{F}_{1/2} \rbra{\sigma, \eta} \geq \frac{\epsilon}{16},
\end{equation}
which is similar to Eq.~\eqref{eq:diff-fidelity}. 

Such an observation can be generalized to the general case when $0 < \alpha < 1$, which, however, becomes a bit more complicated. 
We first note that
\begin{align}  
\widehat{F}_\alpha\rbra{\rho, \eta} & = \rbra*{\frac{1}{4}}^{1-\alpha} \rbra*{\frac{1+\epsilon}{2}}^{\alpha} + \rbra*{\frac{3}{4}}^{1-\alpha} \rbra*{\frac{1-\epsilon}{2}}^{\alpha}, \\
\widehat{F}_\alpha\rbra{\sigma, \eta} & = \rbra*{\frac{1}{4}}^{1-\alpha} \rbra*{\frac{1+2\epsilon}{2}}^{\alpha} + \rbra*{\frac{3}{4}}^{1-\alpha} \rbra*{\frac{1-2\epsilon}{2}}^{\alpha}.
\end{align}
To make the construction in the proof of Lemma~\ref{lemma:fidelity-lower-bound} applicable to $\widehat{F}_\alpha\rbra{\cdot, \cdot}$ for $0 < \alpha < 1$, we only have to show that there is a constant $c > 0$ and $\epsilon_0 > 0$ (which  depends only on $\alpha$) such that for all $0 < \epsilon < \epsilon_0$, it holds that
\begin{equation}
    \widehat{F}_\alpha\rbra{\rho, \eta} - \widehat{F}_\alpha\rbra{\sigma, \eta} \geq c \epsilon. 
\end{equation}
To complete the proof, we show that this is achievable by noting that 
\begin{equation}
    \lim_{\epsilon \to 0} \frac{\widehat{F}_\alpha\rbra{\rho, \eta} - \widehat{F}_\alpha\rbra{\sigma, \eta}}{\epsilon} = \rbra*{\frac{1}{2}}^{\alpha} \sbra*{\rbra*{\frac{3}{4}}^{1-\alpha} - \rbra*{\frac{1}{4}}^{1-\alpha}} \alpha > 0.
\end{equation}
\end{proof}

\section{Sample complexity for fidelity estimation} \label{app:fidelity-estimation-sample}

In this section, we show how to extend our quantum query algorithm in Theorem~\ref{thm:fidelity-estimation} to a quantum algorithm that only uses samples of quantum states as input. 
To this end, we need the technique of density matrix exponentiation \cite{lloyd2014quantum,KLL+17}. 
Here, we use the extension given in \cite{WZ24} that is easy to use for quantum query algorithms.

\begin{lemma}[Samplizer, {\cite[Theorem 1.3]{WZ24}}] \label{lemma:samplizer}
    Let $\mathcal{A}^{U}$ be a quantum query algorithm that uses $Q$ queries to the unitary oracle $U$. 
    Then, for any $\delta \in \rbra{0, 1}$ and quantum state $\rho$, we can implement a quantum channel $\mathsf{Samplize}_{\delta} \ave{\mathcal{A}^{U}}\sbra{\rho}$ by using $\widetilde{O}\rbra{Q^2/\delta}$ samples of $\rho$, such that there is a unitary operator $U_{\rho}$ that is a $\rbra{1, a, 0}$-block-encoding of $\rho/2$ for some $a > 0$ satisfying
    \begin{equation}
    \Abs*{\mathcal{A}^{U_\rho} - \mathsf{Samplize}_{\delta} \ave{\mathcal{A}^{U}}\sbra{\rho}}_\diamond \leq \delta.
    \end{equation}
\end{lemma}

Now let $\mathcal{A}^{U_A, U_C}$ be the quantum query algorithm in Lemma~\ref{lem:simplericcati}, where $U_A$ and $U_C$ are supposed to be $\rbra{1, a, 0}$-block-encodings of matrices $A$ and $C$, respectively. 
Assume that if it is known that $A \geq I/\kappa_A$ and $C \geq I/\kappa_C$, then $\mathcal{A}^{U_A, U_C}$ uses $Q_A = \widetilde{O}\rbra*{\kappa_A^2\kappa_C\log^3\rbra*{1/\delta}}$ queries to $U_A$ and $Q_C = \widetilde{O}\rbra*{\kappa_A\kappa_C\log^2\rbra*{1/\delta}}$ queries to $U_C$.
Here, $\mathcal{A}^{U_A, U_C}$ is a $\rbra{2\kappa_A, b, \delta}$-block-encoding of $A^{-1}\# C$, where $b = 5a+11$. 

Let $\delta_A, \delta_C \in \rbra{0, 1}$ be parameters to be determined. 
By Lemma~\ref{lemma:samplizer}, we can implement a quantum query algorithm (using queries to $U_C$)
\begin{equation}
    \mathcal{B}^{U_C} \coloneqq \mathsf{Samplize}_{\delta_A} \ave{\mathcal{A}^{\xboxed{U_A}, U_C}}\sbra{\sigma}
\end{equation}
that uses $S_A = \widetilde{O}\rbra{Q_A^2/\delta_A}$ samples of $\sigma$ such that there is a unitary operator $U_\sigma$ that is a $\rbra{1, a_\sigma, 0}$-block-encoding of $\sigma/2$ for some $a_\sigma > 0$ satisfying
\begin{equation} \label{eq:samplizer-ab}
    \Abs*{\mathcal{A}^{U_\sigma, U_C} - \mathcal{B}^{U_C}}_\diamond \leq \delta_A,
\end{equation}
Here, the boxed oracle $\xboxed{U}$ denotes the oracle to be samplized.

Again, by Lemma~\ref{lemma:samplizer}, we can implement a quantum channel
\begin{equation}
    \mathcal{C} \coloneqq \mathsf{Samplize}_{\delta_C} \ave{\mathcal{B}^{U_C}}\sbra{\rho}
\end{equation}
that uses additional $S_C = \widetilde{O}\rbra{Q_C^2/\delta_C}$ samples of $\rho$ such that there is a unitary operator $U_\rho$ that is a $\rbra{1, a_\rho, 0}$-block-encoding of $\rho/2$ for some $a_\rho > 0$ satisfying
\begin{equation} \label{eq:samplizer-bc}
    \Abs*{\mathcal{B}^{U_\rho} - \mathcal{C}}_\diamond \leq \delta_C.
\end{equation}
By Eq.~\eqref{eq:samplizer-ab} and Eq.~\eqref{eq:samplizer-bc}, we have
\begin{equation} \label{eq:precision-ac}
    \Abs*{\mathcal{A}^{U_\sigma, U_\rho} - \mathcal{C}}_\diamond \leq \delta_A + \delta_C.
\end{equation}
By taking $A \coloneqq \sigma/2$ and $C \coloneqq \rho/2$, we know that $\mathcal{A}^{U_\sigma, U_\rho}$ is a $\rbra{4 \kappa_\sigma, b, \delta}$-block-encoding of $\rbra{\sigma/2}^{-1} \# \rbra{\rho/2} = \sigma^{-1} \# \rho$. 
Then, following the analysis of Theorem~\ref{thm:fidelity-estimation}, by the Hadamard test (given in Lemma~\ref{lemma:hadamard-test}), there is a quantum circuit $C$ that outputs $0$ with probability 
\begin{equation} \label{eq:sample-prob}
    p = \frac{1}{2}\rbra*{ 1+\Re{\Tr\rbra{\bra{0}_b \mathcal{A}^{U_\sigma, U_\rho} \ket{0}_b \sigma}} },
\end{equation}
using one query to $\mathcal{A}^{U_\sigma, U_\rho}$ and one sample of~$\sigma$.
Note that
\begin{equation} \label{eq:sample-error}
    \abs*{ 4\kappa_\sigma \Re{\Tr\rbra{\bra{0}_b \mathcal{A}^{U_\sigma, U_\rho} \ket{0}_b \sigma}} - \Tr\rbra{\rbra{\sigma^{-1} \# \rho}\sigma} } \leq \Theta\rbra{\delta}.
\end{equation}
If we construct another quantum circuit $C'$ by replacing $\mathcal{A}^{U_\sigma, U_\rho}$ by $\mathcal{C}$ in the implementation of $C$,
then $C'$ outputs $0$ with probability $p'$ such that
\begin{equation} \label{eq:sample-prob-prime}
    \abs*{p - p'} \leq \Theta\rbra{\delta_A + \delta_C}
\end{equation}
because of Eq.~\eqref{eq:precision-ac}. 
By $O\rbra{1/\epsilon_H^2}$ repetitions of $C'$, we can obtain an $\epsilon_H$-estimate $\tilde p$ of $p'$, i.e., 
\begin{equation} \label{eq:sample-hadamard-test}
    \abs{\tilde p - p'} \leq \epsilon_H.
\end{equation}
By Eqs.~\eqref{eq:sample-prob}, \eqref{eq:sample-error}, \eqref{eq:sample-prob-prime}, and \eqref{eq:sample-hadamard-test}, we have 
\begin{equation}
    \abs*{ 4\kappa_\sigma \rbra{2\tilde p - 1} - F\rbra{\rho, \sigma}} \leq \Theta\rbra{\delta + \kappa_\sigma\rbra{\delta_A + \delta_C + \epsilon_H}}.
\end{equation}
By taking $\delta = \Theta\rbra{\epsilon}$ and $\delta_A = \delta_C = \epsilon_H = \Theta\rbra{\epsilon/\kappa_\sigma}$, we can estimate $F\rbra{\rho, \sigma}$ to within additive error $\epsilon$.
The number of samples of $\sigma$ used is 
\begin{equation}
    O\rbra*{\frac{1}{\epsilon_H^2}} \cdot \rbra{S_A+1} = \widetilde{O}\rbra*{\frac{\kappa_\sigma^7 \kappa_\rho^2}{\epsilon^3}},
\end{equation}
and the number of samples of $\rho$ used is
\begin{equation}
    O\rbra*{\frac{1}{\epsilon_H^2}} \cdot S_C = \widetilde{O}\rbra*{\frac{\kappa_\sigma^5 \kappa_\rho^2}{\epsilon^3}}.
\end{equation}
Therefore, the total number of samples of $\rho$ and $\sigma$ used is $\widetilde{O}\rbra{\kappa_\sigma^7 \kappa_\rho^2/\epsilon^3}$.
By considering the symmetric case, the sample complexity for fidelity estimation is $\widetilde{O}\rbra{\min\cbra{\kappa_\rho^5, \kappa_\sigma^5} \cdot \kappa_\rho^2 \kappa_\sigma^2/\epsilon^3}$. 

\begin{lemma} [Sample complexity for fidelity estimation]
    Suppose that two quantum states $\rho$ and $\sigma$ satisfy $\rho \geq I/\kappa_\rho$ and $\sigma \geq I/\kappa_\sigma$ for some known parameters $\kappa_\rho, \kappa_\sigma \geq 1$. 
    Then, we can estimate their fidelity by using $\widetilde{O}\rbra{\min\cbra{\kappa_\rho^5, \kappa_\sigma^5} \cdot \kappa_\rho^2 \kappa_\sigma^2/\epsilon^3}$ samples of them.
\end{lemma}

We can also derive a sample lower bound for fidelity estimation.

\begin{lemma} [Sample lower bound for fidelity estimation]
    Suppose that two quantum states $\rho$ and $\sigma$ satisfy $\rho \geq I/\kappa_\rho$ and $\sigma \geq I/\kappa_\sigma$ for some known parameters $\kappa_\rho, \kappa_\sigma \geq 1$. 
    Then, every quantum algorithm that estimates $F\rbra{\rho, \sigma}$ within additive error~$\epsilon$ requires sample complexity $\Omega\rbra{1/\epsilon^2}$ even if $\kappa_\rho = \kappa_\sigma = \Theta\rbra{1}$.
\end{lemma}

\begin{proof}
    Using the same instance in the proof of Lemma~\ref{lemma:fidelity-lower-bound}, we can distinguish the following two quantum states $\rho$ and $\sigma$ defined by Eq.~\eqref{eq:hard-instance-fidelity-estimation} by estimating the fidelity $F\rbra{\rho, \eta}$ and $F\rbra{\sigma, \eta}$, where $\eta$ is defined by Eq.~\eqref{eq:hard-instance-eta}.
    On the other hand, distinguishing $\rho$ and $\sigma$ requires sample complexity $\Omega\rbra{1/\epsilon^2}$ by the Helstrom-Holevo bound \cite{Hel67,Hol73}. 
\end{proof}

\section{Proof of Lemma~\ref{lemma:BQP-containment}} \label{app:lemmaMGM}
We consider how to solve $\textup{MGM}\rbra{\kappa_A, \kappa_C}$ in quantum time $\poly\rbra{n}$ for $\kappa_A = \poly\rbra{n}$ and $\kappa_C = \poly\rbra{n}$.
    It is straightforward that the given uniform classical circuit $\mathcal{C}_n$ implies the quantum implementations of the sparse oracles of $A$ and $C$, which are (uniform) quantum circuits of size $\poly\rbra{n}$.
    By Lemma~\ref{lemma:sparse-to-block}, we can implement $U_A$ and $U_C$ such that $U_A$ and $U_C$ are $\rbra{O\rbra{1}, \poly\rbra{n}, \epsilon}$-block-encodings of $A$ and $C$, respectively, using $O\rbra{1}$ queries to the sparse oracles of $A$ and $C$ and $O\rbra{\poly\rbra{n}+\polylog\rbra{1/\epsilon}}$ one- and two-qubit quantum gates.
    Here, we choose $\epsilon = 1/\exp\rbra{n}$ for convenience, and  
     we assume that $U_A$ and $U_C$ are $\rbra{O\rbra{1}, \poly\rbra{n}, 0}$-block-encodings of $\hat A$ and $\hat C$ such that $\Abs{\hat A - A} \leq \epsilon$ and $\Abs{\hat C - C} \leq \epsilon$.
    
    Let $\delta = O\rbra{\kappa_A^{-5}\kappa_C^{-5}} = 1/\poly\rbra{n}$. 
    By Lemma~\ref{lem:simplericcati}, we can implement an $\rbra{O\rbra{1}, \poly\rbra{n}, \delta}$-block-encoding $U_{Y}$ of $\kappa_A^{-1} \hat Y$ using $\widetilde O\rbra{\kappa_A \kappa_C \log^2\rbra{1/\delta}} = \poly\rbra{n}$ queries to $U_C$, $\widetilde O\rbra{\kappa_A^2 \kappa_C \log^3\rbra{1/\delta}} = \poly\rbra{n}$ queries to $U_A$, and $\poly\rbra{n} \cdot \poly\rbra{\kappa_A, \kappa_C, \log\rbra{1/\delta}} = \poly\rbra{n}$ one- and two-qubit quantum gates, where $\hat Y = \hat A^{-1} \# \hat C$. 
    Moreover, the quantum circuit description of $U_Y$ can be computed in classical time $\poly\rbra{\kappa_A, \kappa_C, \log\rbra{1/\delta}} = \poly\rbra{n}$. 
    By Lemma~\ref{lmm:product of block-encoding}, we can implement an $\rbra{O\rbra{1}, \poly\rbra{n}, O\rbra{\delta}}$-block-encoding $U_{Y^2}$ of $\kappa_A^{-2} \hat Y^2$ using $O\rbra{1}$ queries to $U_Y$. 
    Here, we assume that $U_{Y^2}$ is an $\rbra{O\rbra{1}, \poly\rbra{n}, 0}$-block-encoding of
    \begin{equation}
    Z = \bra{0}^{\otimes a} U_{Y^2} \ket{0}^{\otimes a}    
    \end{equation}
    such that $\Abs{Z - \kappa_A^{-2} \hat Y^2} \leq O\rbra{\delta}$, and it can be easily shown that $\Abs{\hat Y - Y} \leq O\rbra{\epsilon}$.
    Then, 
    \begin{equation}
        \Abs{Z - \kappa_A^{-2} Y^2} \leq O\rbra{\delta + \kappa_A^{-2}\epsilon},
    \end{equation}
    \begin{equation}
        \rbra*{\kappa_A^{-3} \kappa_C^{-1} - O\rbra{\epsilon} \kappa_A^{-2} - O\rbra{\delta}} I \leq Z \leq I.
    \end{equation}
    The latter can be seen by noting that $\kappa_A^{-1}\kappa_C^{-1} I \leq Y^2 \leq \kappa_A^2 I$. 

    Now we prepare the quantum state $\ket{\psi} = U_{Y^2} \ket{0} = \ket{0}^{\otimes a} \otimes Z \ket{0} + \ket{\perp}$ where $\ket{\perp}$ is orthogonal to $\ket{0}^{\otimes a} \otimes \ket{\varphi}$ for any $\ket{\varphi}$.
    By measuring the first $a$ qubits of $\ket{\psi}$, the outcome will be $0^a$ with probability
    \begin{equation}
        \Abs{Z\ket{0}}^2 \geq \Theta\rbra{\kappa_A^{-6}\kappa_C^{-2}}  = \frac{1}{\poly\rbra{n}}
    \end{equation}
    and $\ket{\psi}$ will become the state $\ket{u_Z} \coloneqq Z\ket{0}/\Abs{Z\ket{0}}$. 
    Let $\ket{u_Y} \coloneqq Y^2\ket{0}/\Abs{Y^2\ket{0}}$. 
    We have
    \begin{align}
        \Abs{ \ket{u_Z} - \ket{u_Y} } 
        & \leq \Abs*{ \frac{Z\ket{0}}{\Abs{Z\ket{0}}} - \frac{\kappa_A^{-2}Y^2\ket{0}}{\Abs{Z\ket{0}}} } + \Abs*{ \frac{\kappa_A^{-2}Y^2\ket{0}}{\Abs{Z\ket{0}}} - \frac{\kappa_A^{-2}Y^2\ket{0}}{\Abs{\kappa_A^{-2}Y^2\ket{0}}} } \\
        & \leq O\rbra*{\frac{\Abs{Z - \kappa_A^{-2}Y^2}}{\Abs{Z\ket{0}}}} \\
        & \leq O\rbra*{ \frac{\delta + \kappa_A^{-2}\epsilon}{\kappa_A^{-3}\kappa_C^{-1}} } = \frac{1}{\poly\rbra{n}}. \label{eq:diff-uz-uy}
    \end{align}
    Let $p_Z$ (resp.\ $p_Y$) be the probability that outcome $0$ will be obtained by measuring the first qubit of $\ket{u_Z}$ (resp.\ $\ket{u_Y}$). 
    Note that $p_Z = \bra{u_Z}M\ket{u_Z}$ and $p_Y = \bra{u_Y}M\ket{u_Y}$, where $M = |0\rangle\!\langle 0| \otimes I$ measures the first qubit of $\ket{u_Z}$ and $\ket{u_Y}$. 
    Then, $\abs{p_Z - p_Y} \leq 1/\poly\rbra{n}$ by Eq.~\eqref{eq:diff-uz-uy}. 
    
    Finally, we can estimate $p_Z$ to precision, say $0.1$, by repeating the above procedure for $O\rbra{1}$ times; in this way, we can determine whether $p_Y \geq 2/3$ or $p_Y \leq 1/3$ with high probability. 
    As all the procedures mentioned above take $\poly\rbra{n}$ time, we obtain a polynomial-time quantum algorithm for $\textup{MGM}\rbra{\kappa_A, \kappa_C}$ if $\kappa_A = \poly\rbra{n}$ and $\kappa_C = \poly\rbra{n}$. 
    Therefore, we conclude that $\textup{MGM}\rbra{\poly\rbra{n}, \poly\rbra{n}}$ is in $\mathsf{BQP}$. 

\end{document}